\documentclass{jpc}

\usepackage[english]{babel}
\usepackage[labelfont=bf]{caption}
\usepackage{subcaption}
\usepackage{amsmath}
\usepackage{amsfonts}
\usepackage{amssymb}
\usepackage{cite}
\usepackage{graphicx}
\usepackage{thm-restate}
\usepackage[usenames,dvipsnames,table]{xcolor}
\usepackage{algpseudocode}
\usepackage{algorithm}
\usepackage{balance}
\usepackage{xspace}
\usepackage[capitalize,noabbrev]{cleveref}
\usepackage{outlines}
\usepackage{outlines}
\usepackage{multicol}
\usepackage{multirow}
\usepackage{bm}
\usepackage{array}
\usepackage{tabularx}
\usepackage{dcolumn}
	\newcolumntype{d}[1]{D{.}{.}{#1|}}
\usepackage{booktabs}
\usepackage[dvipsnames]{xcolor}
\usepackage{soul}
\usepackage[group-separator={,}]{siunitx}
\usepackage{mathbbol}

\DeclareSymbolFontAlphabet{\mathbb}{AMSb}
\DeclareSymbolFontAlphabet{\mathbbl}{bbold}

\DeclareMathOperator*{\argmin}{arg\,min}

\begin{document}

\newtheorem{remark}{Remark}
\newtheorem{theorem}{Theorem}
\newtheorem{lemma}{Lemma}
\newtheorem{definition}{Definition}
\newtheorem{problem}{Problem}
\newtheorem{proposition}{Proposition}
\newtheorem{corollary}{Corollary}
\newtheorem{example}{Example}
\newtheorem{conjecture}{Conjecture}
\newtheorem{operator}{Operator}

\def\select{\texttt{SELECT}\xspace}
\def\measure{\texttt{MEASURE}\xspace}
\def\reconstruct{\texttt{RECONSTRUCT}\xspace}
\def\Error{\text{TSE}}
\def\map{{\mbox{\sf MAP}}}
\def\plus{{+}}
\def\opt{{\mbox{\sf OPT}}}
\def\metaopt{\ensuremath{{\mbox{\sf OPT}_{\sf HDMM}}}}
\def\WW{\mathbb{W}}
\def\VV{\mathbb{V}}
\def\vv{\mathbbl{v}}
\def\Wlog{\mathcal{W}}
\def\AA{\mathbb{A}}
\def\BB{\mathbb{B}}
\def\CC{\mathbb{H}}
\def\GG{\mathbb{G}}
\def\QQ{\mathbb{Q}}
\def\MM{\mathbb{M}}
\def\optgp{\mbox{\sf OPT}_{0}}
\def\optk{\mbox{\sf OPT}_{\otimes}}
\def\optkk{\mbox{\sf OPT}_{+}}
\def\optm{\mbox{\sf OPT}_{\mathsf{M}}}

\def\Imp{\mbox{\sf ImpVec}}
\def\mult{\mbox{\sf Multiply}}
\def\invert{\mbox{\sf LstSqr}}

\newcommand{\newcontent}[1]{{\color{MidnightBlue} #1}}

\newcommand{\tighttimes}{\!\times\!}

\newcommand{\algoname}[1]{{\textsc{#1}}\xspace}
\newcommand{\sys}{\algoname{HDMM}}
\newcommand{\sysFull}{High-Dimensional Matrix Mechanism\xspace}
\newcommand{\LMW}{\algoname{LM}}
\newcommand{\Identity}{\algoname{Identity}}
\newcommand{\Workload}{\ensuremath{MM(W)}}
\newcommand{\PrivBayes}{\algoname{PrivBayes}}
\newcommand{\matrixMech}{\algoname{MM}}
\newcommand{\HB}{\algoname{HB}}
\newcommand{\Privelet}{\algoname{Privelet}}
\newcommand{\Quadtree}{\algoname{Quadtree}}
\newcommand{\GreedyH}{\algoname{GreedyH}}
\newcommand{\DataCube}{\algoname{DataCube}}
\newcommand{\DAWA}{\algoname{DAWA}}
\newcommand{\LRM}{\algoname{LRM}}

\newcommand{\datasetname}[1]{{\emph{#1}}\xspace}
\newcommand\Census{\datasetname{CPH}}
\newcommand\Person{\datasetname{Person}}
\newcommand\Adult{\datasetname{Adult}}
\newcommand\CPS{\datasetname{CPS}}
\newcommand\Patent{\datasetname{Patent}}
\newcommand\Taxi{\datasetname{Taxi}}

\newcommand{\workloadname}[1]{{\emph{#1}}\xspace}
\newcommand\FixedWidthRange{\workloadname{Width 32 Range}}
\newcommand\PrefixOneD{\workloadname{Prefix 1D}}
\newcommand\PrefixTwoD{\workloadname{Prefix 2D}}
\newcommand\PrefixThreeD{\workloadname{Prefix 3D}}
\newcommand\PermutedRange{\workloadname{Permuted Range}}
\newcommand\PrefixIdentity{\workloadname{Prefix Identity}}
\newcommand\SFOne{\workloadname{SF1}}
\newcommand\SFOnePlus{\workloadname{SF1+}}
\newcommand\AllMarginals{\workloadname{All Marginals}}
\newcommand\TwoWayMarginals{\workloadname{2-way Marginals}}
\newcommand\ThreeWayMarginals{\workloadname{3-way Marginals}}
\newcommand\AllRangeMarginals{\workloadname{All Range-Marginals}}
\newcommand\TwoWayRangeMarginals{\workloadname{2-way Range-Marginals}}
\newcommand\ThreeWayRange{\workloadname{All 3-way Ranges}}

\newcommand{\errRatio}{\ensuremath{Ratio(W, \mathcal{A}_{other})}\xspace}

\newcommand{\NA}{-}
\newcommand{\NS}{*}


\newcommand{\op}[1]{{\sc #1}\xspace}
\newcommand{\stitle}[1]{\vspace{-2.5mm}\paragraph*{#1}}
\newcommand{\eat}[1]{}
\newcommand{\eatrev}[1]{}
\newcommand{\mybullet}{\vspace{1mm}\noindent$\bullet$~~}

\newcommand{\todo}[1]{[[\emph{\color{teal}TODO: #1}]]}
\newcommand{\gm}[1]{[[\emph{\color{red}GM: #1}]]}
\newcommand{\ry}[1]{[[\emph{\color{blue}RM: #1}]]}
\newcommand{\mh}[1]{[[\emph{\color{red}MH: #1}]]}
\newcommand{\am}[1]{[[\emph{\color{magenta}AM: #1}]]}

\newcommand{\gmref}[2]{{\color{cyan} #1}~[[\emph{{\color{red} GM: #2}}]]}
\newcommand{\ryref}[2]{{\color{red} #1}~[[\emph{{\color{blue} RM: #2}}]]}
\newcommand{\mhref}[2]{{\color{cyan} #1}~[[\emph{{\color{red} MH: #2}}]]}
\newcommand{\amref}[2]{{\color{cyan} #1}~[[\emph{{\color{magenta} AM: #2}}]]}

\renewcommand*\Call[2]{\textproc{#1}(#2)}



\def\Lap{\mbox{Lap}}
\def\Gaus{\mbox{Gaus}}
\def\sens{\Delta}
\def\LM{\mathcal{L}}
\def\GM{\mathcal{G}}

\newcommand{\Lk}[1]{\left\Vert #1  \right\Vert_\algG}
\newcommand{\Lone}[1]{\left\Vert #1  \right\Vert_\LM}
\newcommand{\Ltwo}[1]{\left\Vert #1  \right\Vert_\GM}
\newcommand{\norm}[1]{\left\lVert#1\right\rVert}
\newcommand{\set}[1]{\{#1\}}   

\newcommand{\vect}[1]{\bm{#1}}
\newcommand{\matr}[1]{\bm{#1}}
\def\btheta{\vect{\theta}}
\def\blambda{\vect{\lambda}}
\def\bkappa{\vect{\kappa}}
\def\bTheta{\matr{\Theta}}
\def\z{\vect{z}}
\def\c{\vect{c}}
\def\W{{\matr{W}}}
\def\V{\matr{V}}
\def\A{\matr{A}}
\def\Q{\matr{Q}}
\def\B{\matr{B}}
\def\C{\matr{C}}
\def\D{\matr{D}}
\def\L{\matr{L}}
\def\I{\matr{I}}
\def\T{\matr{T}}
\def\R{\matr{R}}
\def\P{\matr{P}}
\def\S{\matr{S}}
\def\M{\matr{M}}
\def\Y{\matr{Y}}
\def\X{\matr{X}}
\def\x{\vect{x}}
\def\q{\vect{q}}
\def\a{\vect{a}}
\def\u{\vect{u}}
\def\v{\vect{v}}
\def\w{\vect{w}}
\def\e{\vect{e}}

\def\xhat{\bar{\x}}
\def\y{\vect{y}}

\newcommand{\Wnt}{\ensuremath{\mathcal{W}_{\mbox{\tiny SF1}}}}
\newcommand{\Wst}{\ensuremath{\mathcal{W}_{\mbox{\tiny SF1+}}}}

 \def\db{I}  
 \def\nbrs{nbrs}

 \def\algG{\mathcal{K}}  

\title[Title]{HDMM: Optimizing error of high-dimensional statistical queries under differential privacy
}
\author{Ryan McKenna}	
\address{College of Information \& Computer Sciences, The University of Massachusetts, Amherst, MA 10002}
\email{rmckenna@cs.umass.edu}  

\author{Gerome Miklau}	
\address{College of Information \& Computer Sciences, The University of Massachusetts, Amherst, MA 10002}
\email{miklau@cs.umass.edu}  

\author{Michael Hay}	
\address{Department of Computer Science, Colgate University, Hamilton, NY 13346}	
\email{mhay@colgate.edu}  

\author{Ashwin Machanavajjhala}	
\address{Department of Computer Science, Duke University, Durham, NC, 27708}	
\email{ashwin@cs.duke.edu}  

\maketitle

\pagestyle{plain}

\begin{abstract}

In this work we describe the High-Dimensional Matrix Mechanism (\sys), a differentially private algorithm for answering a workload of predicate counting queries.
\sys represents query workloads using a compact implicit matrix representation and exploits this representation to efficiently optimize over (a subset of) the space of differentially private algorithms for one that is unbiased and answers the input query workload with low expected error.  $\sys$ can be deployed for both $\epsilon$-differential privacy (with Laplace noise) and $(\epsilon, \delta)$-differential privacy (with Gaussian noise), although the core techniques are slightly different for each.  We demonstrate empirically that \sys can efficiently answer queries with lower expected error than state-of-the-art techniques, and in some cases, it nearly matches existing lower bounds for the particular class of mechanisms we consider.
\end{abstract}

\section{Introduction}  \label{sec:introduction}

Institutions like the U.S. Census Bureau and Medicare regularly release summary statistics about individuals, including population statistics cross-tabulated by demographic attributes \cite{census-sf1,onthemap} and tables reporting on hospital discharges organized by medical condition and patient characteristics \cite{hcupnet}.
%
%
These data have the potential to reveal sensitive information, especially through joint analysis of multiple releases~\cite{onthemap:icde08,sigmod:haney17,vaidya2013hcupnet}.  Differential privacy~\cite{dwork2006calibrating,Dwork14Algorithmic} offers a framework for releasing statistical summaries of sensitive datasets, while providing formal and quantifiable privacy to the contributing individuals.  

We consider the problem of \emph{batch query answering} under differential privacy.  That is, our goal is to release answers to a given query \emph{workload}, consisting of a set of \emph{predicate counting queries}, while satisfying differential privacy. A predicate counting query computes the number of individuals in the dataset who satisfy an arbitrary predicate $\phi$ (e.g., how many individuals have Income $\geq \$50,\!000$).  Workloads of predicate counting queries are quite versatile as they are capable of expressing histograms, multi-dimensional range queries, group-by queries, data cubes, marginals, and arbitrary combinations thereof.
Answering a batch of predicate counting queries has been widely studied by the research community. Past results have established theoretical lower bounds~\cite{bhaskara2012unconditional, hardt2010geometry,nikolov2013geometry,li13optimal} as well as a wealth of practical  algorithms~\cite{zhang16privtree,yuan2016convex,li2015matrix,zhangtowards,xiao2014dpcube,qardaji2014priview,li2014data,Yaroslavtsev13Accurate,xu2013differential,qardaji2013understanding,qardaji2013differentially,yuan2012low,xu12histogram,li2012adaptive,cormode2012differentially,Acs2012compression,xiao2011differential,ding2011differentially,li2010optimizing,hay2010boosting,barak2007privacy,qardaji2014priview,Zhang2014}.
\eat{
\begin{table*}
	\caption{\label{fig:alg_mm} The key steps of the Matrix Mechanism (\matrixMech), an instance of the select-measure-reconstruct paradigm~\cite{li2010optimizing}. (Below is the Laplace version of the Matrix Mechanism; the Gaussian version differs only in noise generation and the matrix norm applied to $\A$.)} 
\centering
		\begin{tabular}{rlcl}
			\multicolumn{4}{l}{\textbf{Input:} workload $\W$, in matrix form}\\
			\multicolumn{4}{l}{\hspace{.45in} data $\x$, in vector form}\\
			\multicolumn{4}{l}{\hspace{.45in} privacy parameter $\epsilon$}\\
			\hline \hline
			& \\ 
 			{\sf SELECT} $\begin{cases}\end{cases}$ \hspace{1ex} \qquad  & $\A$ & = & $\opt_{MM}(\W)$ \\
			\multirow{2}{*}{{\sf MEASURE} $\begin{cases} \mathstrut \end{cases}$} &$\a$ & = & $\A\x$ \\
			& $\y$ & = & $\a + Lap(\norm{\A}_1 / \epsilon)$ \\
			\multirow{2}{*}{{\sf RECONSTRUCT} $\begin{cases} \mathstrut \end{cases}$} & $\xhat$ & = & $\A^+\y$ \\
			& $ans$ & = & $\W \xhat$
		\end{tabular}
\end{table*}}

One of the simplest mechanisms for answering a workload of queries is to add carefully calibrated Laplace or Gaussian noise directly to each of the workload query answers.  The noise magnitude is calibrated to a property of the workload known as its \emph{sensitivity}, which can be large for some workloads, resulting in a significant amount of noise. This method fails to adequately exploit structure in the workload and correlation amongst queries, and thus it often adds more noise than is strictly necessary to preserve differential privacy, resulting in suboptimal utility. 

A better approach generalizes the basic noise addition mechanism by first \textbf{selecting} a new set of \emph{strategy} queries, then \textbf{measuring} the strategy queries using a noise addition mechanism, and \textbf{reconstructing} answers to the workload queries from the noisy measurements of the strategy queries. Choosing an effective query answering strategy (different from the workload) can result in orders-of-magnitude lower error than the Laplace mechanism, with no cost to privacy.
Many mechanisms for workload answering fall within the select-measure-reconstruct paradigm~\cite{zhang16privtree,yuan2016convex,li2015matrix,zhangtowards,xiao2014dpcube,qardaji2014priview,li2014data,Yaroslavtsev13Accurate,xu2013differential,qardaji2013understanding,qardaji2013differentially,yuan2012low,xu12histogram,li2012adaptive,cormode2012differentially,Acs2012compression,xiao2011differential,ding2011differentially,li2010optimizing,hay2010boosting,li2015matrix}, differing primarily in the strategy selection step.
We can characterize strategy selection as a search problem over a space of strategies, distinguishing prior work in terms of key algorithmic design choices: the search space, the cost function, and the type of search algorithm (greedy, local, global, etc.).  These design choices impact the two key performance considerations: accuracy and scalability.

At one extreme are techniques that explore a narrow search space, making them efficient and scalable but not particularly accurate (in particular, their search space may include accurate strategies only for a limited class of workloads).  For example, \HB~\cite{qardaji2013understanding} considers strategies consisting of hierarchically structured interval queries.  It performs a simple search to find the branching factor of the hierarchical strategy that minimizes an error measure that assumes the workload consists of all range queries (regardless of the actual input workload).  It is efficient and can scale to higher dimensions, but it achieves competitive accuracy only when the workload consists of range queries and the data is low dimensional.

At the other extreme are techniques that search a large space, and adapt to the workload by finding a strategy within that space that offers low error on the workload, thereby making them capable of producing a more accurate strategy for the particular workload. 
However, this increased accuracy comes at the cost of high runtime and poor scalability.  This is exemplified by the Matrix Mechanism \cite{li2015matrix}. The Matrix Mechanism represents the workload and strategy as a matrix, and the data as a vector.  With this representation, the select, measure, and reconstruct steps can be completely defined in the language of linear algebra.  In addition, there is a simple formula for the expected error of any selected strategy matrix in terms of elementary matrix operations.  This enables the Matrix Mechanism to select the optimal strategy (i.e., the one that offers least expected error) by solving a numerical optimization problem.  

Using a matrix to represent a workload is appealing because the representation is expressive enough to capture an arbitrary collection of predicate counting queries, and it reveals any structure that may exist between the workload queries.  However, the size of the workload matrix is equal to the number of queries times the size of the domain, and is infeasible to represent large workloads defined over multi-dimensional domains as a matrix.  Moreover, solving the optimization problem underlying strategy selection is nontrivial and expensive. In short, there is no prior work that is accurate for a wide range of input workloads, and also capable of scaling to large multi-dimensional domains.

\eat{
\begin{table*}
	\caption{\label{fig:alg_hdmm} An overview of the {\em High Dimensional Matrix Mechanism} (\sys), which shares many of the steps of the Matrix Mechanism, but uses specialized representations of the workload and strategy, along with the  efficient operations they enable.}
		\begin{tabular}{lcll}
			\textbf{Input:} & \multicolumn{3}{l}{workload $\Wlog$, in logical form}\\
			& \multicolumn{3}{l}{data $\x$, in vector form}\\
			& \multicolumn{3}{l}{privacy parameter $\epsilon$}\\
			\hline \hline
			 $\WW$ & = & $\Imp(\Wlog)$ & \emph{// Compact vector representation}\\
			  $\AA$ & = & $\metaopt(\WW)$ & \emph{// Optimized strategy selection}\\
			 $\a$ & = & $\mult(\AA, \x)$ & \emph{// Strategy query answering} \\
			$\y$ & = & $\a + Lap(\norm{\AA}_1 / \epsilon)$ & \emph{// Noise addition}  \\
			$\xhat$ & = & $\invert(\AA, \y)$  & \emph{// Inference}\\
			$ans$ & = & $\mult(\WW, \xhat)$ & \emph{// Workload answering}
		\end{tabular}
\end{table*}}

\subsection*{Overview of approach and contributions} \label{sec:sub:overview}

This paper describes the \sysFull (\sys) which is a practical instantiation of the Matrix Mechanism (MM), capable of scaling to large multi-dimensional domains.  HDMM offers the flexibility and workload-adaptivity of the Matrix Mechanism, while offering the scalability of simpler mechanisms.  There are a number of innovations that distinguish HDMM from the Matrix Mechanism: 

\begin{itemize}
\item The Matrix Mechanism represents query workloads \emph{explicitly}, as fully materialized matrices, while \sys uses a compact \emph{implicit} matrix representation. This permits a lossless representation of queries that avoids a representation exponential in the number of attributes. The implicit representation consists of sub-workload matrices (usually one per attribute) which are used as terms in a Kronecker product.  Further, we allow the workload to be expressed as unions of such Kronecker terms.  This allows us to represent large multi-dimensional workloads efficiently while maintaining the key benefits that the explicit matrix representation offers.  
\item The numerical optimization problem at the heart of the Matrix Mechanism is practically infeasible, even for a single attribute with a domain of size 10. \sys introduces four optimization routines for strategy selection: $\optgp$, $\optk$, $\optkk$, and $\optm$.   $\optgp$ is designed for \emph{explicitly} represented workloads, and can scale to domains as large as $8192$.  $\optk$, $\optkk$, and $\optm$ are three different techniques for optimizing \emph{implicitly} represented workloads (with implicitly represented strategies), and can scale to significantly larger domains.\footnote{$\optk$ and $\optkk$ have linear dependence on the number of attributes, while $\optgp$ and any method that deals with explicitly represented workloads has an exponential dependence on the number of attributes.}  These optimization routines differ in the space of strategies they consider.  In all cases, the strategy search space is chosen so that is expressive enough to encode high-quality strategies, while also enabling tractable optimization.  
\item We also propose efficient algorithms for the measure and reconstruct steps of HDMM.  In the Matrix Mechanism, these steps are implemented by performing matrix operations with the explicit workload and strategy matrices and the data vector.  \sys exploits the implicit representation of the selected strategies to significantly speed up these steps.
\end{itemize}

As a result of these innovations, \sys achieves high accuracy on a wide variety of realistic input workloads, in both low and high dimensions.  In fact, in our experiments, we find it has higher accuracy than all prior select-measure-reconstruct techniques, even on input workloads for which the prior techniques were specifically designed.  It also achieves reasonable runtime and scales more effectively than prior work that performs non-trivial optimization (see \cref{sec:experiments} for a detailed scalability evaluation).  The main bottleneck of HDMM is \emph{representing the data in vector form}, which requires space proportional to the domain size; HDMM can scale to domains as large as $10^9$.\footnote{We show in this paper that in certain special cases, we can bypass this fundamental limitation.}

HDMM was first described by the authors in \cite{mckenna2018optimizing}. This paper provides a more complete description of HDMM and adds several new technical contributions:

\begin{enumerate}
\item We generalize and extend HDMM to support $(\epsilon, \delta)$-differential privacy via Gaussian noise.  This is a nontrivial extension: changing noise distributions fundamentally changes the optimization problems underlying MM and HDMM. We analyze this change and derive new optimization routines for strategy selection in this regime.  All four of our core optimization routines $\optgp,\optk,\optkk$, and $\optm$ require different changes to support Gaussian noise. 
\item We provide new results on the SVD bound \cite{li13optimal}, a simple formula which provides a lower bound on the achievable error of the Matrix Mechanism in terms of the properties of the workload.  Specifically, we show how the SVD bound can be efficiently computed for implicitly represented workloads, and we use the SVD bound to provide additional theoretical justification for our optimization routines.  We include the SVD bound in experiments to inform the optimality of the strategies found by our optimization routines.
\item We provide a complete description of $\optm$, the optimization routine that searches over marginal query strategies. This is one of most important optimization routines because it generally produces the best strategies for marginal query workloads, one of the most common types of workloads for multi-dimensional data.  In addition, we show that for marginal query workloads, it is sometimes possible to derive the optimal strategy in closed form, a remarkable new result.
\item We outline an approach to bypass the main bottleneck of HDMM, which is representing the data in vector form and performing the \reconstruct step.  In particular, we show that in certain special-but-common cases, representing the data in vector form is not necessary.  This is achieved by integrating HDMM with \texttt{Private-PGM} \cite{mckenna2019graphical}, an alternate technique for \reconstruct that produces a compact factorized representation of the data.  This modification allows HDMM to scale to even higher dimensional domains, far beyond settings where the data can be represented in vector form.  
\item We provide a more comprehensive set of experiments, in both low-dimensional and high-dimensional settings, showing consistent utility improvements over other mechanisms.  We also provide a detailed analysis of the experimental results.
\end{enumerate}

\subsection*{Organization} 

This paper is organized as follows.  In \cref{sec:data_query} we provide background on the data model and query representation.  In \cref{sec:privacy} we provide background on differential privacy, including the Matrix Mechanism.  In \cref{sec:optimization}, we describe $\optgp$, an optimization routine that approximately solves the Matrix Mechanism optimization problem for {\em explicitly} represented workloads.  In \cref{sec:implicit}, we show how many common workloads over high-dimensional domains can be {\em implicitly} represented in terms of Kronecker products.  In \cref{sec:hd_opt}, we describe $\optk$ and $\optkk$: two optimization routines that can effectively optimize implicitly represented workloads.  In \cref{sec:impmarg}, we show how marginal query workloads can be represented implicitly, using an even more compact representation than the one given in \cref{sec:implicit}.  In \cref{sec:margopt}, we describe $\optm$, an optimization routine that can efficiently optimize implicitly represented workloads with marginal query strategies.  In \cref{sec:running}, we describe the remaining steps of HDMM, including efficient \measure and \reconstruct.  We perform a thorough experimental study in \cref{sec:experiments}. Related work is discussed in \cref{sec:related} and proofs are provided in the appendix. 

\eat{

\gm{integrate some of this text above:}

{\it 
We use as a running example, and motivating use case, the differentially private release of a collection of 10 tabulations from the 2010 {\em Summary File 1 (SF1)}\cite{census-sf1}, an important data product based on the Census of Population and Housing (CPH). Statistics from SF1 are used for redistricting, demographic projections, and other policy-making.

In anticipation of the 2020 census, the Census Bureau recently produced test publications protected by a disclosure limitation system based on differential privacy \ryref{\cite{??}}{fix; not sure which publication this refers to}.  If the test of this disclosure limitation system is successful, then the expectation is that the publications of the 2020 Census will also be protected using differential privacy.  The work discussed in this paper is part of the research and development activity for these disclosure limitation systems.

}}

\eat{
\section{Outline of Intro}
\begin{outline}
		\1 Problem of releasing sets of linear queries under differential privacy
		\1 has many applications
			\2 Census bureau
			\2 healthcare databases that tabulate statistics
			\2 visualizations
		\1 state of the art: either is not error optimal, or require high computational overhead and are restricted to 1- or 2-D
			\2 First class of solutions considers tables in relational form, queries in a logical form (like SQL) and adds noise calibrated to (a bound over) the sensitivity of the queries.
				\3 sensitivity of each linear query is easy to compute (1 if counting query)
				\3 but, computing sensitivity of a set of queries is NP-hard
				\3 Xiaokui, follow up, ProPER, etc.
			\2 Second class of solutions have two important differences:
				\3 consider the databases and queries in vector form expressed over the full domain of each tuple (which is exponential in the number of attributes)
				\3 rather than answering the original set of queries, these techniques (a) select a set of measurement queries, (b) measure them using Laplace mechanism and (c) reconstruct the desired answers from the measurements.
				\3 theory: [hardt-talwar] and follow ups; [Matrix Mechanism] and variants; data dependent stuff (where measurements are incomplete wrt input queries)
				\3 These are the state of the art techniques in terms of error (maybe give an example) because:
					\4 error using these strategies are an order of magnitude smaller than the naive strategy
					\4 computing sensitivity is trivial when represented in the vector form (circumventing the hardness results)
					\4 standard inference algorithms that tend to use vectors and linear algebra can readily used to reconstruct answers to input queries using noisy measurements.
				\3 But representing queries and data in vector form makes these techniques extremely inefficient, limiting their applicability
					\4 matrix mechanism algorithm too slow for standard schemas and workloads
					\4 a lot of work on specialized algorithms for special cases
			\2 other work but not workload specific
		\1 Problem:
			\2 There is no general purpose differentially private algorithm that can answer sets of statistical queries posed over a table with (a) competetive error and (b) which scales to workloads over schemas with more than a few attributes.
		\1 Contributions:
			\2 A new suite of algorithms that can answer queries with competetive error on higher dimensional tables than was possible before
			\2 scales to domains of size $10^9$
			\2 key innovation:
				\3 translate logical queries into one out of a few compact vector representation. These representations are based on outer-products or weighted marginals
				\3 For each representation, we have designed novel scalable algorithms to identify a measurement strategy for which noisy answers are derived for the database
				\3 ...
			\2 nice results
\end{outline}
}

\section{Data model and query representation}  \label{sec:data_query}

In this section we introduce much of the notation and relevant background on the data model and query representation required to understand this work.  We use as a running example, and motivating use case, the differentially private release of a collection of 10 tabulations from the 2010 {\em Summary File 1 (SF1)}\cite{census-sf1}, an important data product based on the Census of Population and Housing (CPH). 

\subsection{Notation}

\begin{table}[b]
\begin{tabular}{|c|c|}
\hline
\textbf{Symbol} & \textbf{Meaning} \\\hline
$A$ & Attribute \\
$t$ & Tuple (database item) \\
$\phi$ & Predicate \\
$\Phi$ & Set of predicates \\
$\mathbb{I}$ & Indicator function \\
$\mathcal{W}$ & Logical workload \\
$\x$ & Data vector \\
$\q$ & Query vector \\
$\W$ & (Explicit) Workload matrix \\
$\A$ & (Explicit) Strategy matrix \\\hline
\end{tabular}
\begin{tabular}{|c|c|}
\hline \textbf{Symbol} & \textbf{Meaning} \\\hline
$\epsilon, \delta$ & Privacy parameters \\
$\algG$ & Noise addition mechanism \\
$\algG = \LM$ & Laplace mechanism \\
$\algG = \GM$ & Gaussian mechanism \\
$\Lone{\A}$ & $L_1$ sensitivity \\
$\Ltwo{\A}$ & $L_2$ sensitivity \\
$\norm{\A}_F$ & Frobenius norm \\
$\otimes$ & Kronecker product \\
$\WW$ & (Implicit) Workload matrix \\
$\AA$ & (Implicit) Strategy matrix  \\
\hline
\end{tabular}
\caption{Table of notation.} \label{tbl:notation}
\end{table}

A table of common notations is given in \cref{tbl:notation}.  In general, we adhere to the following conventions.  Scalars and tuples are lowercase, non-bold.  Sets are uppercase, non-bold.  Vectors are lowercase, bold.  Matrices are uppercase, bold.  Implicit matrices are uppercase, blackboard bold.

\newcommand{\attset}[1]{{\mathcal #1}}

\subsection{Data and schema}
We assume a single-table relational schema $R(A_1 \dots A_d)$, where $attr(R)$ denotes the set of attributes of $R$. Subsets of attributes are denoted $\attset{A} \subseteq attr(R)$. Each attribute $A_i$ has a finite domain $dom(A_i)$ with size $|dom(A_i)| = n_i$.  The full domain of $R$ is $dom(R) = dom(A_1) \times \dots \times dom(A_d)$, and has size $n = \prod_i n_i$.  An instance $I$ of relation $R$ is a multiset whose elements are tuples in $dom(R)$.

\begin{example}
The \texttt{\Person} relation has the following schema: six boolean attributes describing Race, two boolean attributes for Hispanic Ethnicity and Sex, Age in years between 0 and 114, and a Relationship-to-householder field that has 17 values. These queries are on a multidimensional domain of size $2^6 \times 2 \times 2 \times 115\times 17=\num{500480}$. The data also includes a geographic attribute encoding state (51 values including D.C.).  The SF1+ queries are defined on a domain of size $\num{500480} \times 51=\num{25524480}$.
\end{example}

\subsection{Logical view of queries} 
Predicate counting queries are a versatile class, consisting of queries that count the number of tuples satisfying any logical predicate. We define below a natural logical representation of these queries, distinguished from a subsequent vector representation.
\begin{definition}[Predicate counting query] \label{def:query}
A predicate on $R$ is a boolean function $\phi: dom(R) \rightarrow  \{0,1\}$.  A predicate can be used as a counting query on instance $I$ of $R$ whose answer is
	$\phi(I)=\sum_{t\in I} \phi(t)$.
\end{definition}

%
A predicate corresponds to a condition in the {\tt WHERE} clause of an SQL statement, so in SQL a predicate counting query has the form: \texttt{SELECT Count(*) FROM R WHERE $\phi$}.


When a predicate refers {\em only} to a subset of attributes $\attset{A} \subset attr(R)$ we may say that it is defined with respect to $\attset{A}$ and annotate it $ \phi_{\attset{A}} : dom(\attset{A}) \rightarrow \set{0,1}$.
If $\phi_\attset{A}$ and $\phi_\attset{B}$ are predicates on attribute sets $\attset{A}$ and $\attset{B}$, then their conjunction is a predicate $\phi_\attset{A} \wedge \phi_\attset{B} : dom(\attset{A} \cup \attset{B}) \rightarrow \set{0,1}$.


We assume that each query consists of \emph{arbitrarily complex} predicates on each attribute, but require that they are combined across attributes with conjunctions. In other words, each $\phi$ is of the form $\phi = \phi_{A_1} \land \dots \land \phi_{A_d}$.  This facilitates the compact implicit representations described in \cref{sec:implicit}.  One approach to handling disjunctions (and other more complex query features) is to transform the schema by merging attributes.  We illustrate this below in its application to the SF1 workload and provide a more general approach to disjunctive queries in \cref{sec:disjuncts}.

\begin{example}
The SF1 workload consists of conjunctive conditions over its attributes, with the exception of conditions on the six binary race attributes, which can be complex disjunctions of conjunctions (such as ``The number of Persons with two or more races'').  We simply merge the six binary race attributes and treat it like a single $2^6=64$ size attribute (called simply {\em Race}). This schema transformation does not change the overall domain size, but allows every SF1 query to be expressed as a conjunction.
\end{example}

\subsection{Logical view of query workloads} \label{sec:sub:logical_workloads}

A workload is a set of predicate counting queries $\Phi = \set{\phi_1, \dots, \phi_m}$. A workload may consist of queries designed to support a variety of analyses or user needs, as is the case with the SF1 workload described above.  Workloads may also be built from the sufficient statistics of models, or generated by tools that aid users in exploring data, or a combination thereof.  For the privacy mechanisms considered here it is preferable for the workload to explicitly mention all queries of interest, rather than a subset of the queries that could act like a supporting view, from which the remaining queries of interest could be computed.  Enumerating all queries of interest allows error to be optimized collectively.  In addition, a workload query can be repeated, or equivalently, weighted, to express the preference for greater accuracy on that query.

\begin{example}
	Our example workload is a subset of queries from SF1 that can be written as predicate counting queries over a \texttt{\Person} relation.  (We omit other queries in SF1 that involve groups of persons organized into households; for brevity we refer to our selected queries simply as SF1.)  Our SF1 workload has 4151 predicate counting queries, each of the form \texttt{SELECT Count(*) FROM \Person WHERE $\phi$}, where $\phi$ specifies some combination of demographic properties (e.g. number of Persons who are Male, over 18, and Hispanic) and thus each query reports a count at the national level.
\end{example}

\begin{example}
A workload we call SF1+ consists of the national level queries in SF1 {\em as well as} the same queries at the state level for each of 51 states. We can succinctly express the state level queries as an additional 4151 queries of the form:
\texttt{SELECT state, Count(*) FROM \Person WHERE $\phi$ GROUP BY state}. Thus, SF1+ can be represented by a total of $4151+4151=8302$ SQL queries. (The \texttt{GROUP BY} is a succinct way to represent a potentially large set of predicate counting queries.) The SF1+ queries are defined on a domain of size $\num{500480} \times 51=\num{25524480}$. In addition to their SQL representation, the SF1 and SF1+ workloads can be naturally expressed in a logical form defined in \cref{sec:implicit_conjunctions}. We use \Wnt and \Wst to denote the logical forms of SF1 and SF1+ respectively.
\end{example}



\paragraph*{\textbf{Structured multi-dimensional workloads}}

Multi-dimens\-ional workloads are often defined in a structured form, as {\em products} and {\em unions of products}, that we will exploit later in our implicit representations.  Following the notation above, we write $\Phi_\attset{A}$ to denote a set of predicates, each mentioning only attributes in $\attset{A}$.  For example, the following are common predicate sets defined over a single attribute $A$ of tuple $t$:\vspace{1ex} \\ \vspace{1ex}
\begin{tabular}{lll}
$I$ & $\mbox{Identity}_A$ & $= \set{ \: \mathbb{I}[t_A = a] \mid a \in dom(A) \:}$  \\
$P$ & $\mbox{Prefix}_A$   & $= \set{ \: \mathbb{I}[t_A \leq a] \mid a \in dom(A) \:}$ \\
$R$ & $\mbox{AllRange}_A$ & $= \set{ \: \mathbb{I}[a \leq t_A \leq b] \mid a, b \in dom(A), a \leq b  \:}$ \\
T & $\mbox{Total}_A$ & $=\set{ \: \mathbb{I}[\text{True}] \: }$ 
\end{tabular}

Above, $\mathbb{I}$ is the indicator function, e.g., $\mathbb{I}[t_A = a] = \phi(t_A) = \begin{cases} 1 & \text{if } t_A = a \\ 0 & \text{otherwise} \end{cases} $.

$\mbox{Identity}_A$ contains one predicate for each element of the domain.  Both $\mbox{Prefix}_A$ and $\mbox{Range}_A$ rely on an ordered $dom(A)$; they contain predicates defining a CDF (i.e. sufficient to compute the empirical cumulative distribution function), and the set of all range queries, respectively. The predicate set $\mbox{Total}_A$, consists of a single predicate, returning {\em True} for any $a \in dom(A)$, and thus counting all records.

We can construct multi-attribute workloads by taking the cross-product of predicate sets defined for single attributes, and {\em conjunctively} combining individual queries.

\begin{definition}[Product] \label{def:product-workload}
The product of two predicate sets $\Phi_\attset{A}$ and $\Phi_\attset{B}$ is another predicate set $\Phi_\attset{A} \times \Phi_\attset{B} = \{ \phi_\attset{A} \wedge \phi_\attset{B} \mid \phi_\attset{A} \in \Phi_\attset{A}, \phi_\attset{B} \in \Psi_\attset{B} \} $ containing the conjunction of every pair of predicates.
\end{definition}


We describe several examples of workloads constructed from products and unions of products below.  

\begin{example}[Single query as product] \label{ex:sql_query}
A predicate counting query in the SF1 workload is: \texttt{SELECT Count(*) FROM \Person WHERE sex=M AND age < 5}.
We can express this query as a product: first, define predicate set $\Phi_1 = \set{\mathbb{I}[t_{sex} = M] }$ and predicate set $\Phi_2 = \set{ \mathbb{I}[t_{age} < 5] }$.  The query is expressed as the product $\Phi_1 \times \Phi_2$.  (We omit Total on the other attributes for brevity.)
\end{example}

\begin{example}[\texttt{GROUP BY} query as product] \label{ex:groupby}
A \texttt{GROUP BY} query can be expressed as a product by including an Identity predicate set for each grouping attribute and a singleton predicate set for each attribute in the \texttt{WHERE} clause.
The product would also include Total for each attribute not mentioned in the query.
For example, the query \texttt{SELECT sex, age, Count(*) FROM \Person WHERE his\-panic = TRUE GROUP BY sex, age} is expressed as
$\mbox{I}_{\texttt{sex}} \times \mbox{I}_{\texttt{age}} \times \Phi_{3}$ where $\Phi_{3} = \set{\mathbb{I}[t_{hispanic} = True]}$.  This product contains $2 \times 115$ counting queries, one for each possible setting of sex and age.
\end{example}
\begin{example}[Marginal and Prefix-Marginal] \label{ex:marg}
A marginal query workload is defined by the product of one or more Identity predicates on selected attributes and Total on all other attributes.  For example, $\mbox{I}_{\tt Sex} \times \mbox{I}_{\tt Age} \times \mbox{I}_{\tt Hispanic}$ contains predicates to compute one three-way marginal, and consists of $2 \times 115 \times 2$ counting queries, one for each possible setting of Sex, Age and Hispanic.  This is equivalent to a \texttt{GROUP BY} query on Sex, Age, and Hispanic with no \texttt{WHERE} clause.  

A prefix-marginal query workload is a natural generalization of a marginal query workload which can be obtained by replacing one or more of the Identity predicates with Prefix.  For example, $\mbox{I}_{\tt Sex} \times \mbox{P}_{\tt Age} \times \mbox{I}_{\tt HIspanic}$ contains predicates to compute one three-way prefix-marginal, and consists of $ 2 \times 115 \times 2 $ counting queries of the form $ \mathbb{I}[t_{sex}=a, t_{age} \leq b, t_{hispanic} = c]$.  
\end{example}

Marginals query workloads are very common workloads that make sense for domains with categorical attributes, while prefix-marginals are more natural for domains with both categorical and discretized numeric attributes.
\begin{example}[SF1 Tabulation as Product] \label{ex:tabulation}
Except for the population total, the queries in the P12 tabulation of the Census SF1 workload \cite{census-sf1} can be described by a single product: $\mbox{I}_{\tt Sex} \times \mbox{R}_{\tt Age} $ where $ \mbox{R}_{\tt Age} $ is a particular set of range queries including $[0,114], [0,4]$, $[5,9]$, $[10,14], \dots [85,114]$.
\end{example}

\paragraph*{Unions of products}
Our workloads often combine multiple products as a union of the sets of queries in each product.  For example, the queries required to compute all three-way marginals could be represented as a union of ${ d \choose 3}$ workloads, each a product of the Identity predicate set applied to three attributes.  The input to the algorithms that follow is a logical workload consisting of a union of products, each representing one or possibly many queries.

\begin{definition}[Logical workload] \label{def:logical_workload}
A logical workload $\Wlog=\{Q_1 \dots Q_k\}$ consists of a set of products $Q_i$ where each $Q_i=\Phi^{(i)}_{A_1} \times \dots \times \Phi^{(i)}_{A_d}$.
\end{definition}

\begin{example}[SF1 as union of products] \label{ex:product_terms_sf1}
The SF1 workload can be represented in a logical form, denoted $\Wnt$, that consists of a union of $k=4151$ products, each representing a single query.  Because these queries are at the national level, there is a Total predicate set on the State attribute.
The logical form of the SF1+ workload, denoted \Wst, includes those products, plus an additional 4151 products that are identical except for replacing the Total on State with an Identity predicate set.  There are a total of $k=8302$ products, representing a total of $4151 + 51 \times 4151 = \num{215852}$ predicate counting queries.  While this is a direct translation from the SQL form, this representation can be reduced.  First, we can reduce to $k=4151$ products by simply adding $True$ to the Identity predicate set on State to capture the national counts.
Furthermore, through manual inspection, we found that both $\Wnt$ and $\Wst$ can be even more compactly represented as the union of 32 products---we  use $\Wnt^*$ and $\Wst^*$ to denote more compact logical forms. This results in significant space savings (as described shortly in \cref{ex:implicit_workload_size}) and runtime improvements.
\end{example}

\subsection{Explicit data and query vectorization} \label{sec:sub:vdata}

The vector representation of predicate counting queries (and the data they are evaluated on) is central to the select-measure-reconstruct paradigm.  
The vector representation of instance $I$ is denoted $\x_I$ (or simply $\x$ if the context is clear) and called the {\em data vector}.

\begin{definition}[Data vector]
The data vector representation of an instance $I$, denoted $\x_I$, is a vector indexed by tuples $ t \in dom(R) $, so that $ \x_I(t) = \sum_{t' \in dom(R)} \mathbb{I}[t = t'] $.
\end{definition}

Informally, $\x(t)$ counts the number of occurences of $t$ in $I$.
Note that, throughout the paper, the representation of the data vector is {\em always} explicit; it is the representation of queries that will be implicit.
Every predicate counting query $\phi$ also has a corresponding vector form.
\begin{definition}[Vectorized query] \label{def:vect-query}
The vector representation of a predicate counting query $\phi$, denoted $\q_{\phi} = vec(\phi)$ is a vector indexed by tuples $t \in dom(R)$, so that $ \q_{\phi}(t) = \phi(t) $.  
\end{definition}

The function $vec(\phi)$ which transforms a logical predicate into its corresponding vector form has a simple implementation: compute $\phi(t)$ for every $t \in dom(R) $ and store the results in a vector.   Note that both the data vector and the vectorized query have size $|dom(R)| = N$.

\begin{proposition}[Query evaluation]
A predicate counting query can be answered by computing the dot product between the query vector and data vector: $ \phi(I) = \q_{\phi}^T \x_I $.
\end{proposition}

To see why this works, observe that $ \phi(I) = \sum_{t \in I} \phi(t) = \sum_{t \in dom(R)} \phi(t) \x_I(t) = \q_{\phi}^T \x_I $.
A single predicate counting query is represented as a vector, so a workload of predicate counting queries can be represented as a matrix in which queries are rows.  For logical workload $\Wlog$, its (explicit) matrix form is written $\W$, and the evaluation of the workload is equivalent to the matrix product $\W \x_I$.  Note that the size of the workload matrix is $m \times n$ where $m$ is the number of queries, $\x_I$ is $n \times 1$, and the vector of workload answers is $m \times 1$.

\section{Privacy background} \label{sec:privacy}
Differential privacy is a property of a randomized algorithm that bounds the difference in output distribution induced by changes to an individual's data. Let $\nbrs(\db)$ be the set of databases differing from $I$ in at most one record.
%
\begin{definition}[Differential Privacy~\cite{dwork2006calibrating}] \label{def:diffp}
A randomized algorithm $\algG$ is $(\epsilon,\delta)$-differentially private if for any instance $\db$, any $\db' \in \nbrs(\db)$, and any outputs $O \subseteq Range(\algG)$,
\[
Pr[ \algG(\db) \in O] \leq \exp(\epsilon) \times Pr[ \algG(\db') \in O] + \delta
\]
\end{definition}

%

When $\delta=0$, we say $\algG$ is $\epsilon$-differentially private.  In this work we consider algorithms for answering a workload of predicate counting queries, in which case we invoke the mechanism with the workload matrix and data vector: $ \algG(\W, \x) $.
We focus on mechanisms that give unbiased answers to the workload queries (sometimes called data-independent mechanisms).  
Below we define two simple mechanisms of this form: the Laplace mechanism ($\algG = \LM$) and the Gaussian mechanism ($\algG = \GM$).

\begin{definition}[Laplace mechanism] \label{prop:laplace}
Given an $m \times n$ query matrix $\W$, and a noise magnitude $\sigma$, the Laplace Mechanism $\LM$ outputs the vector:
$\LM(\W,\x) = \W\x + \Lone{\W} \Lap(b)^m$ where $\Lap(b)^m$ is a vector of $m$ i.i.d. samples from a Laplace distribution with scale $b$, and $\Lone{\W}$ denotes the maximum $L_1$ norm ofthe columns of $\W$.
\end{definition}

\begin{definition}[Gaussian mechanism] \label{prop:gaussian}
Given an $m \times n$ query matrix $\W$, and a noise magnitude $\sigma$, the Gaussian mechanism $\GM$ outputs the vector:
$\GM(\A,\x) = \A\x + \Ltwo{\W} \Gaus(\sigma)^m$ where $\Gaus(\sigma)^m$ is a vector of $m$ i.i.d. samples from a Gaussian distribution with scale $\sigma$ and $\Ltwo{\W}$ denotes the maximum $L_2$ norm of the columns of $\W$.
\end{definition}

The quantities $\Lone{\W}$ and $\Ltwo{\W}$ above are the {\em $L_1$ sensitivity} and {\em $L_2$ sensitivity} of the query set defined by $\W$ respectively, since they measure the maximum difference in the answers to the queries in $\W$ on any two databases that differ only by a single record \cite{li2015matrix}.  In the remainder of the paper, we will use $\norm{\W}_{\algG}$ to denote this sensitivity norm when referring to a general quantitity that applies for both Laplace noise ($\algG = \LM$) and Gaussian noise ($\algG = \GM$).  As long as $b$ or $\sigma$ is sufficiently large, these two mechanisms can provide differential privacy.  The precise conditions are stated in \cref{prop:privacy1}.

\begin{proposition}[Privacy of Laplace and Gaussian mechanisms~\cite{Dwork14Algorithmic,balle2018improving}] \label{prop:privacy1}
The Laplace mechanism is $\epsilon$-differentially private as long as $b \geq \frac{1}{\epsilon}$ and the Gaussian mechanism is $(\epsilon, \delta)$ differentially private as long as $\delta \geq \Psi(\frac{1}{2\sigma} - \epsilon \sigma) - \exp{(\epsilon)} \Psi(-\frac{1}{2 \sigma} - \epsilon \sigma) $, where $\Psi$ is the cumulative distribution function of the standard gaussian distribution.
\end{proposition}

Above, the condition on $\sigma$ for the Gaussian mechanism corresponds to the so-called ``analytic gaussian mechanism'' \cite{balle2018improving}, which \emph{exactly} calibrates the minimum $\sigma$ needed to ensure $(\epsilon, \delta)$-differential privacy.  For this method, $\sigma$ can be obtained numerically with a root-finding algorithm.  Other simpler formulas exist for $\sigma$ that are not as tight, using a classical analysis of the Gaussian mechanism \cite{Dwork14Algorithmic} or through R\'{e}nyi differential privacy or concentrated differential privacy \cite{mironov2017renyi,bun2016concentrated}. 
It is straightforward to analyze the error of these two mechanisms since they add i.i.d. noise with the same (known) variance to all workload queries.  In particular, we use expected total squared error as our error metric:

\begin{definition}[Expected Error~\cite{li2010optimizing}] \label{def:error}
Given a $m \times n$ workload matrix $\W$ and a differentially-private algorithm $\algG$, the expected total squared error is:
$$ \Error(\W, \algG) = \mathbb{E}[\norm{\W \x - \algG(\W, \x)}_2^2] $$
where the expectation is taken over the randomness in the privacy mechanism $\algG$.
\end{definition}

\begin{proposition}[Error of Laplace and Gaussian mechanisms] \label{prop:error0}
The Laplace and Gaussian mechanisms are unbiased and have the following expected error:

\begin{align*}
\Error(\W, \LM) &= 2 m b^2 \Lone{\W}^2 &
\Error(\W, \GM) &= m \sigma^2 \Ltwo{\W}^2
\end{align*}
\end{proposition}

As evident by \cref{prop:error0}, the expected error of these mechanisms depends crucially on the sensitivity of the query matrix $\W$, and if this is large then the total squared error will also be large.  We now introduce a generalization of these mechanisms that often has better expected error.

\subsection{The matrix mechanism} \label{sec:sub:mm}

The core idea of the matrix mechanism is to apply the mechanism $\algG$ on a new query matrix $\A$, then use the noisy answers to the queries in $\A$ to estimate answers to the queries in $\W$.  The benefits of this approach will become clear when we reason about the expected error. 

\begin{definition}[Matrix mechanism \cite{li2010optimizing}] \label{prop:matrixmech}
Given a $m \times n$ workload matrix $\W$, a $p \times n$ strategy matrix $\A$, and a differentially private algorithm $\algG(\A, \x)$ that answers $\A$ on $\x$, the mechanism $\mathcal{M}_{\A,\algG}$ outputs the following vector: $ \mathcal{M}_{\A, \algG}(\W, \x) = \W \A^+ \algG(\A,\x) $.
\end{definition}

The privacy of the Matrix mechanism follows from the privacy of $\algG$, since that is the only part of the mechanism that has access to the true data.  In this paper, we assume $\algG$ is either the Laplace mechanism or the Gaussian mechanism, although in principle other noise-addition mechanisms are also possible \cite{ghosh2012universally, hardt2010geometry, li2015matrix}.  Under some mild conditions stated below, the Matrix mechanism is unbiased and the error can be expressed analytically as shown below:


\begin{proposition}[Error of the Matrix mechanism \cite{li2010optimizing}] \label{prop:error}
The matrix mechanism is unbiased, i.e., $ \mathbb{E}[\mathcal{M}_{\A, \algG}(\W, \x)] = \W \x $, and has the following expected error:
\begin{align*}
\Error(\W, \mathcal{M}_{\A, \algG}) &= \Error(\A, \algG) \norm{\W \A^+}_F^2 \\
&\propto \Lk{\A}^2 \norm{\W \A^+}_F^2
\end{align*}
as long as $ \W \A^+ \A = \W $ and $\algG$ adds i.i.d. noise with mean $0$.
\end{proposition}

When invoked with $ \A = \W$, the matrix mechanism is very similar to base mechanism $\algG$, but the error is typically lower.  The term $\norm{\W \W^+}_F^2$ in the error formula is equal to the rank of $\W$, which is bounded above by $m$.  This implies that the error of $\mathcal{M}_{\W, \algG}$ can never be higher than $\algG$.  Further, there are often much better strategies to select than $\A = \W$.  

Finding the best strategy $\A$ for a given workload $\W$ is the main technical challenge of the matrix mechanism.  This strategy selection problem can be formulated as a constrained optimization problem.  However, solving this problem is computationally expensive, especially when $\algG = \LM$, where it is generally infeasible to solve it for any nontrivial input workload.

\subsection{Lower bounds on error} \label{sec:sub:lowerbounds}

An important theoretical question is to identify, or bound, how low the error of the matrix mechanism can be for a given workload $\W$.  This is useful because finding the strategy with minimum error is a difficult (and often intractable) problem, but computing a lower bound on error can be done efficiently.  Knowing how low the error can be allows one to compare the error of a concrete strategy to the lower bound to see how close to optimal it is.   Additionally, the lower bound can be used to make important policy decisions, such as setting the privacy budget, or whether it is worth investing the resources to find a good strategy (as opposed to using other types of mechanisms like data-dependent ones). 
Li and Miklau studied this problem and derived the SVD bound \cite{li13optimal}.

\begin{definition}[SVD Bound \cite{li13optimal}]
Given a $m \times n $ workload $\W$, the singular value bound is:

$$SVDB(\W) = \frac{1}{n} \big( \lambda_1 + \dots + \lambda_n \big)^2 $$

where $ \lambda_1, \dots, \lambda_n $ are the singular values of $\W$.
\end{definition}

The SVD bound gives a lower bound on the error achievable by the matrix mechanism.

\begin{proposition}[SVD Bound \cite{li13optimal}]
Given a $ m \times n $ workload $\W$ and a $p \times n $ strategy $\A$ that supports $\W$:
$$ SVDB(\W) \leq \Lk{\A}^2 \norm{\W \A^+}_F^2 $$
\end{proposition}

The SVD Bound is known to be tight when $ \algG = \GM$ and some conditions on $\W$ are satisfied, meaning there is some strategy $\A$ that achieves the equality.  When $ \algG = \LM $ the bound may not be tight however. In this paper we use the SVD Bound to evaluate the quality of the strategies found by our optimization routines.  We also derive expressions for efficiently computing the SVD Bound for implicitly represented workloads, and use this analysis to theoretically justify our optimization routines.

\section{Optimizing Explicit Workloads} \label{sec:optimization}

In this section, we introduce the main optimization problem that underlies strategy selection for the matrix mechanism, and describe $\optgp$, our algorithm for solving it.  We assume for now that the workload is represented \emph{explicitly} as a dense matrix.  The methods we describe are useful by themselves for workloads defined over modest domains (namely, those smaller than about $n=10^4$), and they are an essential building block for the more scalable methods we describe in \cref{sec:hd_opt}. 

\subsection{The optimization problem} \label{sec:sub:problem}

Our goal is to find a query strategy that offers minimal expected error on the workload.  Using the analytic error formula from \cref{prop:error}, this can be defined as a constrained optimization problem. One formulation of this problem is stated below:

\begin{problem}[Matrix Mechanism Optimization \cite{li2015matrix}]\label{prob:strategy_selection}
Given an $m\times n$ workload matrix $ \W $:
\begin{equation} \label{eq:opt}
\begin{aligned}
& \underset{\A}{\text{minimize}}
& & \norm{\A}_{\algG}^2 \norm{ \W \A^+ }_F^2 \\
& \text{subject to} & & \W \A^+ \A = \W \\
\end{aligned}
\end{equation}
\end{problem}

For a number of reasons, this optimization problem is difficult to solve exactly: it has many variables, it is not convex, and both the objective function and constraints involve $\A^+$, which can be slow to compute. In addition, $ \norm{ \W \A^+ }_F^2 $ has points of discontinuity near the boundary of the constraint $ \W\!\A^+\!\A = \W $. This problem was originally formulated as a rank-constrained semi-definite program \cite{li2010optimizing}, and, while algorithms exist to find the global optimum, they require $O(m^4(m^4+n^4))$ time, making it infeasible in practice. 

Gradient-based numerical optimization techniques can be used to find locally optimal solutions to Problem \ref{prob:strategy_selection}.  These techniques begin by guessing a solution $\A_0$ and then iteratively improve it using the gradient of the objective function to guide the search. The process ends after a number of iterations are performed, controlled by a stopping condition based on improvement of the objective function.  The constraints complicate the problem further, but even if we ignore them,
gradient-based optimization is slow, as the cost of computing the objective function for general $\A$ is $O(n^3)$, e.g. requiring more than $6$ minutes for $n=8192$.

In the next sections, we provide algorithms for efficiently solving Problem \ref{prob:strategy_selection}.  We separately consider the two cases of Gaussian noise and Laplace noise, as the required techniques are quite different.

\subsection{Strategy Optimization with Gaussian Noise} \label{sec:sub:convex} 

While Problem \ref{prob:strategy_selection} with $\algG = \GM$ is not convex in its current form, it can be reformulated into an equivalent problem that is convex \cite{li2015matrix,yuan2016convex}.  The key idea is that the objective function can be expressed in terms of $\X = \A^T \A$, since $\Ltwo{\A}^2 = max(diag(\A^T \A))$ and $\norm{\W \A^+}_F^2 = tr[(\A^T \A)^+ (\W^T \W)]$.  This allows us to optimize $\X$ instead of $\A$, and then we can recover $\A$ by performing Cholesky decomposition on $\X$.  Remarkably, the resulting problem is convex with respect to $\X$. 

\begin{definition}[Convex Reformulation \cite{yuan2016convex}]\label{prob:convex}
Given a workload matrix $ \W $ of rank $n$, let $\optgp(\W) = \A$ where $\A^T \A$ is a Cholesky decomposition of $\X^*$ and:
\begin{equation} \label{eq:opt}
\begin{aligned}
\X^* = & \underset{\X}{\text{minimize}}
& & tr[\X^{-1} (\W^T \W)] \\
& \text{subject to} & & diag(\X) = \mathbf{1} \\
& & & \X \succ 0
\end{aligned}
\end{equation}
\end{definition}



While the above problem is convex, it is still nontrivial to solve due to the dependence on the matrix inverse and the constraint $ \X \succ 0 $ ($\X$ is positive definite).  The equality constraint $diag(\X)=1$ and corresponding optimization variables $diag(\X)$ can easily be eliminated from the problem since they must always equal $1$.  Additionally $\X$ must be a symmetric matrix so we can optimize over the lower triangular entries, essentially reducing the number of optimization variables by a factor of two. The full details of a conjugate gradient algorithm for solving this problem are available in \cite{yuan2016convex}.  The per-iteration runtime of their ``COA'' algorithm is $O(n^3)$, and it typically requires about 50 iterations to converge.  In practice it is able to scale up to about $n \approx 10^4$.  The algorithm works well in practice when $n$ is small, but is not particularly robust for some workloads when $n$ is larger, which we observe empirically in \cref{sec:experiments} (e.g., COA on prefix workload in \cref{table:opt02}).  

We thus design our own algorithm to solve the same optimization problem, which is based on the same principles as the COA technique, but is more robust in practice.  There are two key differences in our implementation.  First, we initialize the optimization intelligently by setting $ \X = \P \mathbf{\sqrt{\Lambda}} \P^T $ where $ \P \mathbf{\Lambda} \P^T $ is the eigen-decomposition of $\W^T \W$.  This an approximation to the optimal strategy based on the SVD bound \cite{li13optimal}, and acts as a very good initialization.  Second, instead of using the custom-designed conjugate gradient algorithm proposed in \cite{yuan2016convex}, we simply use scipy.optimize, an off-the-shelf optimizer.  We heuristically ignore the constraint $\X \succ 0$ during optimization, using a large loss value when it is not satisfied.  This makes the problem unconstrained, and readily solvable by scipy.optimize.  The constraint is verified to hold at the end of the optimization.  These changes lead to more robust optimization that produce strategies nearly matching the SVD bound, as we show experimentally in \cref{sec:experiments}.

\subsection{Strategy Optimization with Laplace noise} \label{sec:sub:generalpurpose}

Unfortunately, the techniques used in the previous section do not apply to the Laplace noise setting, as the sensitivity norm $ \Lone{\A} $ cannot be expressed in terms of $ \A^T \A $.  In this section, we describe an alternate approach to approximately solve Problem \ref{prob:strategy_selection}: parameterized strategies.  Our key idea is to judiciously restrict the search space of the optimization problem to simplify the optimization while retaining expressivity of the search space.  While our approach does not necessarily produce a globally optimal solution to Problem \ref{prob:strategy_selection}, with good parameterizations it can still find state of the art strategies.  Below we describe the idea of parameterized strategies in its full generality. Then we propose a specific parameterization that works well for a variety of input workloads. 

A parameterization is a function $ \A(\btheta) $ mapping a real-valued parameter vector $\btheta$ to a strategy $\A$.  Optimizing over a parameterized strategy space can be performed by optimizing $\btheta$ rather than $\A$.  In the extreme case, there may be one entry in $\btheta$ for every entry in $\A$, but a more careful design of the parameterization with fewer parameters and a smart mapping between the entries of the parameter vector and the entries of the strategy matrix can lead to more effective optimization.  There are several design considerations for setting up a good parameterization.  First, the parameterization must be expressive enough to encode high-quality strategies (this may depend on the workload).  Second, the parameterization should have structure that makes the optimization problem more computationally tractable (such as eliminating constraints).
Third, the parameterization may encode domain expertise about what a good strategy should look like, which could make it easier to find high-quality local minima. 

Several existing privacy mechanisms can be thought of as an instance of this parameterization framework \cite{qardaji2013understanding,li2014data,ding2011differentially,yuan2012low,li2012adaptive}.  These mechanisms typically are designed for a specific workload or workload class.  For example, Qardaji et al. consider the space of hierarchical strategies, which is parameterized by a single parameter: the branching factor of the hierarchy, which is chosen to minimize MSE on the workload of \emph{all range queries} \cite{qardaji2013understanding}.  Li et al. consider the space of weighted hierarchical strategies (with fixed branching factor), which is parameterized by vector of scaling factors for each query in the hierarchy \cite{li2014data}.  This method adapts to the input workload, but the hierarchical parameterization only works well for workloads of range queries.  Ding et al. consider the space of marginal query strategies (and workloads of the same form), the parameters can be thought of as a binary vector of length $2^d$ corresponding to which marginals should be included in the strategy \cite{ding2011differentially}.  These parameters are optimized using a greedy heuristic.  Li et al. consider the space of strategies containing the eigenvectors of the workload gram matrix, which is parameterized by a vector of scaling factors for each eigenvector \cite{li2012adaptive}.  This parameterization works well when $\algG = \GM$, but not when $ \algG = \LM $, and it has recently been subsumed by the work of Yuan et al. for $ \algG = \GM $ \cite{yuan2016convex}. 

We now present a new general-purpose parameterization, called $p$-Identity, which handles these considerations without making strict assumptions about the structure of the workload.  It also out-performs all of the existing parameterizations, even on the workloads for which they were designed. The parameters of a p-Identity strategy are more naturally interpreted as a matrix $\bTheta$ rather than a vector $\btheta$ so we instead use the notation $ \A(\bTheta) $.   

\begin{definition}[$p$-Identity strategies] \label{def:gp} 
Given a $p \times n$ matrix of non-negative values $\matr{\Theta}$, the $p$-Identity strategy matrix $\A(\matr{\Theta})$ is defined as follows:
\begin{equation*} \label{eq:reparam}
A(\matr{\Theta}) = \begin{bmatrix} \I \\ \matr{\Theta} \end{bmatrix} \D
\end{equation*}
where $\I$ is the identity matrix and $ \D = diag(1 + \vect{1}^T \matr{\Theta})^{-1} $.
\end{definition}

Above, the matrix $\A(\bTheta)$ contains $n+p$ queries, including $n$ identity queries (i.e., queries that count the number of records in the database for each domain element in $dom(R)$) and $p$ arbitrary queries that depend on the parameters $\bTheta$.  The diagonal matrix $\D$ is used to re-weight the strategy so that $ \Lone{\A} = 1 $ and each column of $\A$ has the same $L_1$ norm.  This is an example of domain knowledge incorporated into the parameterization, as it has been shown that optimal strategies have uniform column norm \cite{li2015matrix}. \footnote{If a strategy did not, a query could be added to it without increasing sensitivity and addition of that query would result in error less than or equal to that of the original strategy.}

\begin{example}
For $p=2$ and $n=3$, we illustrate below how $\A(\matr{\Theta})$ is related to its parameter matrix, $\matr{\Theta}$. 
{\small
\begin{align*} 
\matr{\Theta} = \begin{bmatrix} 1 & 2 & 3 \\ 1 & 1 & 1 \end{bmatrix} & & &
\A(\matr{\Theta}) = \begin{bmatrix} 0.33 & 0 & 0 \\ 0 & 0.25 & 0 \\ 0 & 0 & 0.2 \\ 0.33 & 0.5 & 0.6 \\ 0.33 & 0.25 & 0.2 \end{bmatrix}
\end{align*} }
\end{example} 

For this class of parameterized strategies, the resulting optimization problem requires optimizing $\bTheta$ instead of $\A$ and is stated below; we use $\optgp$ to denote the operator that solves this problem.

\begin{definition}[parameterized optimization]\label{prob:gp}
Given a workload matrix $ \W $ and hyper-parameter $p$, let $\optgp(\W) = \A(\bTheta^*)$ where:
\begin{equation*} 
\begin{aligned}
& \bTheta^* = \underset{\matr{\Theta} \in \mathbb{R}_+^{p \times n}}{\text{argmin}}
& & \norm{ \W \A(\bTheta)^+ }_F^2 \\
\end{aligned}
\end{equation*}
\end{definition}

This parameterization was carefully designed to simplify optimization.  Because $\A(\bTheta)$ is full rank, the constraints are satisfied by construction, and we can solve an unconstrained problem instead.  Furthermore, $\norm{\A(\bTheta)}_1 = 1$ for all $\bTheta$, so that term can be removed from the objective.  Additionally, due to the special structure of $\A(\bTheta)$ we can efficiently evaluate the objective and its gradient in $O(p n^2)$ time instead of $O(n^3)$ time.  For example, when $n = 8192$ it requires $>6$ minutes to evaluate the objective for a general strategy $\A$, while it takes only $1.5$ seconds for a $p$-Identity strategy (with $p=\frac{n}{16}$), which is a $240\times$ improvement.  Despite this imposed structure, $\A(\bTheta)$ is still expressive enough to encode high-quality strategies.  Moreover, $p$ can always be tuned to balance expressivity with efficiency.  

\eat{

\vspace{1ex}
\noindent {\bf \em Constraint resolution} \quad Problem \ref{prob:gp} is a simpler optimization problem than Problem \ref{prob:strategy_selection} because it is unconstrained: $ \W \!\A^+\! \A = \W $ for all $\A=\A(\matr{\Theta})$ because p-Identity strategies are full rank by construction.  Furthermore $\Lone{\A} = 1$ for all $\A=\A(\bTheta)$ by construction, which simplifies the objective. 

\vspace{1ex}
\noindent {\bf \em Expressiveness} \quad Considering only $p$-Identity strategy matrices could mean that we omit from consideration the optimal strategy, but it also reduces the number of variables, allowing more effective gradient-based optimization.  The expressiveness depends on the parameter $p$ used to define matrices $\A(\matr{\Theta})$, which is an important hyper-parameter.  
In practice we have found $p \approx \frac{n}{16}$ to provide a good balance between efficiency and expressiveness for complex workloads (such as the set of all range queries). For less complex workloads, an even smaller $p$ may be used.  

\vspace{1ex}
\noindent {\bf \em Efficient gradient, objective, and inference} \quad 
To a first approximation, the runtime of gradient-based optimization is
$\mbox{\sc \#restarts}*\mbox{\sc \#iter}*(\mbox{\sc cost}_{grad}+\mbox{\sc cost}_{obj})$,
where $\mbox{\sc cost}_{grad}$ is the cost of computing the gradient, $\mbox{\sc cost}_{obj}$ is the cost of computing the objective function, {\sc \#restarts} is the number of random restarts, 
and {\sc \#iter} is the number of iterations per restart.  
For strategy $\A$, the objection function, denoted $C(\A)$, and the gradient function, denoted $ \frac{\partial C}{\partial \A} $, are defined: 
\begin{align}
C(\A) &= \norm{ \W \A^+ }_F^2 = tr[(\A^T \A)^+ (\W^T \W)] \label{eq:obj} \\  
\frac{\partial C}{\partial \A} &= -2 \A (\A^T \A)^+ (\W^T \W) (\A^T \A)^+ \label{eq:grad}
\end{align}

Note that the above expressions depend on $\W$ through the matrix product $\W^T\W$, which can be computed once and cached for all iterations and restarts; we therefore omit this cost from complexity expressions.  It always has size $n \times n$ regardless of the number of queries in $\W$.  For highly structured workloads (e.g all range queries), $\W^T\W$ can be computed directly without materializing $\W$, so we allow $\optgp$ to take $\W^T\W$ as input in these special cases. 

For general $\A$, the runtime complexity of computing $C(\A)$ and $ \frac{\partial C}{\partial \A} $ is $O(n^3)$. By exploiting the special structure of $\A(\matr{\Theta})$, we can reduce these costs to $ O(p n^2) $:

\begin{restatable}[Complexity of $\optgp$]{theorem}{thmoptgp}\label{thm:optgp}
Given any $p$-Identity strategy $\A(\bTheta)$, both the objective function $C(\A(\matr{\Theta}))$
and the gradient $ \frac{\partial C}{\partial \A}$ can be evaluated in $ O(p n^2) $ time.  	
\end{restatable}

The proof of this statement and the corresponding algorithm for efficiently evaluating the objective function is given in the appendix. The speedup resulting from this parameterization is often much greater in practice than the $\frac{n}{p}$ improvement implied by the theorem.  When $n = 8192$, computing the objective for general $\A$ takes $>6$ minutes, while it takes only $1.5$ seconds for a $p$-Identity strategy: a $240\times$ improvement.  
Nevertheless, $\optgp$ is only practical for modest domain sizes ($n \sim 10^4$). For multi-dimensional data we do not use it directly, but as a subroutine in the algorithms to follow.

}

\subsection{Strategy Visualization}

\begin{figure}
\subcaptionbox{ \label{fig:chol} Cholesky strategy}{\includegraphics[width=0.325\textwidth]{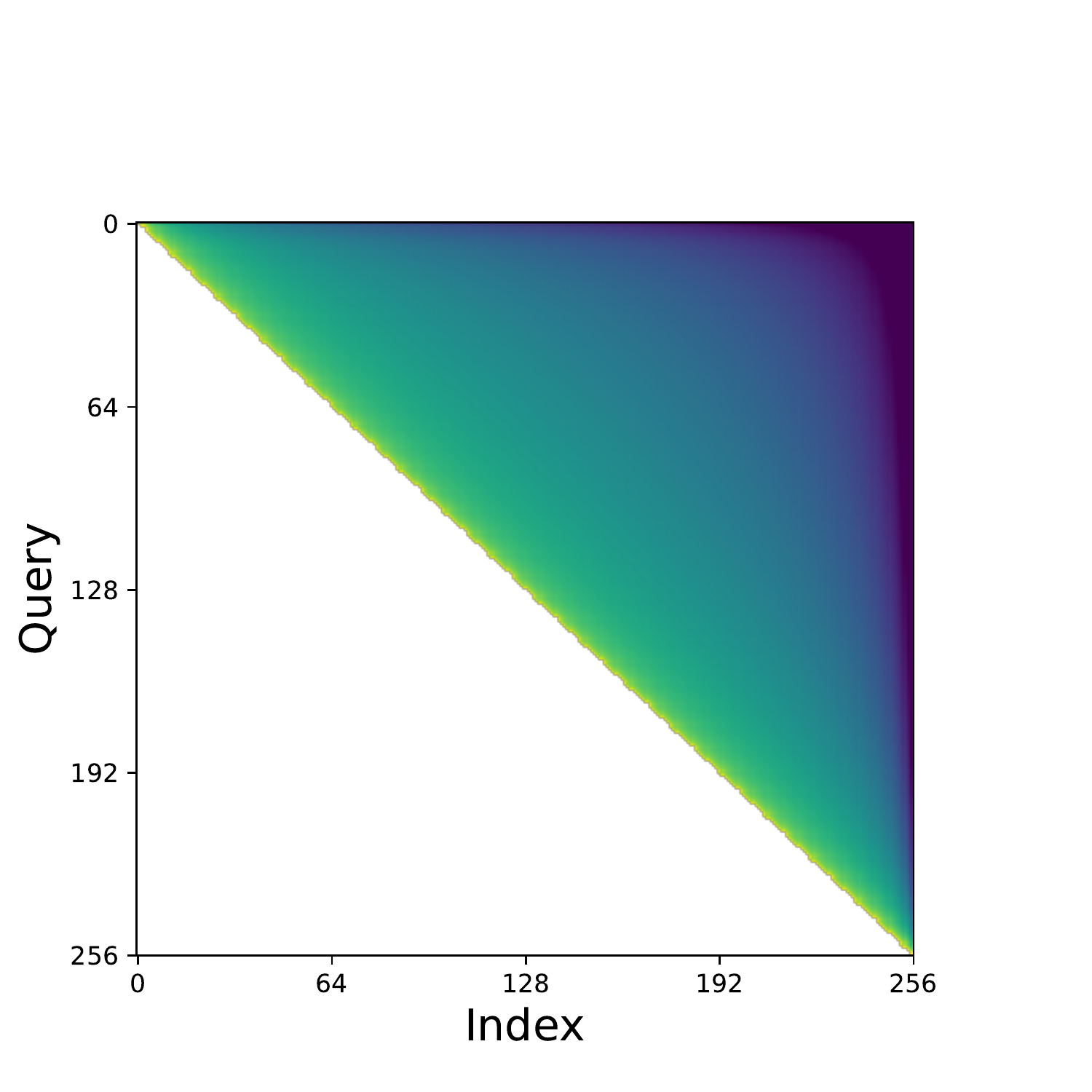}}
\subcaptionbox{ \label{fig:pid} p-Identity strategy \\ (p=16)}{\includegraphics[width=0.325\textwidth]{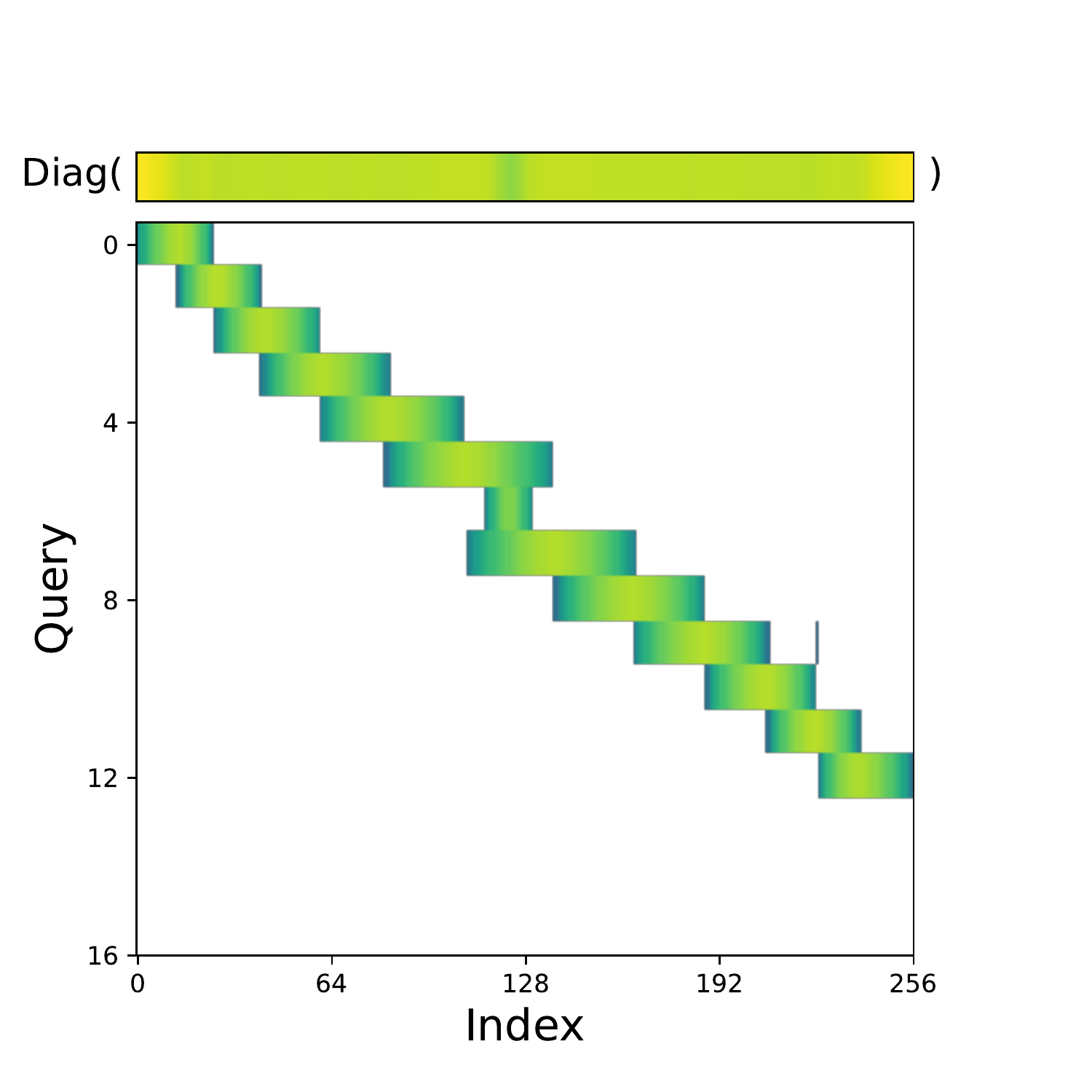}}
\subcaptionbox{ \label{fig:h16} Hierarchical strategy \\ (b=16)}{\includegraphics[width=0.325\textwidth]{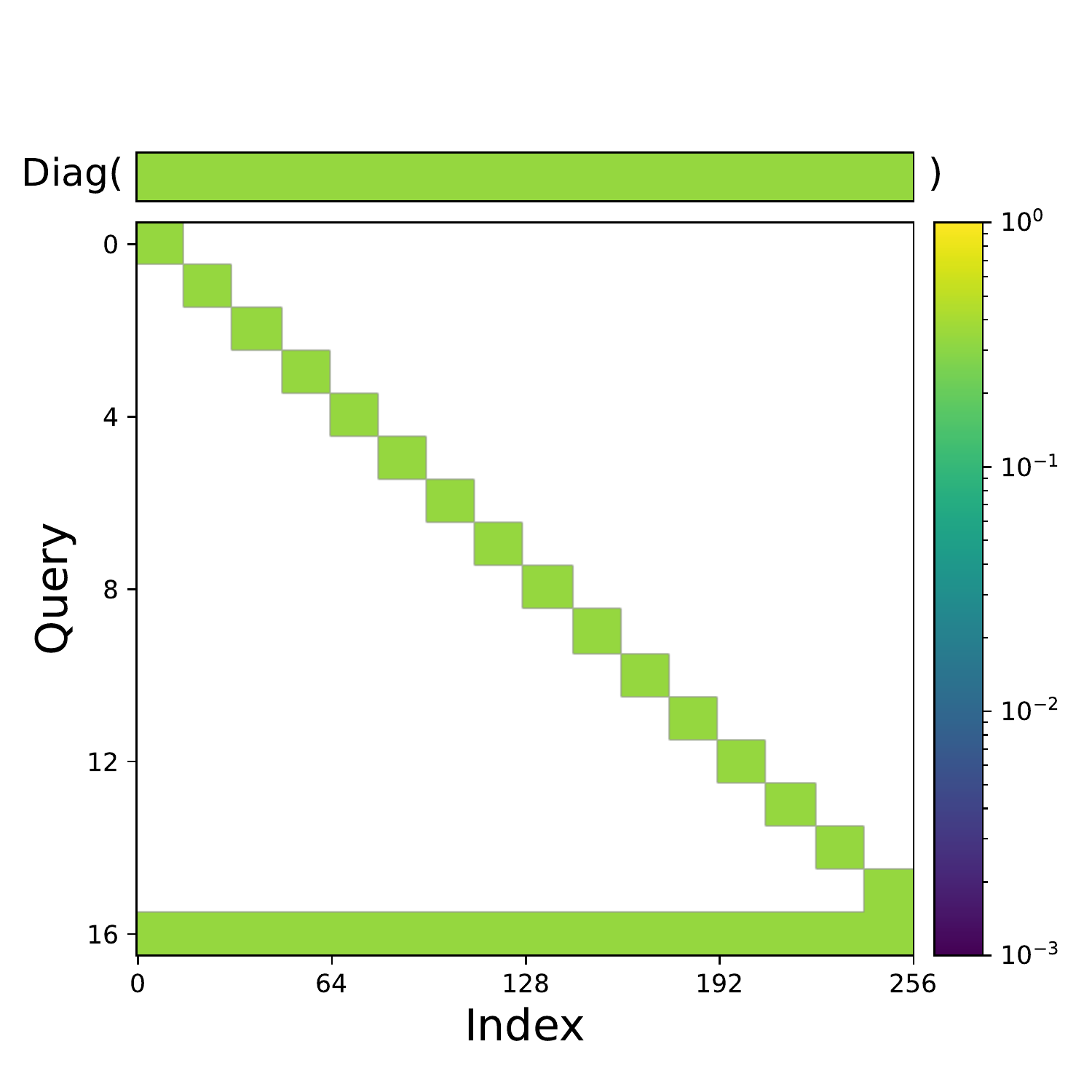}}
\caption{ \label{fig:onedviz} Visualization of three different strategies for answering the workload of All Range queries on a domain of size $256$.  Figures (a) and (b) show strategies optimized by HDMM, and Figure (c) shows a hierarchical strategy with branching factor $16$, a previously state-of-the-art strategy for this workload \cite{hay2010boosting,qardaji2013understanding}.  Each row is a query, and cells are color coded according to their value in the strategy matrix.  Figure (a) plots each query as one very thin row, while Figure (b) and (c) only plot the non-trivial queries as thicker rows.  The diagonal queries are plotted separately as a single row above the main plot, but it actually represents $256$ different queries.
}
\end{figure}

An interested reader may wonder what the optimized strategies actually look like for some common workloads.  In \cref{fig:onedviz}, we plot three strategies designed for the workload of All Range queries.  \cref{fig:chol,fig:pid} show the strategies produced by HDMM for Gaussian and Laplace noise, respectively, and \cref{fig:h16} shows a previously state-of-the-art strategy for range queries.  These visualizations provide interesting insight into the nature of the solution, and reveal why HDMM succeeds in reducing error.

In \cref{fig:chol}, the strategy is an upper triangular square matrix; this is by construction as it is obtained through a Cholesky decomposition.\footnote{There may be equally good strategies that do not have this upper triangular structure, but HDMM will always find one with this structure.}  In each row, weights are largest near the main diagonal, and gradually decrease the further away it gets.  This can be observed from the smooth transition from yellow to green to blue in \cref{fig:chol}.

In \cref{fig:pid}, the strategy was optimized with $p=16$, but only $13$ non-zero queries were found.  Each query has varying width, ranging between about $16$ and $64$, and has greatest weight towards the middle of the query, with gradually decreasing weights away from the center.  In most cases, the queries overlap with half of the two neighboring queries.  The weights on the identity queries are approximately uniform throughout at around $0.5$, with higher weights near the edges.  There is a very natural reason why this structure works well for range queries.  Summing up adjacent queries leads to a bigger query which looks approximately like a range query.  It will have a long uniform center, and decaying weights on the edges.  These decaying weights can be increased to match the uniform center by drawing on the answers from the uniform queries.  Thus, any range query can be answered by summing up a relatively small number of strategy query answers.

This same intuition was used in the derivation of the hierarchical strategy in \cref{fig:h16}.  However, this strategy answers some range queries more effectively than others.  It struggles for queries that require summing up many identity queries.  For example, it can answer the range $[0,16)$ using one strategy query, but it requires summing up $16$ strategy queries to answer the range $[8,24)$.



\section{Implicit representations for conjunctive query sets} \label{sec:implicit}

The optimization methods we described in the previous section work well for small and modest domain sizes; we were able to run them on domains as large as $n=8192$ (see Section \ref{sec:experiments} for a scalability analysis).  However, these methods are fundamentally limited by the need to represent the workload and strategy explicitly in matrix form.  It requires $0.5$ gigabytes just to store a square matrix of size $8192$ using $4$ byte floats, and it is time consuming to perform nontrivial matrix operations on objects of this size.  This limitation is not unique to our mechanism, but is shared by \emph{all possible} methods that optimize explicitly represented workload matrices.  

To overcome this scalability limitation, we propose \emph{implicit query matrices}, which exploit structure in conjunctive query sets and offer a far more concise representation than materialized explicit matrices, while still being able to encode query sets containing an arbitrary collection of conjunctive queries.  These representations are lossless; they save space by avoiding significant redundancy, rather than making approximations.  As we will show later in this section, many important matrix operations can be performed efficiently using the implicit representation.  This property of the representation will be essential for solving the strategy optimization problem efficiently on large domains. 

\subsection{Implicitly vectorized conjunctions} \label{sec:implicit_conjunctions}

Consider a predicate defined on a single attribute, $A$, where $|dom(A)|=n_A$. This predicate, $\phi_{A}$, can be vectorized with respect to just the domain of $A$ (and not the full domain of all attributes) similarly to \cref{def:vect-query}.
When a predicate is formed from the conjunction of such single-attribute predicates, its vectorized form has a concise implicit representation in terms of the {\em outer product} between vectors.  

\begin{definition}[Outer product]
The outer product between two vectors $\q_A$ and $ \q_B $, denoted $ \q_A \otimes \q_B $, is a vector indexed by pairs $t = (t_A, t_B)$ such that: $(\q_A \otimes \q_B)(t) = \q_A(t_A) \q_B(t_B)$.
\end{definition}

The outer product is useful for representing conjunctions in vector form.

\begin{theorem}[Implicit vectorization] \label{thm:implicit-vec}
The vector representation of the conjunction $\phi = \phi_A \wedge \phi_B $ is $ vec(\phi) = vec(\phi_A) \otimes vec(\phi_B) $.
\end{theorem}

To see why this is true, observe that $ \phi(t) = \phi_A(t_A) \phi_B(t_B)$, since multiplication and logical AND are equivalent for binary inputs, and this is the exact equation that defines the outer product for vectors. 
While the explicit representation of $vec(\phi)$ has size $n_A \cdot n_B$, the implicit representation, $vec(\phi_A) \otimes vec(\phi_B)$, requires storing only $vec(\phi_A)$ and $vec(\phi_B)$, which has size $n_A+n_B$.  
%

\begin{example} \label{ex:imp_vec}
Recall that the workload $\Wnt$ consists of 4151 queries, each defined on a data vector of size \num{500480}.  Since explicitly vectorized queries are the same size as the domain, the size of the explicit workload matrix is $4151 \times \num{500480}$, or $8.3$GB.
Using the implicit representation, each query can be encoded using $2+2+64+115+17=200$ values, for a total of 3.3MB.
For $\Wst$, which consists of \num{215852} queries on a data vector of size $\num{25524480}$, the explicit workload matrix would require 22TB of storage. In contrast, the implicit vector representation would require 200MB.
\end{example}

\subsection{Implicitly vectorized products}

Product workloads (as in \cref{def:product-workload} can be encoded even more efficiently using the Kronecker product.

\begin{definition}[Kronecker product]
For two matrices $\A \in \mathbb{R}^{m_A \times n_A}$ and $ \B \in \mathbb{R}^{m_B \times n_B}$, their Kronecker product is $ \A \otimes \B \in \mathbb{R}^{m_A m_B \times n_A n_B} $, where rows are indexed by pairs $q = (q_A, q_B)$ and columns are indexed by pairs $t = (t_A, t_B)$ such that:

$$ (\A \otimes \B)(q, t) = \A(q_A, t_A) \B(q_B, t_B) $$
\end{definition}

The Kronecker product is a generalization of the outer product, where each row is an outer product between a pair of rows from $\A$ and $\B$.  Thus, the same symbol $\otimes$ is used for both operations.

\begin{theorem}[Implicit vectorization]  \label{thm:implicit}
Let $\Phi_A$ and $ \Phi_B $ be two predicate sets defined  on attributes $A$ and $B$ respectively.  Then $ vec(\Phi_A \times \Phi_B) = vec(\Phi_A) \otimes vec(\Phi_B) $.
\end{theorem}

The proof of this claim follows immediately from \cref{def:product-workload}, since it contains a cartesian product of conjunctions, and each conjunction is an outer product.  We implicitly represent a product workload in matrix form by storing the factors of the Kronecker product ($\A = vec(\Phi_A)$ and $ \B = vec(\Phi_B)$).  This requires $m_A n_A + m_B n_B$ space, rather than $m_A m_B n_A n_B$ space, which is required by the explicit representation, and $ m_A m_B (n_A + n_B) $ for the implicit representation using a list of outer products.  Thus, the savings can be quite substantial.  While \cref{thm:implicit} assumes the predicates are conjunctions, we show in \cref{sec:disjuncts} that disjunctions can be handled in a similar manner.

\subsection{Workload encoding algorithm} \label{sec:sub:encoding}



Given as input a logical workload $\Wlog$ (as in \cref{def:logical_workload}), the $\Imp$ algorithm produces an implicitly represented workload matrix with the following form:
\begin{equation} \label{eq:wsum}
\WW =
\begin{bmatrix} w_1\WW_1 \\ \vdots \\ w_k\WW_k
\end{bmatrix} =
\begin{bmatrix} w_1 (\W_1^{(1)} \otimes & \dots & \otimes \W_d^{(1)})  \\ \vdots & \ddots & \vdots \\ w_k (\W_1^{(k)} \otimes & \dots & \otimes \W_d^{(k)})
\end{bmatrix}
\end{equation}
Here stacking sub-workloads is analogous to union and in formulas we will sometimes abuse notation and write an implicit union-of-products workload as $\WW_{[k]} = w_1\WW_1 \plus \dots \plus w_k\WW_k$.  We use blackboard bold font to distinguish an implicitly represented workload $\WW$ from an explicitly represented workload $\W$.
%
%
%
\begin{table}[h]
\begin{tabular}{ll} \hline
\multicolumn{2}{l}{{\bf Algorithm 1}: $\Imp$} \\ \hline
  \multicolumn{2}{l}{ {\bf Input}: Workload $\Wlog=\{q_1 \dots q_k\}$ and weights $w_1 \dots w_k$ }\\
  \multicolumn{2}{l}{ {\bf Output}: Implicit workload $\WW$} \\ \hline
1.  & For each product $q_i \in \Wlog$: $q_i=\Phi^{(i)}_{A_1} \times \dots \times \Phi^{(i)}_{A_d}$ \\
2. & \quad For each $j\in[1..d]$ \\
3. & \quad \quad compute $\W^{(i)}_j = vec(\Phi^{(i)}_{A_j})$ \\
4. & \quad Let $\WW_i = \W_1^{(i)} \otimes  \dots \otimes \W_d^{(i)}$ \\
5. & {\bf Return:} $w_1\WW_1 \plus \dots \plus w_k\WW_k$ \\  \hline
\end{tabular}
\end{table}

Note that line 3 of the $\Imp$ algorithm is \emph{explicit} vectorization, as in \cref{def:vect-query}, of a set of predicates on a single attribute.

\begin{example} \label{ex:implicit_workload_size}
Recall from \cref{ex:product_terms_sf1} that the \num{215852} que\-ries of \Wst can be represented as $k=8032$ products.  If $\Wst$ is represented in this factored form, the $\Imp$ algorithm returns a smaller implicit representation, reducing the 200MB (from \cref{ex:imp_vec}) to 50 MB.  If the workloads are presented in their manually factored format of $k=32$ products, the implicit representation of $\Wnt^*$ requires only 335KB, and $\Wst^*$ only 687KB.
\end{example}

%

\eat{
The SF1 specification \cite{census-sf1} does not include weights, however accuracy is a high priority for a special subset of SF1 queries; we return to this in \cref{sec:experiments}.
}



\subsection{Operations on vectorized objects} \label{sec:sub:kron-identity}
Reducing the size of the workload representation is only useful if critical computations can be performed without expanding them to their explicit representations. Standard properties of the Kronecker product~\cite{van2000ubiquitous} accelerate strategy selection and reconstruction. 

%
The following properties allow us to perform useful operations on Kronecker products without materializing their full matrices.

\begin{proposition}[Kronecker identities] \label{prop:kron} Kronecker products satisfy the following identities ~\cite{van2000ubiquitous}:\\
\begin{tabular}{llc}
Transpose: & $ (\A \otimes \B)^{T} = \A^{T} \otimes \B^{T} $ \\
Pseudo Inverse: & $ (\A \otimes \B)^+ = \A^+ \otimes \B^+ $ \\
Associativity: & $ (\A \otimes \B) \otimes \C = \A \otimes (\B \otimes \C) $  &  \\
Mixed Product: & $ (\A \otimes \B)(\C \otimes \D) = (\A \C) \otimes (\B \D) $ \\
\end{tabular}
\end{proposition}

In addition to the standard properties of Kronecker product above, we can prove additional properties about them which are useful for our privacy mechanism.  

\begin{restatable}[Norm of a Kronecker product]{theorem}{thmkronnorm} \label{thm:kron-sens}
The following matrix norms decompose over the factors of the Kronecker product $\A_1 \otimes \dots \otimes \A_d$.
\begin{align*}
\Lone{\A_1 \otimes \dots \otimes \A_d} &= \prod_{i=1}^d \Lone{\A_i} \\
\Ltwo{\A_1 \otimes \dots \otimes \A_d} &= \prod_{i=1}^d \Ltwo{\A_i} \\
\norm{\A_1 \otimes \dots \otimes \A_d}_F &= \prod_{i=1}^d \norm{\A_i}_F \\
\end{align*}
\end{restatable}

The proof is largely routine algabraic manipulation and is given in the appendix.  We will later use these identities to efficiently evaluate TSE for workloads and strategies built with Kronecker products.

\section{Optimizing conjunctive query workloads with conjunctive query strategies} \label{sec:hd_opt}

We now turn our attention to optimizing implicitly-represented conjunctive query workloads.  We assume the workload is a \emph{union of Kronecker products}, and that it takes the form shown in \cref{eq:wsum} (restated below):

\begin{equation} \tag{\ref{eq:wsum}}
\WW =
\begin{bmatrix} w_1\WW_1 \\ \vdots \\ w_k\WW_k
\end{bmatrix} =
\begin{bmatrix} w_1 (\W_1^{(1)} \otimes & \dots & \otimes \W_d^{(1)})  \\ \vdots & \ddots & \vdots \\ w_k (\W_1^{(k)} \otimes & \dots & \otimes \W_d^{(k)})
\end{bmatrix}
\end{equation}

For notational convenience, we will often write $\WW = w_1 \WW_1 + \dots + w_k \WW_k $, where $ \WW_i = \W_1^{(i)} \otimes \dots \otimes \W_d^{(i)}$ and `$+$' serves the role of stacking matrices.  The methods described in this section will exploit the special structure of this class of workloads to scale more effectively than techniques described in \cref{sec:optimization}.

\subsection{Optimizing product workloads.} \label{sec:sub:optk1}

We begin by considering a special case of the workload in \cref{eq:wsum} that is a \emph{single Kronecker product}. We describe an elegant approach to optimize workloads of this form by analyzing the optimization problem and utilizing properties of Kronecker products.  

\begin{definition}[$\optk$] \label{def:optk}
Given a Kronecker product workload $ \WW = \W_1 \otimes \dots \otimes \W_d $ and an optimization oracle $\optgp$, let $ \optk(\WW) = \A_1 \otimes \dots \otimes \A_d $ where $\A_i = \optgp(\W_i)$.
\end{definition}

Above, $\optgp$ may be any strategy optimization routine that consumes an explicitly represented workload and returns a strategy matrix, such as the techniques discussed in the previous section.  $\optk$ is very efficient: it simply requires solving $d$ small subproblems rather than one large problem.  This decomposition of the objective function has a well-founded theoretical justifiication.  Namely, if we restrict the solution space to a (single) Kronecker product strategy of the form $\AA = \A_1 \otimes \dots \otimes \A_d$, then the error of the workload under $\AA$ decomposes over the factors of the Kronecker products as shown in the following theorem:

\begin{theorem}[Error decomposition]\label{thm:err-decomp}
Given a workload $\WW = \W_1 \otimes \dots \otimes \W_d $ and a strategy $\AA = \A_1 \otimes \dots \otimes \A_d$, the error is:
$$ \Lk{\AA}^2 \norm{ \WW \AA^+ }_F^2 = \prod_{i=1}^d \Lk{\A_i}^2 \norm{\W_i \A_i^+}_F^2 $$
\end{theorem}

The overall error is minimized when $ \A_i $ optimizes $\W_i$ for each $i$, thus it makes sense to optimize each $\A_i$ separately.   If we expect the optimal strategy to be a single Kronecker product, then this approach seems quite appealing.   However it is possible that there exists a strategy that is \emph{not} a single Kronecker product that offers lower error than the best Kronecker product strategy.  The following theorem shows that this is not the case, and gives further justification for the above method, showing that the SVD bound also decomposes over the factors of the Kronecker product.  

\begin{theorem}[SVD bound decomposition]
Given a workload $\WW = \W_1 \otimes \dots \otimes \W_d $, the SVD bound is:
$$ SVDB(\WW) = \prod_{i=1}^d SVDB(\W_i) $$
\end{theorem}

If there exist strategies $\A_1, \dots, \A_d$ that achieve the SVD bound for $\W_1, \dots, \W_d$ and we can find them, then we can construct a Kronecker product strategy $\AA = \A_1 \otimes \dots \otimes \A_d$ that achieves the SVD bound for $\WW = \W_1 \otimes \dots \otimes \W_d$.   Since no other strategy can have lower error than the SVD bound, this implies that \emph{the optimal strategy is a Kronecker product}.  This gives excellent justification for optimizing over the space of Kronecker products, especially when the SVD bound is tight.


\subsection{Optimizing union-of-product workloads.} \label{sec:sub:optk2}

The approach just described is principled and effective when the workload is a single Kronecker product.  We now turn our attention to the more general case where the workload is a union of Kronecker products.  Here, the right approach is less clear. We define three approaches for optimizing implicit workloads in the form of \cref{eq:wsum}.  Each approach restricts the strategy to a different region of the full strategy space for which optimization is tractable. The first computes a strategy consisting of a single product; it generalizes $\optk$.  The second, $\optkk$, can generate strategies consisting of unions of products.  The third, $\optm$, generates a strategy of weighted marginals. The best approach to use will generally depend on the workload, and we will provide some practical guidance and intuition to understand the situations in which each method works best.  

\vspace{1ex}
\noindent {\bf \em Single-product output strategy} \quad
For weighted union of product workloads, if we restrict the optimization problem to a single product strategy, then the objective function decomposes as follows:
\begin{theorem}
Given workload $\WW=w_1\WW_1 \plus \dots \plus w_k\WW_k$ and strategy $\AA = \A_1 \otimes \dots \otimes \A_d$, workload error is:
\begin{equation} \label{eq:unionkron}
\begin{split}
\Lk{\AA}^2 \norm{ \WW \AA^+ }_F^2 &= \Lk{\AA}^2 \sum_{j=1}^k w_j^2 \norm{\WW_j \AA^+}_F^2 \\
&= \sum_{j=1}^k w_j^2 \prod_{i=1}^d \Lk{\A_i}^2 \norm{\W^{(j)}_i \A_i^+}_F^2 
\end{split}
\end{equation}
\end{theorem}
This leads to the following optimization problem:

\begin{definition}[Generalized $\optk$]\label{prob:sum-kron}
Given a workload $\WW=w_1\WW_1 \plus \dots \plus w_k\WW_k$, let $ \optk(\WW) = \A_1 \otimes \dots \otimes \A_d $ where:
\begin{equation*} 
\begin{aligned}
(\A_1, \dots, \A_d) &= \underset{\A_1, \dots, \A_d}{\text{minimize}}
& & \sum_{j=1}^k w_j^2 \prod_{i=1}^d \Lk{\A_i}^2 \norm{\W^{(j)}_i \A_i^+}_F^2 
\end{aligned}
\end{equation*}
\end{definition}

\noindent When $k=1$, the solution to Problem \ref{prob:sum-kron} is given in \cref{def:optk}, so we use matching notation and allow the $\optk$ operator to accept a single product or a union of products.

We can solve this problem efficiently by building on the optimization oracles designed in the previous section.  In particular, suppose we have a black box optimization oracle $\optgp(\W)$ that accepts an explicitly represented workload and gives back an explicitly represented strategy with low (ideally minimal) error on that workload.  Then we use a block method that cyclically optimizes $\A_1, \dots, \A_d$ until convergence.  We begin by initializing $ \A_i = \I $ for all $i$.  We then optimize one $\A_i$ at a time, fixing the other $\A_{i'} \neq \A_i$ using $\optgp$ on a carefully constructed surrogate workload $\hat{\W}$ (\cref{eq:surrogate}) that has the property that the error of any strategy $\A_i$ on $\hat{\W}$ is the same as the error of $\AA$ on $\WW$.  Hence, the correct objective is being minimized.
\begin{equation} \label{eq:surrogate}
\begin{aligned} 
\hat{\W}_i = \begin{bmatrix} c_1 \W_i^{(1)} \\ \vdots \\ c_k \W_i^{(k)} \end{bmatrix}
& &
c_j = w_j \prod_{i' \neq i} \Lk{\A_{i'}} \norm{ \W_{i'}^{(j)} \A_{i'}^+ }_F
\end{aligned}
\end{equation}

The cost of running this optimization procedure is determined by the cost of computing $ \hat{\W}_i^T \hat{\W}_i $ and the cost of optimizing it, which takes $ O(n_i^2 (p_i+k)) $ and $ O(n_i^2 p_i \cdot \text{\sc \#iter}) $ time respectively (assuming each $(\W^T \W)_i^{(j)}$ has been precomputed). As before, this method scales to arbitrarily large domains as long as the domain size of the sub-problems allows $\optgp$ to be efficient.

\vspace{1ex}
\noindent {\bf \em Union-of-products output strategy} \label{sec:sub:unionkron2} \quad
For certain workloads, restricting to solutions consisting of a single product, as $\optk$ does, excludes good strategies, as shown in \cref{example:optkk}.

\begin{example} \label{example:optkk}
Consider the workload $ \WW = \WW_1 + \WW_2$ where $ \WW_1 = \P \otimes \T$ and $\WW_2 = \T \otimes \P$ on a 2-dimensional domain of size $100 \times 100$.  Running $\optk$ on this workload leads to an optimized strategy of the form $ \AA = \A_1 \otimes \A_2 $.  The expected error of this strategy is $33385$, which is much higher than it should be for such a simple workload.  The poor expected error can be explained by the fact that to support the workload, both $\A_1$ and $\A_2$ have to be full rank.  This means $\AA$ has to include at least $100^2$ queries, even though $\WW$ only contains $200$ queries.

A better alternative would be to optimize $\WW_1$ and $\WW_2$ separately using $\optk$.  Doing this we would end up with a strategy $\AA = \AA_1 + \AA_2$, where $\AA_1$ optimizes $\WW_1$ and $\AA_2$ optimizes $\WW_2$.  The resulting strategy is much smaller because $\optk(\AA_1) = \optgp(\P) \otimes \optgp(\T)$, and $\optgp(\T) = \T$.  In fact, it only contains $212$ queries and attains an expected error of $14252$, which is a $2.34 \times$ improvement.
\end{example}

Based on this example, we would like a principled approach to optimize over the space of strategies that are a union of Kronecker products.  Unfortunately, computing the workload error exactly for a strategy of this form is intractable, as the pseudo inverse may not be a union of Kronecker products.  This makes optimization over this space of strategies challenging.  We thus propose the following heuristic optimization routine inspired by \cref{example:optkk}.  This optimization routine individually optimizes each subworkload $\WW_j$ using $\optk$, and then combines the strategies all together to form a single strategy.  It simply requries calling $\optk$ a number of times and computing appropriate weights for each optimized strategy. 

\begin{definition}[$\optkk$] \label{def:optp} Given a workload $ \WW = w_1 \WW_1 + \dots + w_k \WW_k $, let $ \optkk(\WW) = c_1 \AA_1 + \dots + c_k \AA_k $ where $\AA_j = \optk(\WW_j) $
and $$c_j \propto \frac{1}{\Lk{\AA_j}} \begin{cases} \sqrt[3]{2 E_j} & \text{ if } \algG = \LM \\ \sqrt[4]{E_j} & \text{ if } \algG = \GM \end{cases}$$
for $E_j = w_j^2 \Lk{\AA_j}^2 \norm{\WW_j \AA_j^+}_F^2$.
\end{definition}

Above, we assume that $ \AA_j $ will be used to answer $ \WW_j $, and $c_j$ is the weight on $\AA_j$: it corresponds the portion of the privacy budget that will be spent to answer those queries.  It is chosen to minimize total workload error.
Specifically, if we allocate $c_j$ budget to answer $\AA_j$, then the error will be $ E_j / c_j^2 $.  Thus, the choice of $c_j$ above is based on minimizing $ \sum E_j / c_j^2 $ subject to the constraint $ \sum | c_j | = 1 $ (for Laplace noise) and $\sum c_j^2 = 1$ (for Gaussian noise).  We solve this minimization problem exactly, in closed from, using the method of Lagrange multipliers.

We remark that in the definition above, $\WW$ is split up into $k$ sub-workloads $\WW_1, \dots, \WW_k$.  Each subworkload $\WW_j$ is assumed to be a single Kronecker product, but the optimization routine is still well defined even if $\WW_j$ is a union of Kronecker products.   This opens up a nice opportunity: to group the subworkloads into clusters which will be optimized together with $\optk$.   Intuitively, if two subworkloads are similar, it may make sense to group them together to optimize collectively.  We do not provide an automated way to group subworkloads in this paper.  This is a hard problem in general, and is out of scope for this paper.  A domain expert can work out good clusterings on a case-by-case basis, or they can settle for the default clustering (one Kronecker product per cluster).  

\section{Implicit representations for marginal query sets} \label{sec:impmarg}

In the previous section we described $\optk$ and $\optkk$, two methods for optimizing implicitly represented conjunctive query workloads.  These methods differ primarily in the space of strategies they search over.  Our final optimization method, $\optm$, optimizes over the space of marginal query matrices, and offers a preferable alternative to $\optk$ and $\optkk$ in some settings.

In this section, we propose an implicit representation for marginal query sets that is even more compact than our other representation for general conjunctive query sets.  We further show that these matrices can be operated on efficiently, allowing us to solve the strategy optimization problem for large multi-dimensional domains.


A marginal query matrix is a special case of a union of Kronecker products, where each Kronecker product encodes the queries to compute a single marginal (i.e., all the factors are either Identity or Total).  First note that a marginal on a $d$-dimensional domain can be specified by a subset of elements of $\set{1,\dots,d}$.  Hence, there are a total of $2^d$ possible marginals, and each one can be specified by an element of the set $[2^d] = \set{0, \dots, 2^d-1}$.   The most natural correspondence between these integers and the associated marginals is based on the binary representation of the integer.  The query set required to compute the $a^{th}$ marginal would be represented by $\Q_1 \otimes \dots \otimes \Q_d $ where $\Q_i = \I $ if the $i^{th}$ bit of the binary representation of $a$ is $1$ and $ \Q_i = \T $ otherwise.  A collection of weighted marginals can thus be represented as a vector $ \u $ containing a weight for each marginal.  We refer to this marginal query matrix as $\MM(\u)$, which is defined below.  

\begin{definition}[Marginal query matrix]
A marginal query matrix $\MM(\u)$ is defined by a vector of weights $\u \in \mathbb{R}^{2^d}$ and is a special case of the query matrix shown in \cref{eq:wsum} where $k=2^d$, $w_{a+1} = \u(a)$, and
$$ \W^{(a+1)}_i = \begin{cases} \I & a_i = 1 \\ \T & a_i = 0 \end{cases} $$
where $ a \in \set{0, \dots, 2^d-1}$,$i \in \set{1, \dots, d}$,and $a_i$ is the $i^{th}$ bit of the binary representation of $a$. 
\end{definition}

For a marginal query matrix, the weight $\u(a)$ can be interpreted as the relative \emph{importance} of the $a^{th}$ marginal.  As shown in \cref{prop:marginal_sensitivity}, it is particularly simple to reason about the sensitivity of a marginal query matrix.  

\begin{proposition} \label{prop:marginal_sensitivity}
The sensitivity of a marginal query matrix $ \MM(\u) $ is:
\begin{align*}
\Lone{\MM(\u)} = \norm{\u}_1 & & \Ltwo{\MM(\u)} = \norm{\u}_2
\end{align*}
\end{proposition}

Moving forward, it is convenient to work with the Gram matrix representation of the marginal query matrix instead.  As shown below, there is a simple correspondence between the two: 

\begin{restatable}[Marginal Gram matrix]{proposition}{propmarggram} \label{prop:marggram}
Let $\QQ = \MM(\u)$ be a marginal query matrix.  Then the correpsonding marginal Gram matrix is $\QQ^T \QQ = \GG(\u^2)$, where:
$$ \GG(\v) = \sum_{a=0}^{2^d-1} \v(a) \CC(a), \;\;\;\;\;\;\;\;\;\; \CC(a) = \bigotimes_{i=1}^d [\matr{1} (a_i = 0) + \I (a_i = 1)] $$
\end{restatable}

In the proposition above, the term $\matr{1} (a_i=0) + \I(a_i=1)$ is shorthand notation for $ \begin{cases} \matr{1} & a_i = 0 \\ \I & a_i = 1 \end{cases} $.  $\matr{1}$ and $\I$ are both $n_i \times n_i$ matrices, corresponding to the matrix of all ones, and the identity matrix respectively.  We will use this notation frequently in the section, so it is important to understand the exact meaning.  Another important object that will appear repeatedly throughout this section is the so-called characteristic vector \footnote{this is not to be confused with the eigenvector}, which is defined below:

\begin{definition}[Characteristic Vector]
The characteristic vector $\c \in \mathbb{R}^{2^d}$ is defined so that each entry $\c(a)$ equals the number of entries in the $(\neg a)^{th}$ marginal 
$$ \c(a) = \prod_{i=1}^d n_i (a_i = 0) + 1 (a_i = 1) $$
\end{definition}

The term $ \neg a $ in the definition above is the bitwise negation of $a$, and it is obtained by flipping each of the $d$ bits of the integer $a$.  We will rely heavily on this type of bitwise manipulation in this section to reason about the behavior of marginal Gram matrices.  

\begin{example}
For a 3 dimensional domain, the query matrix for a 2 way marginal can be expressed as $ \QQ = \I \otimes \T \otimes \I $, and $\QQ^T \QQ = \I \otimes \matr{1} \otimes \I = \CC(101_2) = \CC(5) $.
\end{example}

Now that we have introduced the necessary notation for marginal query and Gram matrices, we are ready to show how to perform important matrix operations while respecting the implicit representation.  We begin with \cref{thm:margmult}, which shows that marginal Gram matrices interact nicely under matrix multiplication.

\begin{restatable}[Multiplication of Marginal Gram Matrices]{theorem}{thmmargmult} \label{thm:margmult}
For any $a, b \in [2^d] $, $$\CC(a) \CC(b) = \c(a | b) \CC(a \& b)$$ where $ a | b $ denotes ``bitwise or'', $ a \& b $ denotes ``bitwise and'', and $\c$ is the characteristic vector.  
Moreover, for any $\u, \v \in \mathbb{R}^{2^d}$
$$\GG(\u) \GG(\v) = \GG(\X(\u) \v)$$
where $\X(\u)$ is a $2^d \times 2^d$ triangular matrix with entries $\X(\u)(k,b) = \sum_{a : a \& b = k} \u(a) \c(a | b)$.
\end{restatable}

\cref{thm:margmult} allows us to efficiently multiply two matrices while maintaining the compact implicit representation.  Additionally, it follows immediately from the proof of \cref{thm:margmult} that $\GG(\u) \GG(\v) = \GG(\v) \GG(\u)$ --- i.e., matrix multiplication is commutative.  We can apply \cref{thm:margmult} to efficiently find the inverse or generalized inverse of $\GG(\u)$ as well.

\begin{restatable}[Inverse of Marginal Gram Matrices]{theorem}{thmmarginv} \label{thm:marginv}
Let $\X(\u)$ be the matrix defined in \cref{thm:margmult}.  If $\X(\u)$ is invertible, then $\GG(\u)$ is invertible with inverse:
$$\GG^{-1}(\u) = \GG(\X^{-1}(\u) \z) $$
where $\z(2^d-1) = 1$ and $\z(a) = 0$ for all other $a$.
Moreover, if $\X^g(\u)$ is a generalized inverse of $\X(\u)$, then a generalized inverse of $\GG(\u)$ is given by:
$$ \GG^g(\u) = \GG(\X^g(\u) \X^g(\u) \u) $$.
\end{restatable}
Because $ \X(\u) $ is a triangular matrix, we can compute $\X^{-1}(\u) \z$ efficiently in $O(4^d)$ time using back-substitution (quadratic in the size of $\z$).  Note that $\GG(\u)$ and $\X(\u)$ are invertible if and only if $\u(2^d - 1) > 0$.  The generalized inverse result holds even for non-invertible matrices.  This result is slightly more complicated, but is important because we generally expect $\GG$ to be singular (e.g., if it is the Gram matrix of some workload of low-dimensional marginal query matrices).  

As we show in \cref{thm:eigmarg}, we know the eigenvectors and eigenvalues of marginal Gram matrices.  Recall that $\v$ is an eigenvector with corresponding eigenvalue $\lambda$ if $ \GG(\w) \v = \lambda \v $ for some real-valued $\lambda$.  We use the term \emph{eigenmatrix} to refer to a matrix where each column is an eigenvector that shares the same eigenvalue.

\begin{restatable}[Eigenvectors and Eigenvalues of Marginal Gram Matrices]{theorem}{thmmargeig} \label{thm:eigmarg}
Let $ a \in [2^d] $ and let 
$$\VV(a) = \bigotimes_{i=1}^d (a_i = 0) \T + (a_i = 1) (\matr{1} - n_i \I)$$
For any $b \in [2^d]$, $\VV(a)$ is an eigenmatrix of $\CC(b)$ with corresponding eigenvalue $\blambda(a) = \c(b) $ if $a \& b = a$ and $\blambda(a) = 0$ otherwise.  Moreover, for any $\w \in \mathbb{R}^{2^d}$, $\VV(a)$ is an eigenmatrix of $\GG(\w)$ with corresponding eigenvalue $\bkappa(a) = \sum_{b : a\&b=a} \w(b) \c(b) $.  That is,
\begin{align*} 
\CC(b) \VV(a) = \blambda(a) \VV(a) 
& & &
\GG(\w) \VV(a) = \bkappa(a) \VV(a)
\end{align*}
\end{restatable}

\cref{thm:eigmarg} is remarkable because it states that the eigenmatrices (and hence the eigenvectors) are the same for all marginal Gram matrices $\GG(\w)$. Furthermore, the corresponding eigenvalues have a very simple (linear) dependence on the weights $\w$.  In fact, there is a triangular matrix $\Y$ such that $ \bkappa = \Y \w$.

\section{Optimizing conjunctive query workloads with marginal query strategies} \label{sec:margopt}

In this section, we describe $\optm$, an optimization operator that consumes a conjunctive query workload $\WW$ and returns a marginal query strategy $\AA = \MM(\btheta)$.\footnote{We reserve the symbol $\btheta$ for \emph{strategies}, and use $\u, \v$ and $\w$ to refer to other marginal Gram matrices.}  \cref{thm:margapprox} is the first key to our approach for this problem.  Intuitively, it states that for any conjunctive query workload $\WW$, there is a marginal Gram matrix $\GG(\w)$ that is equivalent to $\WW^T \WW$ for the purposes of optimization.

\begin{restatable}[Marginal approximation of conjunctive query workload]{theorem}{thmmargapprox} \label{thm:margapprox}
For any conjunctive query workload $ \WW = w_1 \WW_1 + \dots + w_k \WW_k $, there is a marginal Gram matrix $ \GG(\w) $ such that $ tr[\GG(\u) \WW^T \WW] = tr[\GG(\u) \GG(\w)] $ for all $\u$.
\end{restatable}

$\GG(\u)$ in \cref{thm:margapprox} represents $(\AA^T \AA)^+$ in the expected error formula.  We know this is a marginal Gram matrix by \cref{prop:marggram} and \cref{thm:marginv}.  \cref{thm:margapprox} allows us to reduce the problem of optimizing an arbitrary conjunctive query workload to simply optimizing a marginal query workload, which we can do efficiently.  In fact, as we show in \cref{thm:margobj}, we can efficiently evaluate the matrix mechanism objective for marginal query strategies, which is essential for efficient optimization.   

\begin{restatable}[Marginal parameterization objective function]{theorem}{thmmargobj} \label{thm:margobj}
Let $\WW = w_1 \WW_1 + \dots + w_k \WW_k$ be a conjunctive query workload and let
$\GG(\w)$ be the marginal approximation of $\WW^T \WW$ (as in \cref{thm:margapprox}).  For any marginal query strategy $\AA = \MM(\btheta)$, the matrix mechanism objective function can be expressed as: 
$$ \norm{\AA}_{\algG}^2 \norm{\WW \AA^+}_F^2 = \norm{\btheta}_{\algG}^2 [ \vect{1}^T \X^+(\btheta^2) \w] $$
where $\norm{\btheta}_{\algG}$ is the sensitivity norm defined in \cref{prop:marginal_sensitivity}, and $\X$ is the matrix defined in \cref{thm:margmult}. 
\end{restatable}

\cref{thm:margobj} shows that we can efficiently calculate the objective function in terms of $\w$ and $\btheta$, without ever explicitly materializing $ \GG(\w)$ or $\MM(\btheta)$.  This key idea will allow us to solve the strategy selection problem efficiently.  Problem \ref{prob:marginals} states the main optimization problem that underlies $\optm$, which immediately follows from \cref{thm:margobj}.  

\begin{definition}[Marginals parameterization] \label{prob:marginals}
Given a conjunctive query workload $\WW = w_1 \WW_1 + \dots + w_k \WW_k $, let $\optm(\WW) = \MM(\btheta^*) $ where
\begin{equation*}
\begin{aligned}
\btheta^* =& \argmin_{\btheta}
& & \norm{\btheta}^2_{\algG} [ \vect{1}^T \X^+(\btheta^2) \w]  \\
& \text{subject to} & & \X^+(\btheta^2) \X(\btheta^2) \w = \w
\end{aligned}
\end{equation*}
and $\GG(\w)$ is the marginal approximation of $\WW^T \WW$ (as in \cref{thm:margapprox}).
\end{definition}

Above, the constraint ensures that the strategy supports the workload.  In practice, this constraint can usually be ignored, and the resulting unconstrained optimization problem can be solved instead.  The constraint can then be verified to hold at the end of the optimization.  Intuitively, this is because strategies that move closer to the boundary of the constraint will have higher error, so the optimization will never approach it as long as sufficiently small step sizes are taken.  We use scipy.optimize to solve this problem in practice.  

We note that the number of parameters in the above optimization problem is $2^d$ and that we can evaluate the objective in $O(4^d)$ time (quadratic in the number of parameters).  Thus, it is feasible to solve this problem as long as $ d \leq 15 $.  Importantly, this means that the runtime complexity is independent of the domain size of each attribute, so it will take the same amount of time for $n_i = 2$ (binary features), $n_i = 10$, or any other values of $n_i$.  

In addition to being able to efficiently optimize over the space of marginal query strategies, we can also efficiently compute the SVD bound for marginal query workloads.  \cref{thm:margsvdb} gives a remarkably simple formula for computing the SVD bound for marginal query workloads.

\begin{restatable}[SVD Bound for Marginal Query Workloads]{theorem}{thmmargsvdb} \label{thm:margsvdb}
The SVD bound for a marginal query workload $\WW$ with Gram matrix $\GG(\w)$ is:

$$ SVDB(\WW) = \frac{1}{n} \Big( \sum_a \c(\neg a) \sqrt{\sum_{b : a \& b = a} \w(b) \c(b)} \Big)^2 $$
\end{restatable}

Additionally, as a byproduct of this analysis, we give a similarly simple formula to \emph{find the optimal marginal query strategy in closed form} in \cref{thm:optsvdbmarg}, allowing us to bypass the need for numerical optimization in some settings.  

\begin{restatable}[Closed form solution to Problem \ref{prob:marginals}]{theorem}{thmoptsvdbmarg} \label{thm:optsvdbmarg}
Let $\WW$ be a workload with Gram matrix $\GG(\w)$ and let $\btheta = \sqrt{\Y^{-1} \sqrt{\Y \w}}$, where $\Y$ is the $2^d \times 2^d$ matrix with entries:

$$ \Y(a,b) = \begin{cases}
\c(b) & a\&b=a \\
0 & otherwise \\
\end{cases} $$

If $\btheta$ contains real-valued entries then the strategy $\AA = \MM(\btheta)$ attains the SVDB bound when $\algG = \GM$, and is thus an optimal strategy.  That is, $ \norm{\AA}^2_{\GM} \norm{\WW \AA^+}_F^2 = SVDB(\WW)$.
\end{restatable}

While $\btheta$ may sometimes contain imaginary entries, we can always fall back on numerical optimization to solve Problem \ref{prob:marginals}.  The formula in \cref{thm:optsvdbmarg} can still be used to initialize the optimization if the imaginary entries of $\btheta$ are replaced with zeros.  Li and Miklau derived sufficient conditions for the SVD bound to be realizable \cite{li13optimal}, and marginal query workloads satisfy those sufficient conditions.  This implies that the SVD bound should always be attainable for workloads of this form.  If the parameters in \cref{thm:optsvdbmarg} contain imaginary entries, this suggests that the optimal strategy that is \emph{not} a marginal query strategy.  It is an interesting open question to determine what the structure of the optimal strategy is when \cref{thm:optsvdbmarg} does not apply.  In practice, even when the SVD bound is not attained exactly by $\optm$, we get very close to it for marginal query workloads, as we show empirically in \cref{table:hd2} of the experiments.

\section{The $\metaopt$ strategy selection algorithm} \label{sec:sub:prog_space}

\begin{table*}[h]
\resizebox{\textwidth}{!}{
\begin{tabular}{ll|c|c|c|c}
\multicolumn{2}{c}{\bf Definition} &
{\bf Operator} &
\multicolumn{1}{c}{{\bf Input workload}} &
\multicolumn{1}{c}{{\bf Output strategy}} &
{\bf Complexity} \\  \hline\hline
\textsection \ref{sec:optimization} & \cref{prob:convex,prob:gp} &
$\optgp$ &
Explicit matrix $\W$ &
Explicit matrix $\A$ &
$O(n^3)$ \\ \hline

\textsection \ref{sec:sub:optk1} & \cref{def:optk}  &
$\optk$ &
Kronecker Product &
Kronecker product  &
$O(\sum_{i=1}^d n_i^3)$ \\ \hline

\textsection \ref{sec:sub:optk2} & \cref{prob:sum-kron}  &
$\optk$ &
Union of Kronecker Products &
Kronecker Product &
$O(k \sum_{i=1}^d n_i^3)$ \\ \hline

\textsection \ref{sec:sub:optk2} & \cref{def:optp}  &
$\optkk$ &
Union of Kronecker Products &
Union of Kronecker Products &
$O(k \sum_{i=1}^d n_i^3)$ \\ \hline

\textsection \ref{sec:margopt} & \cref{prob:marginals}  &
$\optm$ &
Union of Kronecker Products &
Marginal Query Strategy &
$O(4^d)$ \\ \hline \hline
\end{tabular} }
\caption{\label{tbl:opt-summary} Summary of optimization operators: input and output types, and the time complexity of objective/gradient calculations.}
\end{table*}

In this paper, we proposed four optimization routines: $\optgp$, $\optk$, $\optkk$, and $\optm$.  In this section, we summarize these different approaches, discuss the pros and cons of each one, and propose a meta-optimization algorithm $\metaopt$ that automatically chooses the best one based on the workload.  \cref{tbl:opt-summary} summarizes the basic inputs and outputs of each operator.  $\optgp$ is designed to optimize an explicitly represented workload, and returns an explicitly represented strategy.  The other optimization operators all operate in an implicit space however.  

The time complexity of $\optgp$ is $O(n^3)$ (where $n$ is the domain size), and it generally feasible to run as long as $n \leq 10^4$.   The time complexity of $\optk$ and $\optkk$ is $O(k \sum n_i^3)$, where $k$ is the number of union terms in the workload, and $n_i$ is the domain size of attribute $i$.  It is generally feasible to run as long as $\optgp$ is feasible on each of the individual attributes (i.e., $n_i \leq 10^4$).  In contrast to $\optgp$, the total domain size for these operators can be arbitrarily large.  The time complexity of $\optm$ is $O(4^d)$, which interestingly does not depend on the domain size of individual attributes, only the number of attributes.  It is generally feasible to run as long as $d \leq 15$.  

Each of the operators searches over a different space of strategies, and the best one to use will ultimately depend on the workload.  We illustrate the behavior of each optimization operator on the simple workload of all 2-way marginals in \cref{example:margexample}.  This example highlights and summarizes the key differences between $\optk$, $\optkk$, and $\optm$.  In this case, $\optm$ is the best, which is not surprising because it is the most suitable for marginal workloads.  However, in general predicting which optimization operator will yield the lowest error strategy requires domain expertise and may be challenging for complex workloads.  
Since strategy selection is independent of the input data and does not consume the privacy budget, we can just run each optimization operator, keeping the output strategy that offers the smallest expected error.  Additionally, since the strategies found by each optimization operator may depend on the initialization, we recommend running several random restarts of each optimization operator, returning the best one.

By default, $\metaopt$ invokes all three high-dimensional optimization operators $\optk$, $\optkk$, and $\optm$.  ($\optgp$ may also be included for lower-dimensional workloads).  For $\optk$ and $\optkk$ invoked with the p-Identity strategy we use the following convention for setting the $p$ parameters: if an attribute's subworkload is completely defined in terms of $\T$ and $\I$, we set $p=1$ (this is a fairly common case where more expressive strategies do not help), otherwise we set $p=n_i/16$ for each attribute $A_i$ with size $n_i$.

\newpage
\begin{example}[Optimizing Marginal Query Workload] \label{example:margexample}
Consider the workload containing queries to compute all 2-way marginals on a domain of size $(2,5,50,100)$.  This workload can be represented as a union of $\binom{4}{2}=6$ Kronecker products.  \cref{table:margexample} gives the precise representation of this workload, together with the optimized strategies found by $\optk$, $\optkk$, and $\optm$.  
All optimized strategies can be expressed in terms of the ``Identity'' ($\I$) and ``Total'' ($\T$) building blocks.  $\optk$ tries to find the best single Kronecker product strategy, while $\optkk$ tries to find the optimal weight to assign to each of the six Kronecker products that make up the workload.  $\optm$ identifies a different set of marginal queries from which all 2-way marginals can be derived.  Among the three optimization operators, $\optm$ is the best, followed by $\optkk$ and then $\optk$.  $\optm$ offers a $4.8\times$ reduction in Expected TSE over the Identity baseline, and a $3.3\times$ reduction over the Workload baseline.  
\end{example}

\begin{table}[h]
\begin{tabular}{c|c
p{0.85cm}@{$\otimes$}>{\centering\arraybackslash}p{1.5cm}@{$\otimes$}>{\centering\arraybackslash}p{1.5cm}@{$\otimes$}>{\centering\arraybackslash}p{1.5cm}|c}
&\multicolumn{5}{c}{Query Matrix} & Expected TSE\\\hline
&\multicolumn{5}{c|}{}& \\
\multirow{6}{*}{$\WW$}
&& $\textcolor{Purple}{\T}$ & $\textcolor{Purple}{\T}$ & $\textcolor{Red}{\I}$ & $\textcolor{Red}{\I}$&\multirow{6}{*}{$206,964$} \\
&& $\textcolor{Purple}{\T}$ & $\textcolor{Red}{\I}$ & $\textcolor{Purple}{\T}$ & $\textcolor{Red}{\I}$& \\
&& $\textcolor{Purple}{\T}$ & $\textcolor{Red}{\I}$ & $\textcolor{Red}{\I}$ & $\textcolor{Purple}{\T}$& \\
&& $\textcolor{Red}{\I}$ & $\textcolor{Purple}{\T}$ & $\textcolor{Purple}{\T}$ & $\textcolor{Red}{\I}$& \\
&& $\textcolor{Red}{\I}$ & $\textcolor{Purple}{\T}$ & $\textcolor{Red}{\I}$ & $\textcolor{Purple}{\T}$& \\
&& $\textcolor{Red}{\I}$ & $\textcolor{Red}{\I}$ & $\textcolor{Purple}{\T}$ & $\textcolor{Purple}{\T}$& \\
&\multicolumn{5}{c|}{}& \\\hline
&\multicolumn{5}{c|}{}& \\
$\mathbb{I}$ & & $\textcolor{Red}{\I}$ & $\textcolor{Red}{\I}$ & $\textcolor{Red}{\I}$ & $\textcolor{Red}{\I}$&$300,000$ \\
&\multicolumn{5}{c|}{}& \\\hline
&\multicolumn{5}{c|}{}& \\
$\optk(\WW)$ && $\textcolor{Red}{\I}$ & $\textcolor{Red}{\I}$ & $\begin{bmatrix} 0.80\: \textcolor{Red}{\I} \\ 0.20 \: \textcolor{Purple}{\T} \end{bmatrix}$ & $\begin{bmatrix} 0.82\: \textcolor{Red}{\I} \\ 0.18\: \textcolor{Purple}{\T} \end{bmatrix}$& $213,270$\\
&\multicolumn{5}{c|}{}& \\\hline
&\multicolumn{5}{c|}{}& \\
\multirow{6}{*}{$\optkk(\WW)$}
& 0.39 & $\textcolor{Purple}{\T}$ & $\textcolor{Purple}{\T}$ & $\textcolor{Red}{\I}$ & $\textcolor{Red}{\I}$&\multirow{6}{*}{$85,070$} \\
& 0.18 & $\textcolor{Purple}{\T}$ & $\textcolor{Red}{\I}$ & $\textcolor{Purple}{\T}$ & $\textcolor{Red}{\I}$& \\
& 0.14 & $\textcolor{Purple}{\T}$ & $\textcolor{Red}{\I}$ & $\textcolor{Red}{\I}$ & $\textcolor{Purple}{\T}$& \\
& 0.13 & $\textcolor{Red}{\I}$ & $\textcolor{Purple}{\T}$ & $\textcolor{Purple}{\T}$ & $\textcolor{Red}{\I}$& \\
& 0.11 & $\textcolor{Red}{\I}$ & $\textcolor{Purple}{\T}$ & $\textcolor{Red}{\I}$ & $\textcolor{Purple}{\T}$& \\
& 0.05 & $\textcolor{Red}{\I}$ & $\textcolor{Red}{\I}$ & $\textcolor{Purple}{\T}$ & $\textcolor{Purple}{\T}$& \\
&\multicolumn{5}{c|}{}& \\\hline
&\multicolumn{5}{c|}{}& \\
\multirow{3}{*}{$\optm(\WW)$}
& 0.44 & $\textcolor{Purple}{\T}$ & $\textcolor{Purple}{\T}$ & $\textcolor{Red}{\I}$ & $\textcolor{Red}{\I}$&\multirow{3}{*}{62,886} \\
& 0.31 & $\textcolor{Red}{\I}$ & $\textcolor{Red}{\I}$ & $\textcolor{Purple}{\T}$ & $\textcolor{Red}{\I}$& \\
& 0.25 & $\textcolor{Red}{\I}$ & $\textcolor{Red}{\I}$ & $\textcolor{Red}{\I}$ & $\textcolor{Purple}{\T}$& \\
&\multicolumn{5}{c|}{}& \\\hline
\end{tabular}
\caption{
\label{table:margexample}
A workload containing queries to compute \emph{all 2-way marginals} on a four-dimensional domain of size $(2,5,50,100)$.  The optimal strategy found by each parameterization, and the respective error, is also shown.  The Identity and Total query matrices $\textcolor{Red}{\I}$ and $\textcolor{Purple}{\T}$ are color coded for readability.}
\end{table}

\section{Efficient \measure and \reconstruct} \label{sec:running}

Now that we have fully described how HDMM solves the strategy selection problem, we are ready to discuss how to run the remainder of the mechanism.  Recall from \cref{prop:matrixmech} that the matrix mechanism is defined as $\mathcal{M}_{\A,\algG}(\W,\x) = \W \A^+ \algG(\A, \x)$.  With explicitly represented matrices, this computation can be done directly without problem.  However, HDMM replaces the explicitly represented matrices $\W$ and $\A$ with implicitly represented ones $\WW$ and $\AA$, and it is no longer obvious how to run the mechanism.  Conceptually, these steps can be broken down as follows.  In the \measure step, we have to compute the noisy strategy query answers,  $\y = \AA \x + \bm{\xi}$.  In the \reconstruct step, we have to estimate the data vector and workload query ansewrs, i.e., compute $\hat{\x} = \AA^+ \y$ and return $\WW \hat{\x}$.  A necessary key subroutine to solve these problems is to compute matrix-vector products where the matrix is a Kronecker product.  Importantly, we must do this without ever materializing $\AA$ explicitly, as that is infeasible for large domains.   

\begin{restatable}[Efficient matrix-vector multiplication]{theorem}{thmmatvec} \label{thm:matvec}
Let $ \AA = \A_1 \otimes \dots \otimes \A_d $ and let $\x$ be a data vector of compatible shape.  Then \cref{alg:fastikron} computes the matrix-vector product $\AA \x$.  Furthermore, if $\A_i \in \mathbb{R}^{n_i \times n_i}$ and $ n = \prod n_i $ is the size of $\x$ then \cref{alg:fastikron} runs in $O(n \sum n_i)$ time.  

\end{restatable}

\cref{alg:fastikron} is correct even if the factors of $\AA$ are not square, although the time complexity is not as clean when written down.  

\begin{algorithm}
\caption{Kronecker Matrix-Vector Product} \label{alg:fastikron}
{\small
\begin{algorithmic}[1]
\Procedure{kmatvec}{$\A_1, \dots, \A_d, \x$}
\State $m_i, n_i = $ \Call{shape}{$\A_i$}
\State $r = \prod_{i=1}^d n_i$
\State $\vect{f}_{d+1} = \x$
\For{$i = d, \dots, 1$}
\State $\matr{Z} = $ \Call{reshape}{$\vect{f}_{i+1}$, $n_i$, $r/n_i$}
\State $r = r \cdot m_i / n_i$
\State $\vect{f}_i = $ \Call{reshape}{$\A_i \matr{Z}$, $r$, $1$}
\EndFor
\State \Return $\vect{f}_1$
\EndProcedure
\end{algorithmic}}
\end{algorithm}

Since all of the strategies found by our optimization routines are either Kronecker products or unions of Kronecker products, we can directly apply \cref{alg:fastikron} to efficiently implement the \measure step of HDMM.  Note that computing the matrix-vector product for a union of Kronecker products is a trivial extension of \cref{alg:fastikron}: it simply requires calling \cref{alg:fastikron} for each Kronecker product and concatenating the results into a single vector.  

We can also use \cref{alg:fastikron} to efficiently implement the \reconstruct step of HDMM.  The main challenge is to compute $\AA^+ \y$, or a pseudoinverse of $\AA$ together with a matrix-vector product.  This is done slightly differently for each type of strategy:

\begin{enumerate}
\item $\AA = \optk(\WW) = \A_1 \otimes \dots \times \A_d $.  From \cref{prop:kron} we know that $ \AA^+ = \A_1^+ \otimes \dots \otimes \A_d^+ $.  That is, the pseudoinverse of a Kronecker product is still a Kronecker product.  Thus, we can compute $\hat{\x} = \AA^+ \y$ efficiently using \cref{alg:fastikron}.  
\item $\AA = \optm(\WW) = \MM(\btheta)$.  From basic linear algebra, we know that $ \A^+ = (\A^T \A)^+ \A^T$ for any matrix $\A$.  Applied to this setting, we have $\MM^+(\btheta) = \GG^+(\btheta^2) \MM^T(\btheta)$, since we know $\MM^T \MM(\btheta) = \GG(\btheta^2)$ by \cref{prop:marggram}.  From \cref{thm:marginv} we know how to compute $\GG^+(\btheta^2)$ efficiently, and we know that it equals $\GG(\vect{\eta})$ for some $\vect{\eta}$.  We aim to compute $\MM^+(\btheta) \y = \GG^+(\btheta^2) \MM^T(\btheta) \y$.  We can easily compute $\v = \MM^T(\btheta) \y$ using a sequence of calls to \cref{alg:fastikron} by observing that $\MM^T(\btheta)$ is a just a bunch of Kronecker products \emph{horizontally} stacked together.  In a similar fashion, we can compute $\hat{\x} = \GG^+(\btheta^2) \v$ because $\GG^+(\btheta^2)$ is just the sum of a bunch of Kronecker products, which we can handle efficiently with repeated calls to \cref{alg:fastikron}. 
\item $\AA = \optkk(\WW) = c_1 \AA_1 + \dots + c_k \AA_k$.  Unfortunately, for a strategy of this form, we do not have a way to efficiently compute $\AA^+ \y$. While $\AA$ is a union of Kronecker products, the pseudoinverse is not necessarily, and we are not aware of a simple formula for the pseudoinverse of $\AA$ at all.  However, we can still produce an unbiased estimate of $\WW \x$ by using \emph{local least squares}.  To do this, we will compute $\WW_j \AA^+_j \y_j$ for each $j = 1, \dots, k$, where $\y_j$ is the answers produced for sub-strategy $\AA^+_j$.  Since $\WW_j$ and $\AA^+_j$ are both assumed to be Kronecker products, this can be easily achieved using \cref{alg:fastikron}.  This formula is an unbiased estimator for $\WW_j \x$, but the main drawback is that the workload query answers will not necessarily be consistent between sub-workloads.  
\end{enumerate}

\subsection{Improving scalability to large data vectors} \label{sec:sub:extension}

Thus far, HDMM has addressed the fundamental limitation of the matrix matrix mechanism --- replacing explicit matrix representations with implicit ones, and deriving efficient algorithms to solve the strategy optimization problem in the implicit space.  Our innovations allow HDMM to run in much higher-dimensional settings than the matrix mechanism, but HDMM still has trouble scaling to very high-dimensional settings, when the data vector no longer fits in memory.  Representing the data in vector form requires storing $ n = \prod n_i $ entries, which grows exponentially with the number of dimensions, and quickly becomes infeasible for truly high-dimensional data.  For example, a $30$-dimensional dataset with binary attributes ($n_i=2$) would require storing a data vector with $ 2^{30}$ entries, which is equivalent to approximately $ 4 \text{ GB}$ of space.   Scaling beyond this point would be quite challenging for HDMM.

It is important to note that the bottleneck of HDMM is \measure and \reconstruct, as these steps access and estimate the data vector.  In the matrix mechanism the main bottleneck is \select, as strategy optimization is the most expensive step.  HDMM can often still perform the \select step efficiently even when \measure and \reconstruct are intractable.  In some special-but-common cases, it may be possible for HDMM to bypass this bottleneck on \measure and \reconstruct, even scaling to settings where the data vector no longer fits in memory, making it suitable for arbitrarily large domains. 

The settings where HDMM can bypass this limitation depends crucially on the strategy, and consequently the workload as well.  If the workload is Identity over the whole domain (or any other full rank workload), then very little can be done because the vector of workload query answers (the output of HDMM) is just as large as the data vector itself, and simply enumerating those answers would require too much space.  Thus, the number of queries in the workload cannot be too large.  A special-but-common case occurs when the workload contains conjunctive queries over \emph{small subsets of attributes}.  The workload may cover all attributes of the dataset, but it will generally consist of a number of subworkloads, each which only cover a handful of attributes at a time.  With workloads of this form, strategies produced by HDMM ($\optkk$ and $\optm$ in particular\footnote{strategies produced by $\optk$ will generally be defined over the whole domain, rather than over a small subset of attributes}) will generally contain queries that are also defined over small subsets of attributes.  When this is the case, \measure can be done by keeping the data in its natural tabular format, and only vectorizing the data with respect to the relevant attributes for each sub-workload or sub-strategy.  Since these are assumed to be defined over small subsets of attributes, these smaller data vectors can easily be materialized explicitly and operated on accordingly.  Thus, the main remaining challenge is to \reconstruct the workload query answers while avoiding an explicit representation for $\hat{\x}$.  This can be done using a recently developed technique for efficient inference in differential privacy called ``\texttt{Private-PGM}'' \cite{mckenna2019graphical}.  \texttt{Private-PGM} consumes as input a set of noisy measurements defined over low-dimensional marginals, and produces a compact implicit representation of $\hat{\x}$.  It leverages \emph{probabilistic graphical models} to compactly represent $\hat{\x}$ in terms of a product of low-dimensional factors, and is able to scale to arbitrarily large domains as long as the measurements allow it.  

Using \texttt{Private-PGM} with HDMM does change the mechanism in some subtle but important ways.  The default HDMM method for \reconstruct is based on standard ordinary least squares, as it computes $ \hat{\x} = \AA^+ \y$, which is the solution to the minimization problem $\hat{\x}=\argmin_{\x} \norm{\AA \x - \y}_2^2$.   In contrast, \texttt{Private-PGM} is based on the related non-negative least squares problem: $\hat{\x} = \argmin_{\x > 0} \norm{\AA \x - \y}_2^2$.  We know that the true data vector is non-negative, so for this reason it seems like the \texttt{Private-PGM} approach is more natural.  However, non-negativity comes at the cost of \emph{bias}.  An appealing property of the ordinary least squares solution is that it produces an unbiased estimate of $\x$ under mild conditions.  Non-negative least squares does not share this same guarantee.  However, the introduction of bias often comes with reduced variance, and overall error is usually better when enforcing non-negativity \cite{li2015matrix,mckenna2019graphical}. 
Thus \texttt{Private-PGM} can be used not only to improve scalability of HDMM, but also utility.  In practical settings where some bias can be tolerated for reduced variance, we generally recommend incorporating \texttt{Private-PGM} post-processing into HDMM to improve utility, even when it is not necessary for scalability reasons.   


\section{Experimental Evaluation} \label{sec:experiments}

In this section we evaluate the accuracy and scalability of \sys.  We perform a comprehensive comparison of \sys with a variety of other mechanisms on low and high-dimensional workloads, showing that it consistently offers lower error than competitors and works in a broader range of settings than other algorithms.  We also evaluate the scalability of key components of HDMM, showing that it is capable of scaling effectively to high-dimensional settings.

In accuracy experiments, we report the Root Mean Squared Error (RMSE), which is defined as $ RMSE = \sqrt{\frac{1}{m}\Error(\W, \algG)} $ for an algorithm $\algG $.  We compute this value analytically using the formulas from Proposition \ref{prop:error} whenever possible.  We separately report results for pure differential privacy with Laplace noise and approximate differential privacy with Gaussian noise.  We use $\epsilon = 1.0$ and $\delta = 10^{-6}$ in all experiments, but note that the ratio of errors between two data-independent algorithms remains the same for all values of $\epsilon$ and $\delta$.

These experiments are meant to demonstrate that HDMM offers the best accuracy in the \emph{data-independent} regime.  It is possible that some data-dependent mechanisms will outperform even the best data-independent mechanism, and this will typically depend on the amount of data available and the privacy budget \cite{hay2016principled,vietri2020new}.  Data-independent mechanisms (like HDMM) are generally preferable when there is an abundance of data and/or the privacy budget is not too small, such as the U.S. Census decennial data release \cite{abowd2018us}.

\subsection{Evaluating $\optgp$ on Low Dimensional Workloads}

\begin{table}[] 
\resizebox{\textwidth}{!}{
\begin{tabular}{cc|ccccccc|c}
\multicolumn{10}{c}{$\boldsymbol{\epsilon}$-\textbf{differential privacy (Laplace noise)}}                                                                                                                                                       \\
\textbf{workload}                   & \textbf{domain} & \textbf{Identity} & \textbf{H2} & \textbf{Wavelet} & \textbf{HB} & \textbf{GreedyH} & \textbf{LRM} & \textbf{OPT\textsubscript{0}} & \textbf{SVDB} \\\hline
\multirow{4}{*}{\textbf{all-range}} & \textbf{64}     & 6.63              & 11.28       & 10.11            & 6.63        & 6.34             & 7.02         & \textbf{5.55}          & 3.22          \\
                                    & \textbf{256}    & 13.11             & 16.27       & 14.87            & 8.90        & 9.72             & 15.73        & \textbf{8.07}          & 4.07          \\
                                    & \textbf{1024}   & 26.15             & 21.83       & 20.26            & 12.82       & 14.70            & -             & \textbf{11.08}         & 4.94          \\
                                    & \textbf{4096}   & 52.27             & 27.90       & 26.18            & 16.19       & 22.21            & -             & \textbf{14.38}         & 5.82          \\\hline
\multirow{4}{*}{\textbf{prefix}}    & \textbf{64}     & 8.06              & 9.42        & 9.37             & 8.06        & 6.04             & 7.67         & \textbf{5.32}          & 2.89          \\
                                    & \textbf{256}    & 16.03             & 13.16       & 13.09            & 8.97        & 9.13             & 12.64        & \textbf{7.35}          & 3.50          \\
                                    & \textbf{1024}   & 32.02             & 17.29       & 17.20            & 12.87       & 14.32            & 15.43        & \textbf{9.58}          & 4.11          \\
                                    & \textbf{4096}   & 64.01             & 21.77       & 21.67            & 14.91       & 22.40            & -             & \textbf{12.20}         & 4.74          \\\hline
\multirow{4}{*}{\textbf{width32}}   & \textbf{64}     & 8.00              & 12.02       & 11.09            & 8.00        & 7.32             & 9.44         & \textbf{5.88}          & 2.75          \\
 & \textbf{256}  & 8.00 & 15.50 & 13.57 & 7.41 & 8.00 & 25.81 & \textbf{6.34} & 3.26 \\
                                    & \textbf{1024}   & 8.00              & 18.98       & 16.56            & 9.50        & 8.00             & 16.98        & \textbf{6.41}          & 3.36          \\
                                    & \textbf{4096}   & 8.00              & 22.45       & 19.58            & 10.96       & 8.00             & -            & \textbf{6.46}             & 3.38          \\\hline
\multirow{4}{*}{\textbf{permuted}}  & \textbf{64}     & 6.63              & 25.02       & 18.97            & 6.63        & 6.83             & 7.02         & \textbf{5.55}          & 3.22          \\
                                    & \textbf{256}    & 13.11             & 66.25       & 49.09            & 18.48       & 13.02            & 15.73        & \textbf{8.06}          & 4.07          \\
                                    & \textbf{1024}   & 26.15             & 157.50      & 117.06           & 37.07       & 23.94            & -             & \textbf{11.08}         & 4.94          \\
                                    & \textbf{4096}   & 52.27             & 374.29      & 277.42           & 107.83      & 45.77            & -             & \textbf{14.37}         & 5.82          \\\hline
\end{tabular}}
\caption{Error of strategies for 1D workloads with $\epsilon = 1.0$.} \label{table:opt01}
\end{table}

\begin{table}[]
\resizebox{\textwidth}{!}{
\begin{tabular}{cc|ccccccc|c}
\multicolumn{10}{c}{$(\boldsymbol{\epsilon}, \boldsymbol{\delta})$-\textbf{differential privacy (Gaussian noise)}}                                                                                                                                              \\
\textbf{workload}                   & \textbf{domain} & \textbf{Identity} & \textbf{H2} & \textbf{Wavelet} & \textbf{HB} & \textbf{GreedyH} & \textbf{COA} & \textbf{OPT\textsubscript{0}} & \textbf{SVDB} \\\hline
\multirow{4}{*}{\textbf{All Range}} & \textbf{64}     & 19.82             & 12.74       & 11.42            & 19.82       & 14.64            & \textbf{9.73}             & \textbf{9.73}          & 9.62          \\
                                    & \textbf{256}    & 39.18             & 16.20       & 14.81            & 18.80       & 23.34            & \textbf{12.26}             & \textbf{12.26}         & 12.15         \\
                                    & \textbf{1024}   & 78.13             & 19.66       & 18.24            & 27.07       & 36.20            & 14.89             & \textbf{14.85}         & 14.75         \\
                                    & \textbf{4096}   & 156.14            & 23.12       & 21.69            & 27.92       & 56.21            & 17.92             & \textbf{17.46}         & 17.38         \\\hline
\multirow{4}{*}{\textbf{prefix}}    & \textbf{64}     & 24.08             & 10.64       & 10.58            & 24.08       & 14.04            & \textbf{8.87}             & \textbf{8.87}          & 8.62          \\
                                    & \textbf{256}    & 47.89             & 13.11       & 13.03            & 18.95       & 22.11            & 10.70             & \textbf{10.66}         & 10.44         \\
                                    & \textbf{1024}   & 95.64             & 15.57       & 15.49            & 27.18       & 35.59            & 16.29             & \textbf{12.49}         & 12.29         \\
                                    & \textbf{4096}   & 191.21            & 18.03       & 17.95            & 25.72       & 56.70            & 26.50             & \textbf{14.32}         & 14.15         \\\hline
\multirow{4}{*}{\textbf{width32}}   & \textbf{64}     & 23.90             & 13.57       & 12.52            & 23.90       & 17.30            & 8.79             & \textbf{8.74}          & 8.23          \\
                                    & \textbf{256}    & 23.90             & 15.44       & 13.52            & 15.65       & 23.90            & 12.24             & \textbf{9.93}          & 9.73          \\
                                    & \textbf{1024}   & 23.90             & 17.10       & 14.92            & 20.08       & 23.90            & 16.00             & \textbf{10.08}         & 10.02         \\
                                    & \textbf{4096}   & 23.90             & 18.60       & 16.22            & 18.90       & 23.90            & 18.38             & \textbf{10.11}         & 10.09         \\\hline
\multirow{4}{*}{\textbf{permuted}}  & \textbf{64}     & 19.82             & 28.26       & 21.42            & 19.82       & 16.13            & \textbf{9.73}             & \textbf{9.73}          & 9.62          \\
                                    & \textbf{256}    & 39.18             & 65.97       & 48.88            & 39.04       & 35.22            & \textbf{12.26}             & \textbf{12.26}         & 12.15         \\
                                    & \textbf{1024}   & 78.13             & 141.86      & 105.44           & 78.30       & 60.60            & 14.89             & \textbf{14.85}         & 14.75         \\
                                    & \textbf{4096}   & 156.14            & 310.11      & 229.85           & 185.98      & 118.03           & 17.92             & \textbf{17.45}         & 17.38         \\\hline
\end{tabular}}
\caption{Error of strategies for 1D workloads with $\epsilon = 1.0$ and $\delta=10^{-6}$.} \label{table:opt02}
\end{table}

We begin by studying the effectiveness of $\optgp$ in the one-dimensional setting.  Specifically, we evaluate the quality of the strategies found by our optimization oracle compared with other \emph{data-independent} mechanisms designed for this setting.  It is important to understand the accuracy in the one-dimensional setting well, because $\optgp$ is used as a sub-routine for the higher-dimensional optimization operators $\optk$ and $\optkk$.

\paragraph*{\textbf{Workloads}}

We consider four different workloads: \textbf{All Range}, \textbf{Prefix}, \textbf{Width 32 Range}, and \textbf{Permuted Range}, each defined over domain sizes ranging from $64$ to $4096$.  
\textbf{All Range} contains every possible range query over the specified domain; \textbf{Prefix} contains range queries defining an empirical CDF; \textbf{Width 32 Range} contains all range queries of width 32.  
While the first three workloads are subsets of range queries, the last workload, \textbf{Permuted Range}, is the result of right-multiplying the workload of all range queries by a random permutation matrix.  Many proposed strategies have targeted workloads of range queries and tend to work fairly well on subsets of range queries.  \textbf{Permuted Range} poses a challenge because the structure of the workload is hidden by the permutation, requiring a truly adaptive method to find a good strategy.


Note the large size of some of these workloads: \textbf{All Range} and \textbf{Permuted Range} have $\frac{n(n+1)}{2} $ queries. For large $n$ it is infeasible to write down $\W$ in matrix form, but we can still compute the expected error since it only depends on the workload through its Gram matrix, $\W^T \W $, which is $n \times n$ and has special structure, allowing it to be computed directly without materializing $\W$.

\paragraph*{\textbf{Mechanisms}}

We consider 8 competing mechanisms: \textbf{Identity}, \textbf{Laplace}, \textbf{Gaussian}, \textbf{LRM} \cite{yuan2012low}, \textbf{COA} \cite{yuan2016convex}, \textbf{H2} \cite{hay2010boosting}, \textbf{HB} \cite{qardaji2013understanding}, \textbf{Privelet} \cite{xiao2011differential}, and \textbf{GreedyH} \cite{li2014data}.  The first five mechanisms are general purpose mechanisms, designed to support virtually any workload.  The last four mechanisms were specifically designed to offer low error on range query workloads.
We also report \textbf{SVDB} to understand the gap between the error of the computed strategies and the best lower bound on error we have (via the SVD bound).  

\paragraph*{\textbf{Results and Findings}}

\cref{table:opt01} and \cref{table:opt02} report the error of various mechanisms in each setting, for both Laplace and Gaussian noise respectively.  We remind the reader that these values do not depend on the true data $\x$, and thus they hold for all $\x$.  We report numbers for fixed $\epsilon = 1.0$ and $\delta = 10^{-6}$, but we note that these privacy parameters only impact the error by a constant factor, and hence the relationship between the errors of every pair of mechanisms remains the same for all $(\epsilon, \delta)$.  We have four main findings from these results, enumerated below:

\begin{enumerate}
\item $\optgp$ offers lower error than all competitors in all settings, and the magnitude of the improvement offered by HDMM (over the next best competitor) is as large as $3.18$ for Laplace noise (on Permuted Range) and $1.61$ for Gaussian noise (on Width 32 Range).  Interestingly, $\optgp$ offers lower error than \textbf{H2}, \textbf{HB}, \textbf{Privelet}, and \textbf{GreedyH} on range query workloads, even though these four mechanisms were designed specifically for range queries.  In additional, the second best method after $\optgp$ differs in each setting, which shows that some competing algorithms have specialized capabilities that allow them to perform well in some settings, while \sys performs well in a variety of settings as it does not make strict assumptions about the workload.
\item $\optgp$ gets within a factor of $2.57$ of the SVD bound for Laplace noise and $1.01$ of the SVD bound for Gaussian noise on every tested workload.  The gap between $\optgp$ and \textbf{SVDB} is quite small for Gaussian noise, suggesting that $\optgp$ is finding the best possible strategy.  Note that \textbf{COA} also finds a optimal strategy in many of the settings, but it fails on the Prefix and Width 32 Range workloads when $n \geq 1024 $.  Thus, even though it is solving the same problem underlying $\optgp$ in theory, the implementation is not as robust as ours.  The gap between $\optgp$ and \textbf{SVDB} is larger for Laplace noise, however, and it is unclear if this gap is primarily due to looseness of the SVD bound or suboptimality of the strategy.  Nevertheless, even with Laplace noise the ratio between $\optgp$ and \textbf{SVDB} is at most $2.57$.
\item The error of $\optgp$ (and \textbf{COA} for $(\epsilon, \delta)$ privacy) is the same on the All Range and Permuted Range workloads.  Permuting the workload doesn't impact achievable error or our optimization algorithm in any meaningful way.  However, many of the methods we compared against perform well on All Range but poorly on Permuted Range because they were specifically designed for range queries.  This shows that they exploit specific structure of the input workload and have limited adaptivity.
\item On these workloads, Laplace noise offers better error than Gaussian noise (for appropriately conservative settings of $\delta$).  This is because with Gaussian noise there is an additional $\approx \sqrt{\log(1/\delta)}$ term in the standard deviation of the noise, and this outweighs the benefit using the $L_2$ sensitivity norm instead of the $L_1$ sensitivity norm, despite the fact that we may be finding strategies that are further from optimal in the $L_1$ case.
\end{enumerate}


\paragraph*{\textbf{Scalability}}

\begin{figure} 
\includegraphics[width=0.5\textwidth]{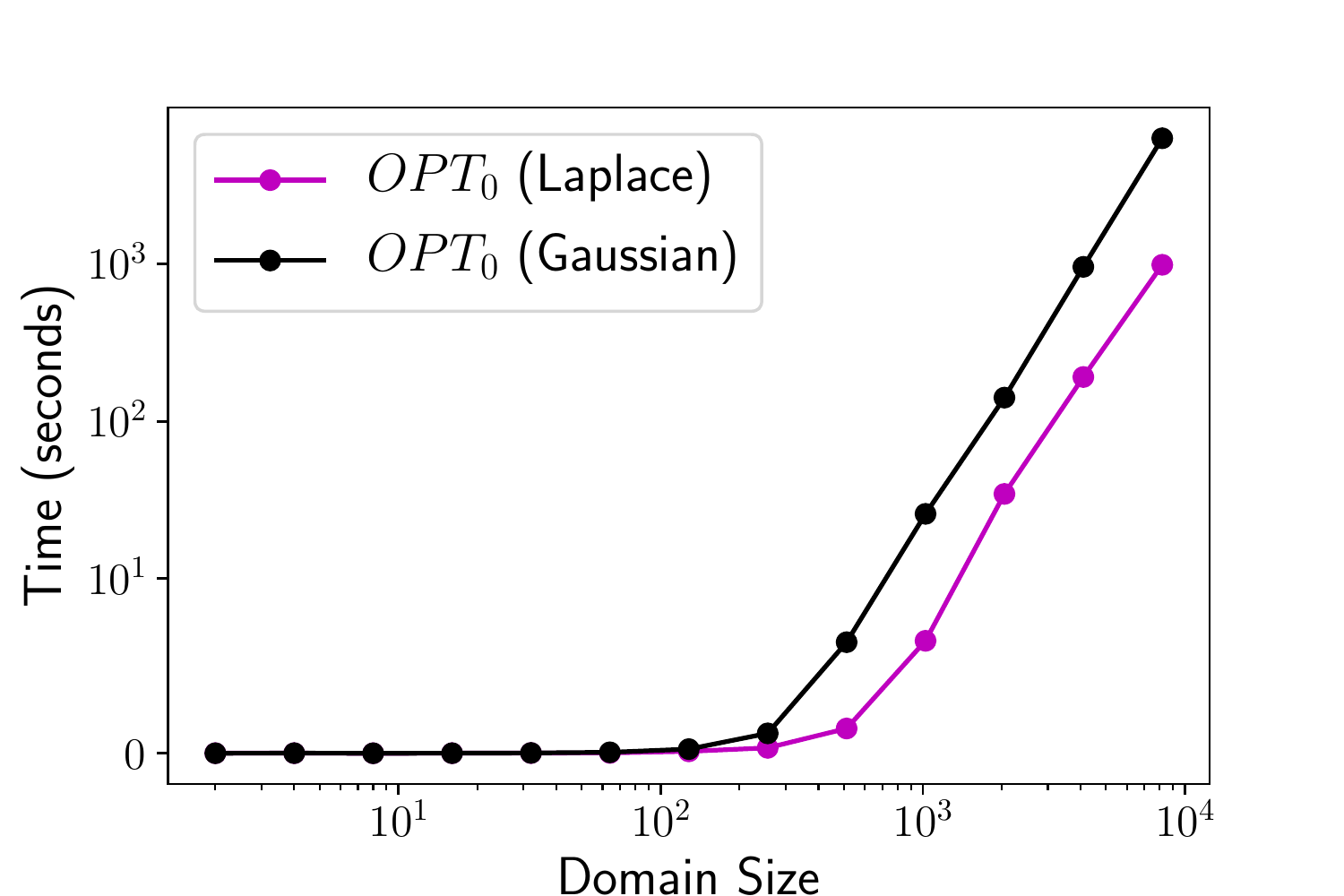}
\caption{ \label{fig:1d} Time required to run $\optgp$ for $100$ iterations on the AllRange workload for increasing domain sizes.}
\end{figure}

We now demonstrate the scalability of $\optgp$.  Note that optimization time dominates in the low-dimensional setting, and the time for \measure and \reconstruct is small in comparison to that.  The per-iteration time complexity only depends on the domain size, and not the contents of the workload.  While the number of iterations required for convergence may differ slightly based on the queries in the workload, for simplicity we measure the time required to run the optimization for 100 iterations on the All Range workload.

\cref{fig:1d} shows the amount of time required to run $\optgp$ for various domain sizes.  It shows that $\optgp$ scales up to $n=8192$, and runs for $n=1024$ in under 10 seconds for Laplace noise and 1 minute for Gaussian noise.  This difference occurs because the per-iteration time complexity is $O(p n^2)$ under Laplace noise but $O(n^3)$ under Gaussian noise.  For $n=8192$ it takes considerably longer, but is still feasible to run.  We remark that trading a few hours of computation time for a meaningful reduction in error is typically a welcome trade-off in practice, especially since workloads can be optimized once and the resulting strategies reused many times.  Additionally, we have a prototype implementation that uses GPUs and PyTorch, and we found that it is possible (although very time consuming) to scale up to $n=16384$.  Beyond this point, it quickly becomes infeasible to even represent the workload (or its Gram matrix) in matrix form, let alone optimize it.

\subsection{Evaluating $\optk$, $\optkk$, and $\optm$ on Multi-Dimensional Workloads}
\begin{table}[]
\resizebox{\textwidth}{!}{
\begin{tabular}{cc|ccccccc|c}
\multicolumn{10}{c}{$\boldsymbol{\epsilon}$-\textbf{differential privacy (Laplace noise)}} \\
\textbf{dataset} & \textbf{workload}                     & \textbf{Identity} & \textbf{Laplace} & \textbf{DataCube} & \textbf{OPT\textsubscript{$\boldsymbol{\otimes}$}} & \textbf{OPT\textsubscript{$\boldsymbol{+}$}} & \textbf{OPT\textsubscript{M}} & \textbf{HDMM} & \textbf{SVDB} \\\hline
\multirow{2}{*}{\textbf{Census (5D)}}  & \textbf{SF1}                          & 23.20             & 70.71            & -                 & 7.30            & 30.55           & 9.56            & 7.30          & -             \\
& \textbf{SF1+}                         & 32.50             & 141.42           & -                 & 10.23           & 42.31           & 12.88           & 10.23         & -             \\
\multirow{2}{*}{\textbf{CPS (5D)}}     & \textbf{All Marginals}                & 5.38              & 45.25            & 18.49             & 4.85            & 4.85            & 4.84            & 4.84          & 2.63          \\
& \textbf{All Prefix-Marginals}         & 98.06             & 56568.54         & -                 & 40.59           & 40.59           & 69.38           & 40.59         & 9.32          \\\hline
\multirow{2}{*}{\textbf{Adult (14D)}}   & \textbf{$\leq$ 3D Marginals} & 5352117.26        & 664.68           & 494.06            & 872.58          & 306.33          & 225.35          & 225.35        & 15.08         \\
& \textbf{2D Prefix-Marginals}          & 475602516.60      & 138602.83        & -                 & 1119.16         & 484.07          & 553.56          & 484.07        & -             \\
\multirow{2}{*}{\textbf{Loans (12D)}}   & \textbf{Small Marginals}              & 3330650.46        & 265.87           & 113.98            & 654.35          & 204.17          & 100.92          & 100.92        & 11.61         \\
& \textbf{Small Prefix-Marginals}       & 15340082.96       & 11013.90         & -                 & 1707.67         & 485.67          & 288.29          & 288.29        & -            
\end{tabular}}
\captionsetup{width=0.9\textwidth}
\caption{RMSE of HDMM strategies and baseline strategies on multi-dimensional workloads (ranging from 5D to 14D) for $\epsilon = 1.0$ with Laplace noise.} \label{table:hd1}
\end{table}

\begin{table}[]
\resizebox{\textwidth}{!}{
\begin{tabular}{cc|ccccccc|c}
\multicolumn{10}{c}{$(\boldsymbol{\epsilon}, \boldsymbol{\delta})$-\textbf{differential privacy (Gaussian noise)}}  \\
\textbf{dataset} & \textbf{workload}                     & \textbf{Identity} & \textbf{Laplace} & \textbf{DataCube} & \textbf{OPT\textsubscript{$\boldsymbol{\otimes}$}} & \textbf{OPT\textsubscript{$\boldsymbol{+}$}} & \textbf{OPT\textsubscript{M}} & \textbf{HDMM} & \textbf{SVDB} \\\hline
\multirow{2}{*}{\textbf{Census (5D)}}  & \textbf{SF1}                          & 69.33             & 29.87             & -                 & 9.80            & 75.66           & 14.31           & 9.80          & -             \\
& \textbf{SF1+}                         & 97.08             & 42.25             & -                 & 10.90           & 84.16           & 15.88           & 10.90         & -             \\
\multirow{2}{*}{\textbf{CPS (5D)}}     & \textbf{All Marginals}                & 16.08             & 23.90             & 19.53             & 7.85            & 7.85            & 7.86            & 7.85          & 7.85          \\
& \textbf{All Prefix-Marginals}         & 292.93            & 844.94            & -                 & 29.48           & 29.49           & 104.09          & 29.48         & 27.85         \\\hline
\multirow{2}{*}{\textbf{Adult (14D)}}   & \textbf{$\leq $ 3D Marginals} & 15988375.02       & 91.59             & 77.36             & 82.42           & 899.04          & 46.44           & 46.44         & 45.06         \\
& \textbf{2D Prefix-Marginals}          & 1420766966.19     & 1322.58           & -                 & 126.17          & 639.43          & 296.12          & 126.17        & -             \\
\multirow{2}{*}{\textbf{Loans (12D)}}   & \textbf{Small Marginals}              & 9949649.11        & 57.92             & 37.37             & 81.51           & 631.40          & 34.91           & 34.91         & 34.67         \\
& \textbf{Small Prefix-Marginals}       & 45825415.90       & 372.83            & -                 & 132.43          & 994.04          & 99.72           & 99.72         & -            
\end{tabular}}
\captionsetup{width=0.9\textwidth}
\caption{RMSE of HDMM strategies and baseline strategies on multi-dimensional workloads (ranging from 5D to 14D) for $\epsilon = 1.0$ and $\delta=10^{-6}$ with Gaussian noise.} \label{table:hd2}
\end{table}

We now shift our attention to the multi-dimensional setting.  

\paragraph*{\textbf{Workloads}}

We consider four multi-dimensional schemas and two workloads for each schema.  The first schema, Census of Population and Housing (\textbf{Census}), has been used as a running example throughout the paper.  The second schema, Current Population Survey (\textbf{CPS}), is another Census product.  These schemas have $5$ attributes each and domain sizes of about $1$ million.  The last two schemas, \textbf{Adult} and \textbf{Loans} are much higher-dimensional, having $15$ and $12$ attributes respectively.

For the \textbf{Census} schema, we use the SF1 and SF1+ workloads introduced in the paper.  For the other schemas, we use workloads based on Marginals and Prefix-Marginals, as defined in \cref{ex:marg}.  For \textbf{CPS} we use the workload of All Marginals and All Prefix-Marginals.  For Adult, we use All $0,1,2,$ and $3$-way marginals and all $2$-way Prefix-Marginals.  For Loans, we use All Small Marginals and All Small Prefix-Marginals.  A ``Small'' Marginal can be any $k$-way Marginal whose size is less than 5000.  This means the workload will be an interesting combination of $0, 1, 2, \dots, k$-way marginals.


We note that for the Adult and Loans schema, the domain is far too large to allow $\x$ to be represented in vector form.  
Thus, to actually run the mechanism on these high-dimensional datasets, we would have to use the extension discussed in \cref{sec:sub:extension}.
However, we remind the reader that in this section we are simply reporting \emph{expected errors}, which we can compute efficiently without ever materializing $\x$.

\paragraph*{\textbf{Mechanisms}}

In the high-dimensional setting, there are far fewer data-independent mechanisms to choose from.  We thus compare against Identity, Laplace, and Gaussian, which are the only methods from the previous section which are applicable and scalable to high-dimensional settings.  In addition to these simple baselines, we also compare against DataCube, which is applicable in this setting, but only for (unweighted) marginal query workloads. 

\paragraph*{\textbf{Results and Findings}}

\cref{table:hd1} and \cref{table:hd2} report the RMSE of the baselines as well as each optimization operator.  We compute the SVD bound when possible (i.e., the workload is either a single Kronecker product or a marginal query workload).  We have four main findings which we enumerate below:

\begin{enumerate}
\item HDMM is better than all competitors on all tasks, and the magnitude of the improvement is as large as $38$ for Laplace noise and $29$ for Gaussian noise.  
\item HDMM gets within a factor $1.06$ of the SVD bound when it is possible to compute it for Gaussian noise.  This is consistent with the theoretical result in \cref{sec:hd_opt} which justifies the defintion of $\optk$.  For Laplace noise, the ratio is as high as $14$, however.
\item Gaussian noise offers lower error than Laplace noise for the two highest dimensional schemas, and comparable error for the two five-dimensional schemas.  In contrast to the one-dimensional setting, this occurs because the savings from using the $L_2$ sensitivity norm outweighs the cost of $\approx \sqrt{\log(1 / \delta)}$ to use Gaussian noise with $(\epsilon, \delta)$-differential privacy. 
\item The parameterization that offers the lowest error differs based on the workload and the type of noise added.  For example, $\optm$ is always the best for workloads consisting of Marginals, but it is also sometimes the best for other workloads too.  $\optk$ is the best for the CPH and CPS workloads, but not as good for the Adult and Loans workloads.  $\optkk$ is best for the low-dimensional Prefix Marginals workloads.  
\end{enumerate}

\paragraph*{\textbf{Scalability}}

\begin{figure} 
\subcaptionbox{Time to run \select \\operators on $5$-dimensional \\ domains of size $(c, c, c, c, c)$.}{\includegraphics[width=0.325\textwidth]{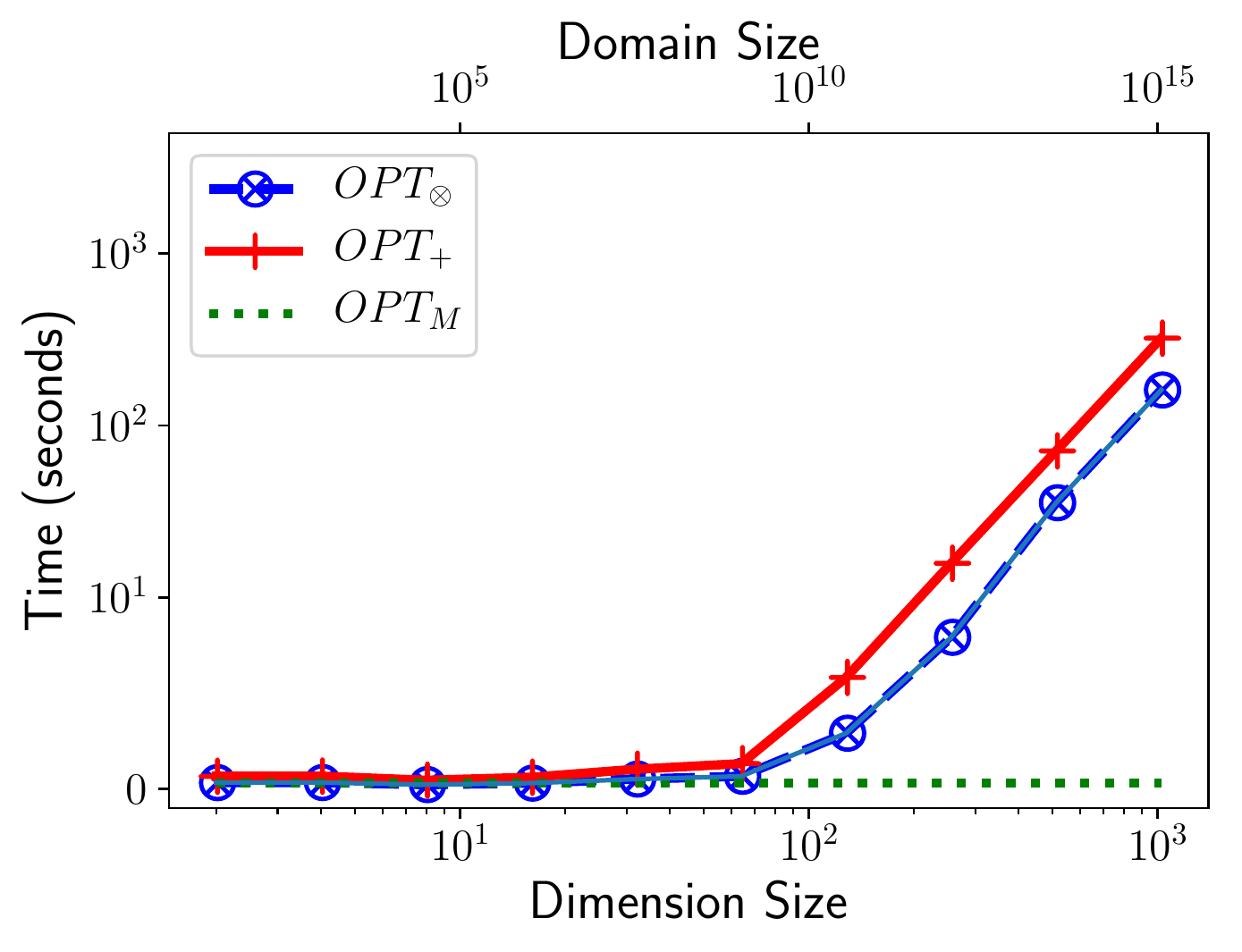}}
\subcaptionbox{Time to run \select \\ operators on $d$-dimensional \\ domains of size $(10, \dots, 10)$.}{\includegraphics[width=0.325\textwidth]{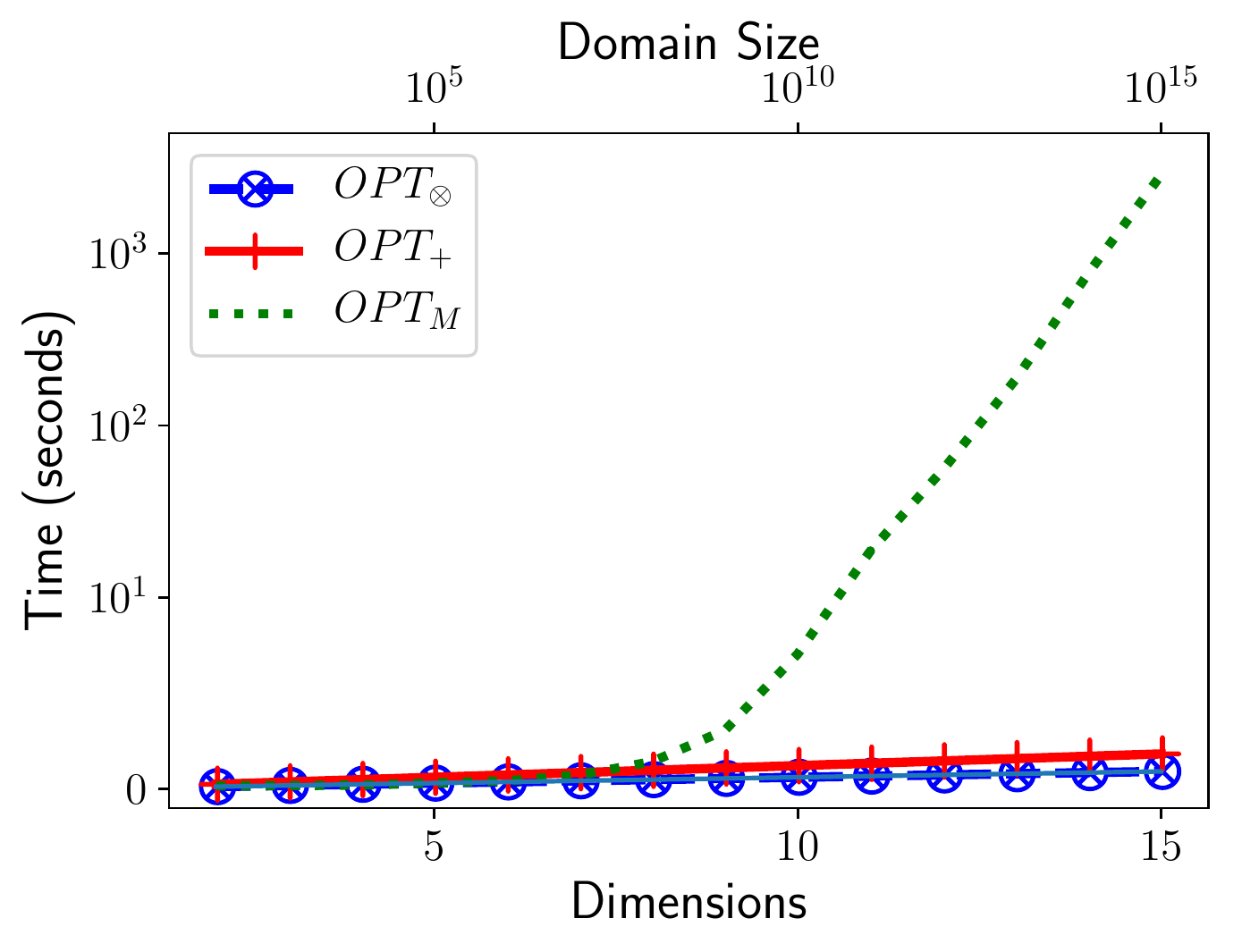}}
\subcaptionbox{Time to run \reconstruct for each type of strategy with varying domain sizes.}{\includegraphics[width=0.33\textwidth]{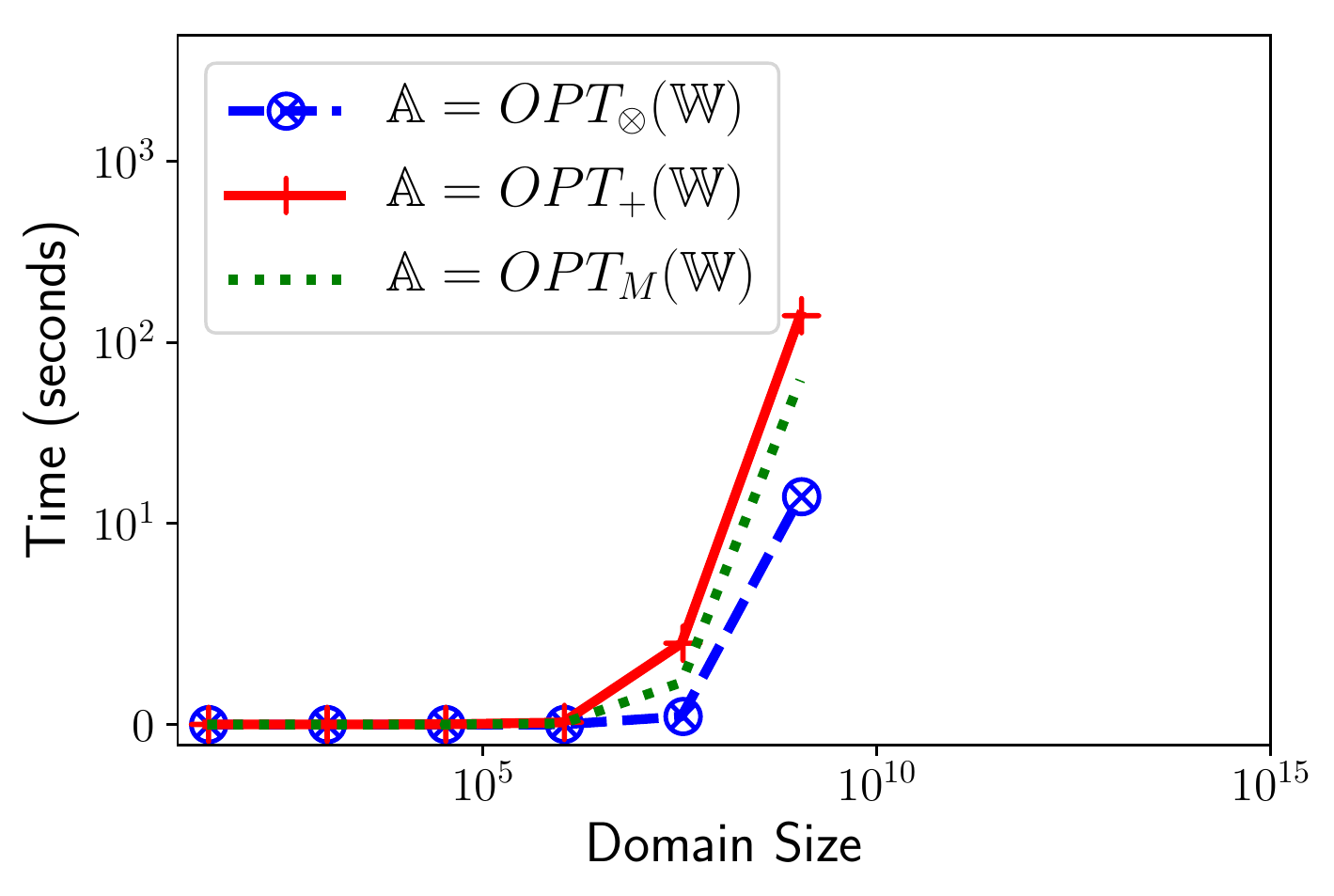}}
\caption{Scalability of different components of HDMM when run on multi-dimensional domains of varying size and shape.}\label{fig:hd}
\end{figure}

We now evaluate the scalability HDMM.  The main factor that influences the scalability of HDMM is the domain.  
The optimization time primarily depends on the number of dimensions and the size of each dimension, while reconstruction time primarily depends on the total domain size.  Thus, the bottleneck of HDMM depends on all of these factors in a nuanced way, and for some domains optimization will be the bottleneck, while for others reconstruction will be.  We show how the key components scale with respect to these properties of the domain in \cref{fig:hd}.

In \cref{fig:hd} (a), we fix the number of dimensions of the domain at $d=5$ and vary the size of each dimension from $n_i=2$ to $n_i=1024$.  We measure and report the optimization time for $\optk, \optkk,$ and $\optm$.  We run $\optk$ for 100 inner iterations (in calls to $\optgp$) and 5 outer iterations.  We use a workload consisting of a union of 10 Kronecker products, where each subworkload is All Range queries.
In \cref{fig:hd} (b), we fix the domain size of each dimension at $n_i = 10$ and vary the number of dimensions from $ d=2 $ to $d=15$.  We again use the same workload as before.
In \cref{fig:hd} (c), we use the strategies produced from \cref{fig:hd} (a), and measure the time required to perform the \reconstruct step of HDMM.

From the figure we can see that the optimization time of $\optk$ and $\optkk$ primarily depends on the size of each dimension, rather than the number of dimensions.  In contrast, the optimization time of $\optm$ primarily depends on the number of dimensions and not the size of each dimension.  This confirms the theoretical complexity results.  All three optimization operators are capable of running in settings where the total domain size is far too large to allow $\x$ to be represented in vector form.
The figure also shows that we can solve the \reconstruct step up to domains as large as $10^9$.  Beyond this point, it is infeasible to even represent $\x$ in vector form on the machine used for experiments.  Further scalability is possible by using the extension described in \cref{sec:running}.

\section{Related Work} \label{sec:related}

Much research has been done to develop differentially private algorithms for accurately answering linear queries \cite{zhang16privtree,yuan2016convex,li2015matrix,zhangtowards,xiao2014dpcube,qardaji2014priview,li2014data,Yaroslavtsev13Accurate,xu2013differential,qardaji2013understanding,qardaji2013differentially,yuan2012low,xu12histogram,li2012adaptive,cormode2012differentially,Acs2012compression,xiao2011differential,ding2011differentially,li2010optimizing,hay2010boosting,barak2007privacy,qardaji2014priview,Zhang2014}.  These algorithms are either data-dependent (such as DAWA~\cite{li2014data}) or data-independent (such as HB~\cite{qardaji2013understanding}).  Hay et al. found that in the high signal setting (number of records is large relative to $\epsilon$ and $n$), data-independent algorithms dominate, while in the low signal setting, data-dependent algorithms dominate~\cite{hay2016principled}. 
Most of the data-independent mechanisms belong to the \textbf{select-measure-reconstruct} paradigm, and much research has been done on the strategy selection problem for particular (usually fixed) workloads such as range queries or marginals.  Some research has been done on the strategy selection problem that is automatically tuned to a user-specified workload.  However, none of the existing approaches offer the scalability, generality, and utility of HDMM.  

\paragraph{\textbf{Mechanisms for Range Query Workloads}}

One notable class of workloads that has recieved considerable attention in the literature is range queries.  For these workloads, Xiao et al. propose a strategy based on Wavelet transforms~\cite{xiao2011differential}, Hay et al. propose a hierarchical strategy~\cite{hay2010boosting}, Cormode et al propose similar hierarchical strategies for multi-dimensional domaains~\cite{cormode2012differentially}, and Qardaji et al. generalize and improve the hierarchical approach~\cite{qardaji2013understanding}.  All of these strategies are designed for workloads of range queries, and they are not workload-adaptive.  The HB approach proposed by Qardaji et al. chooses a branching factor for a hierarchical strategy by optimizing an analytically computable approximation of $ \Error $.  Li et al. propose an approach called GreedyH~\cite{li2014data} that is workload adaptive, but based on a template strategy designed for range query-like workloads.  GreedyH can be seen as an instance of \sys, since it optimizes over strategies parameterized by a small set of weights, but the parameterization is only reasonable for 1D range query workloads, and is not expressive enough to capture the strategies produced by our $p$-Identity parameterization.

\paragraph{\textbf{Mechanisms for Marginal Query Workloads}}

Another notable class of workloads that has recieved special attention in the literature is marginal query workloads.  Barak et al. propose a strategy based on Fourier basis vectors~\cite{barak2007privacy} for answering low-dimensional marginals over a binary domain.  This method offers some workload adaptivity, in that the set of Fourier basis vectors in the strategy depends on the marginals in the workload.  

Ding et al. propose a strategy of marginals that adapts to the workload through a greedy heuristic~\cite{ding2011differentially} that approximately solves a combinatoric optimization problem.  This space of strategies considered by this approach is a subset of those representable by our marginals parameterization, where $ \vect{\theta} \in \set{0,1}^{2^d} $.  They give a method for efficiently doing least squares estimation for consistency, but unlike HDMM, their objective function doesn't account for this in the strategy selection phase. 

Qardaji et al. also propose a strategy of marginals that adapts to a workload through the solution to an optimization problem~\cite{qardaji2014priview}, but they don't require that the strategy support the workload.  They approximate the answers to unsupported queries by doing maximum entropy estimation.  Hence, this method is data-dependent.

\paragraph{\textbf{Workload-Adaptive Mechanisms}}

Some existing methods are truly workload-adaptive, such as the low rank mechanism~\cite{yuan2012low} and COA \cite{yuan2016convex}.  These methods both rely on the matrix reperesentation of the workload and hence suffer from the same scalability limitations of the matrix mechanism.  There are a few notable workload-adaptive mechanisms that do not rely on a matrix representation of the workload, however they do require some structural assumptions about the workload, just like HDMM assumes the workload contains conjunctive queries.  Some notable examples include MWEM \cite{hardt2010simple}, DualQuery \cite{gaboardi2014dual}, and FEM \cite{vietri2020new}.  These mechanisms are all data-dependent and can run for marginal query workloads.  


\section{Discussion and conclusions}  \label{sec:conclusion}

In this paper, we introduce HDMM, a general and scalable method for privately answering collections of counting queries over high-dimensional data.  \sys is capable of running on multi-dimensional datasets with very large domains.  This is primarily enabled by our implicit workload representation in terms of Kronecker products, and our optimization routines for strategy selection that exploit this implicit representation.  Because \sys provides state-of-the-art error rates in both low- and high-dimensions, and fully automated strategy selection, we believe it will be broadly useful to algorithm designers. 

In this paper, we extend \sys to handle Gaussian noise in addition to Laplace noise, showing that in several cases, strategy optimization is actually simpler and more effective.  We also study the SVD bound with implicitly represented workloads, and used it to reason about the effectiveness of our optimization operators theoretically and empirically.  
While \sys was previously limited to cases for which it is possible to materialize and manipulate the data vector, we show that we can sometimes bypass this limitation by integrating \texttt{Private-PGM} for the \reconstruct step.

\vspace{1ex}
{\footnotesize
\noindent\textbf{Acknowledgements:}  This work was supported by the National Science Foundation under grants 1253327, 1408982, 1409125, 1443014, 1421325, and 1409143; and by DARPA and SPAWAR under contract N66001-15-C-4067. The U.S. Government is authorized to reproduce and distribute reprints for Governmental purposes not withstanding any copyright notation thereon. The views, opinions, and/or findings expressed are those of the author(s) and should not be interpreted as representing the official views or policies of the Department of Defense or the U.S. Government. \par
}

\clearpage
\bibliographystyle{abbrv}
\bibliography{bib/refs}

\clearpage
\appendix
\section{Implicit Vectorization of Disjunctive Queries} \label{sec:disjuncts}

HDMM can also optimize workloads containing disjunctive queries with no modification to the underlying optimization operators being necessary.  The theorems below show how different logical operators on predicates impact the vector representation of the queries.  

\begin{theorem}
The vector representation of the negation $\neg \phi$ is $vec(\neg \phi)=\T - vec(\phi)$.
\end{theorem}

\begin{theorem}
The vector representation of the disjunction $\phi = \phi_A \vee \phi_B$ is 

$$ vec(\phi) = vec(\neg(\neg \phi_A \wedge \neg \phi_B)) = \T \otimes \T - (\T - vec(\phi_A)) \otimes (\T - vec(\phi_B)). $$
\end{theorem}

The theorem above uses DeMorgan's law to show that the disjunctive query can be converted into the negation of a conjunctive query.  We can also similarly define a Cartesian product of disjunctive queries, as follows:

$$ \WW = \mathbf{1} \otimes \mathbf{1} - \W_1 \otimes \W_2 $$

where $\W_1 = \mathbf{1} - vec(\Phi_A)$ and $\W_2 = \mathbf{1} - vec(\Phi_B)$ and $\mathbf{1}$ is a matrix of ones having the same shape as $\W_1$ and $\W_2$ respectively.  Thus, we can represent a Cartesian product of disjunctive queries as a \emph{difference} of two Kronecker products.  

Taking the gram matrix of $\WW$ we observe:

$$ \WW^T \WW = (\mathbf{1} \otimes \mathbf{1})^T (\mathbf{1} \otimes \mathbf{1}) - (\mathbf{1} \otimes \mathbf{1})^T (\W_1 \otimes \W_2) - (\W_1 \otimes \W_2)^T (\mathbf{1} \otimes \mathbf{1}) + (\W_1 \otimes \W_2)^T (\W_1 \otimes \W_2)$$

Each term of the above expression simplifies to a single Kronecker product, so $\WW^T \WW$ is actually a sum of four Kronecker products.  Note that the standard conjunctive query workloads, containing a union of Kronecker products, also have a Gram matrix that is a sum of Kronecker products.  Furthermore, the optimization operators only depend on $\WW$ through $\WW^T \WW$, and they expect the Gram matrix to be a sum of Kronecker products.  Thus, they can run without modification on workloads having the above disjunctive form.  Additionally, the workloads may contain arbitrary conbinations of conjunctive and disjunctive queries.

\section{Missing Proofs}

\thmkronnorm*

\begin{proof}
We prove these statements directly with algebraic manipulation:
\begin{align*}
\Lone{\A \otimes \B} &= \max_{t} \sum_{q} | \A(q_A, t_A) \B(q_B, t_B) | \\
&= \max_{t} \sum_{q} | \A(q_A, t_A) | | \B(q_B, t_B) | \\
&= \max_{t_A} \sum_{q_A} | \A(q_A, t_A) | \max_{t_B} \sum_{q_B} | \B(q_B, t_B) | \\
&= \Lone{\A} \Lone{\B}
\end{align*}
\begin{align*}
\Ltwo{\A \otimes \B}^2 &= \max_{t} \sum_{q} ( \A(q_A, t_A) \B(q_B, t_B) )^2 \\
&= \max_{t} \sum_{q} \A(q_A, t_A)^2 \B(q_B, t_B)^2 \\
&= \max_{t_A} \sum_{q_A} \A(q_A, t_A)^2 \max_{t_B} \sum_{q_B}  \B(q_B, t_B)^2 \\
&= \Ltwo{\A}^2 \Ltwo{\B}^2
\end{align*}

\begin{align*}
\norm{\A \otimes \B}_F^2 &= \sum_{q, t} ( \A(q_A, t_A) \B(q_B, t_B) )^2 \\
&= \sum_{q, t} \A(q_A, t_A)^2 \B(q_B, t_B)^2 \\
&= \sum_{q_A, t_A} \A(q_A, t_A)^2 \sum_{q_B, t_B} \B(q_B, t_B)^2 \\
&= \norm{\A}_F^2 \norm{\B}_F^2
\end{align*}
\end{proof}

\begin{restatable}[Complexity of $\optgp$]{theorem}{thmoptgp}\label{thm:optgp}
Given any $p$-Identity strategy $\A(\bTheta)$, both the objective function $C(\A(\matr{\Theta}))$
and the gradient $ \frac{\partial C}{\partial \A}$ can be evaluated in $ O(p n^2) $ time.
\end{restatable}

\begin{proof}
Assume $\W^T \W$ has been precomputed
and now express $ \A^T \A $ in terms of $\matr{\Theta}$ and $\D$:
$$ \A^T \A = \D^T \D + \D^T \matr{\Theta}^T \matr{\Theta} \D = \D[\I_n + \matr{\Theta}^T \matr{\Theta}] \D $$
Applying the identity $ (\X \matr{Y})^{-1} = \matr{Y}^{-1} \X^{-1} $ together with the Woodbury identity~\cite{hager1989updating} yields an expression for the inverse:
\begin{align*}
(\A^T \A)^{-1} &= \D^{-1} [\I_n + \matr{\Theta}^T \matr{\Theta}]^{-1} \D^{-1} \\
&= \D^{-1} [\I_n - \matr{\Theta}^T (\I_p + \matr{\Theta} \matr{\Theta}^T)^{-1} \matr{\Theta}] \D^{-1}
\end{align*}

We can compute $ (\A^T \A)^{-1} (\W^T \W) $ in $O(n^2 p)$ time by evaluating the following expression from right to left:
$$ (\A^T \A)^{-1} (\W^T \W) = \D^{-2} (\W^T \W) - \D^{-1} \matr{\Theta}^T (\I_p + \matr{\Theta} \matr{\Theta}^T)^{-1} \matr{\Theta} \D^{-1} (\W^T \W)  $$

By carefully looking at the dimensionality of the intermediate matrices that arise from carrying out the matrix multiplications from right-to-left, we see that the most expensive operation is the matrix-matrix product between an $ n \times p$ matrix and a $p \times n$ matrix, which takes $O(n^2 p)$ time.  The inverse $ (\I_p + \matr{\Theta} \matr{\Theta}^T)^{-1} $ takes $O(p^3)$ time and the operations involving $D$ take $O(n^2)$ time since it is a diagonal matrix.

The result still holds even if $\W^T \W$ is replaced with an arbitrary $n \times n$ matrix, so $ \X = (\A^T \A)^{-1} (\W^T \W) (\A^T \A)^{-1} $ can be computed in $O(n^2 p)$ time as well.  The gradient is $ -2 \A \X $ whose components can be calculated separately as $ -2 \D \X $ and $ -2 \bTheta \X $.  $ -2 \D \X $ takes $O(n^2)$ time and $ \matr{\Theta} \X$  takes $O(n^2 p)$ time, so the overall cost of computing the gradient is $O(n^2 p)$.
\end{proof}

\thmmargmult*

\begin{proof}
First observe how the matrices $\I$ and $\matr{1}$ interact under matrix multiplication:
\begin{align*}
\I \I = \I && \I \matr{1} = \matr{1} && \matr{1} \I = \matr{1} && \matr{1} \matr{1} = n_i \matr{1} \\
\end{align*}
Now consider the product $\CC(a) \CC(b)$ which is simplified using Kronecker product identities, logical rules, and bitwise manipulation.
\begin{align*}
= &\bigotimes_{i=1}^d [\matr{1} (a_i = 0) + \I (a_i = 1)] [\matr{1} (b_i = 0) + \I (b_i = 1)]  \\
= &\prod_{i=1}^d [n_i (a_i = 0 \text{ and } b_i = 0) + 1 (a_i = 1 \text{ or } b_i = 1)]
  \bigotimes_{i=1}^d [\matr{1} (a_i = 0 \text{ or } b_i = 0) + \I (a_i = 1 \text{ and } b_i = 1)] \\
= &\prod_{i=1}^d [ n_i ( (a | b)_i = 0) + 1 ( (a | b)_i = 1) ]
  \bigotimes_{i=1}^d [\matr{1} ( (a \& b)_i = 0) + \I ( (a \& b)_i = 1)] \\
= &\c(a|b) \CC(a \& b)
\end{align*}

Now let $\u, \v \in \mathbb{R}^{2^d}$ and consider the following product:
\begin{align*}
\GG(\u) \GG(\v) &= \Big( \sum_a \u(a) \CC(a) \Big) \Big( \sum_b \v(b) \CC(b) \Big) \\
&= \sum_{a,b} \u(a) \v(b) \CC(a) \CC(b) \\
&= \sum_{a,b} \u(a) \v(b) \c(a | b) \CC(a \& b) \\
\end{align*}
Observe that $ \GG(\u) \GG(\v) = \GG(\w) $ where
$$ \w(k) = \sum_{a \& b = k} \u(a) \v(b) \c(a | b) $$
The relationship between $\w$ and $\v$ is clearly linear, and by carefully inspecting the expression one can see that $ \w = \X(\u) \v $ where $ \X(\u)(k,b) = \sum_{a : a \& b = k} \u(a) \c(a | b) $.  $\X(\u)$ is an upper triangular matrix because $k = a\&b$, and $ a \& b \leq b$ for all $a$.
\end{proof}

\thmmarginv*

\begin{proof}
First note that $ \GG(\z) = \mathbb{I} $ (the identity matrix).  By \cref{thm:margmult},

\begin{align*}
\GG(\u) \GG^{-1}(\u) &= \GG(\u) \GG(\X^{-1}(\u) \z) \\
&= \GG(\X(\u) \X^{-1}(\u) \z) \\
&= \GG(\I \z) = \GG(\z) = \mathbb{I} \\
\end{align*}

This proves the first part of the theorem statement.  For the second part, note that if $ \GG(\u) \GG(\v) \GG(\u) = \GG(\u) $, then $\GG(\v)$ is a generalized inverse of $\GG(\u)$.  Using $ \v = \X^g(\u) \X^g(\u) \u $, we have

\begin{align*}
\GG(\u) \GG(\v) \GG(\u) &= \GG(\u) \GG(\u) \GG(\v) \\
&= \GG(\u) \GG(\u) \GG(\X^g(\u) \X^g(\u) \u) \\
&= \GG(\X(\u) \X(\u) \X^g(\u) \X^g(\u) \u) \\
&= \GG(\X(\u) \X^g(\u) \X(\u) \X^g(\u) \u) \\
&= \GG(\I \X(\u) \X^g(\u) \u) \\
&= \GG(\I \u) = \GG(\u) \\
\end{align*}

Thus, $\GG(\v)$ is a generalized inverse as desired.  This completes the proof.
\end{proof}

\thmmargeig*

\begin{proof}

Recall that $ \CC(b) = \bigotimes_{i=1}^d [\matr{1} (b_i = 0) + \I (b_i = 1)] $ and $\c(k) = \prod_{i = 1}^d [n_i (k_i=0) + 1 (k_i=1)] $.
The proof follows from direct calculation:

\begin{align*}
\CC(b) \VV(a) &= \bigotimes_{i=1}^d [(b_i = 0) \matr{1} + (b_i = 1) \I] \bigotimes_{i=1}^d [(a_i = 0) \T + (a_i = 1) (\matr{1} - n_i \I)] \\
&= \bigotimes_{i=1}^d [(b_i = 0) \matr{1} + (b_i = 1) \I] [(a_i = 0) \T + (a_i = 1) (\matr{1} - n_i \I)] \\
&= \bigotimes_{i=1}^d [(a_i = 0 \text{ and } b_i = 0) n_i \T + (a_i=0 \text{ and } b_i = 1) \T \\
&+ (a_i=1 \text{ and } b_i=0) \matr{0} + (a_i=1 \text{ and } b_i=1) (\matr{1} - n_i \I)] \\
&= \begin{cases}
\prod_{i=1}^d n_i (b_i=0) + 1 (b_i=1) \bigotimes_{i=1}^d [(a_i=0) \T + (a_i=1) (\matr{1}-n_i \I)] & a\&b=a\\
\mathbbl{0} & otherwise
\end{cases} \\
&= \begin{cases}
\c(b) \VV(a)  & a\&b=a\\
0 \VV(a) & otherwise
\end{cases} \\
&= \blambda(a) \VV(a)
\end{align*}

This completes the first part of the proof.  For the second part, we have:

\begin{align*}
\GG(\w) \VV(a) &= \sum_b \w(b) \CC(b) \VV(a) \\
&= \sum_b \w(b) \blambda(a) \VV(a) \\
&= \sum_{b : a \& b = a} \w(b) C(b) \VV(a) \\
&= \bkappa(a) \VV(a) \\
\end{align*}

\end{proof}

\thmmargapprox*

\begin{proof}
Let $\VV = \WW^T \WW$ be the Gram matrix of $\WW$:
$ \VV = \sum_{j=1}^k w_j^2 \bigotimes_{i=1}^d \V^{(j)}_i $
 where $ \V^{(j)}_i = (\W^T \W)^{(j)}_i $.  Now consider the following quantity:

\begin{align*}
tr[\GG(\u) \VV] &= tr\Big[\Big(\sum_{a = 0}^{2^d-1} \u(a) \bigotimes_{i=1}^d [\matr{1} (a_i = 0) + \I (a_i = 1)] \Big) \Big( \sum_{j=1}^k w_j^2 \bigotimes_{i=1}^d \V^{(j)}_i \Big) \Big] \\
 &= tr\Big[\sum_{a = 0}^{2^d-1} \u(a) \sum_{j=1}^k w_j^2 \bigotimes_{i=1}^d [\matr{1} (a_i = 0) + \I (a_i = 1)] \V^{(j)}_i \Big] \\
 &= \sum_{a = 0}^{2^d-1} \u(a) \sum_{j=1}^k w_j^2 \prod_{i=1}^d tr[\matr{1} \V^{(j)}_i] (a_i = 0) + tr[\I \V^{(j)}_i] (a_i = 1) \\
 &= \sum_{a = 0}^{2^d-1} \u(a) \sum_{j=1}^k w_j^2 \prod_{i=1}^d sum[\V^{(j)}_i] (a_i = 0) + tr[\V^{(j)}_i] (a_i = 1) \\
\end{align*}

And observe that it only depends on $\V^{(j)}_i$ through its $sum$ and $trace$.  Thus, we could replace $\V^{(j)}_i$ with any matrix that has the same $sum$ and $trace$.  In particular, we could use $ \hat{\V}^{(j)}_i = b \I + c \matr{1} $, where $b$ and $c$ are chosen to satisfy the following linear system:

$$ \begin{bmatrix} n_i & n_i \\ n_i & n_i^2 \\ \end{bmatrix} 
\begin{bmatrix} b \\ c \end{bmatrix}
=
\begin{bmatrix} tr[\V^{(j)}_i] \\ sum[\V^{(j)}_i] \end{bmatrix}
$$

The matrix $\hat{\VV}_j = w_j^2 (\hat{\V}^{(j)}_1 \otimes \dots \otimes \hat{\V}^{(j)}_d) $ is nothing more than the Gram matrix for a collection of weighted marginals, or $\GG(\w_j)$.  This is because each factor in the Kronecker product is a weighted sum of $\I$ and $\matr{1}$, and by using the distributive property it can be converted into the canonical representation.

Thus, the matrix $\sum_j \GG(\w_j) = \GG(\sum_j \w_j) = \GG(\w) $ satisfies $tr[\GG(\u) \VV] = tr[\GG(\u) \GG(\w)]$ as desired.
\end{proof}

\thmmargobj*

\begin{proof}
\begin{align*}
\norm{\MM(\btheta)}_{\algG}^2 \norm{\WW \MM(\btheta)^+}_F^2
&= \norm{\btheta}^2 \norm{\WW \MM(\btheta)^+}_F^2 && \text{by \cref{prop:marginal_sensitivity}} \\
&= \norm{\btheta}^2 tr[\GG^+(\btheta^2) \WW^T \WW] && \text{} \\
&= \norm{\btheta}^2 tr[\GG^+(\btheta^2) \GG(\w)] && \text{by \cref{thm:margapprox}} \\
&= \norm{\btheta}^2 tr[\GG(\X^+(\btheta^2) \X^+(\btheta^2) \btheta^2) \GG(\w)] && \text{by \cref{thm:marginv}} \\
&= \norm{\btheta}^2 tr[\GG(\X(\w) \X^+(\btheta^2) \X^+(\btheta^2) \btheta^2)] && \text{by \cref{thm:margmult}} \\
&= \norm{\btheta}^2 tr[\GG(\X^+(\btheta^2) \X^+(\btheta^2) \X(\btheta^2) \w)] && \text{by commutativity} \\
&= \norm{\btheta}^2 tr[\GG(\X^+(\btheta^2) \w)] && \text{by constraint} \\
&= \norm{\btheta}^2 [\vect{1}^T \X^+(\btheta^2) \w] && \text{} \\
\end{align*}
\end{proof}

\thmmargsvdb*

\begin{proof}
From \cref{thm:eigmarg} we know all $2^d$ unique eigenvalues and corresponding eigenmatrices.  The number of rows in each eigenmatrix corresponds to the number of eigenvectors with that eigenvalue.  To compute the SVD bound, we need to take the square root of each unique eigenvalue (which is a singular value of $\WW$) and multiply that by it's multiplicity, then sum across all unique eigenvalues.  Note that the eigenmatrix $\VV(a)$ has $\c(\neg a)$ rows.  Hence, the SVD  bound is:

\begin{align*}
SVDB(\WW) &= \frac{1}{n} \Big(\sum_a \c(\neg a) \sqrt{\bkappa(a)} \Big)^2 \\
&= \frac{1}{n} \Big( \sum_a \c(\neg a) \sqrt{\sum_{b : a \& b = a} \w(b) \c(b)} \Big)^2
\end{align*}
\end{proof}

\thmoptsvdbmarg*

\begin{proof}
We will prove optimality by showing that $\AA = \MM(\btheta)$ matches the SVD bound.  Li et al. showed that the SVD bound is satisfied with equality if $\AA$ and $\WW$ share the same singular vectors and the singular values of $\AA$ are the square root of the singular values of $\WW$, at least in the case of Gaussian noise.  Recall from \cref{thm:eigmarg} we know that all marginal Gram matrices share the same eigenvectors.  
The unique eigenvalues of $\GG(\w)$  are $ \bkappa  = \Y \w$.  The gram matrix of $\AA=\MM(\btheta)$ is $\AA^T \AA = \GG(\btheta^2)$.  The eigenvalues of this are $\Y \btheta^2 = \Y (\Y^{-1} \sqrt{\Y \w}) = \sqrt{\Y \w}$.  Thus, the eigenvalues are exactly the square root of the eigenvalues of $\GG(\w)$, as desired.  This certifies that $\AA = \MM(\btheta)$ matches the SVD bound and is optimal.
\end{proof}

\thmmatvec*

\begin{proof}
Let $\y = \AA \x $. Then

\begin{align*}
\y(q) &= \sum_{t} \AA(q, t) \x(t) \\
&= \sum_{t} \A_1(q_1, t_1) \dots \A_d(q_d, t_d) \x(t) \\
&= \sum_{t_1} \A_1(q_1, t_1) \dots \sum_{t_d} \A_d(q_d, t_d) \x(t_1, \dots, t_d) \\
\end{align*}

Now define $\vect{f}_k$ to be the vector indexed by tuples $(t_1, \dots, t_{k-1}, q_k, \dots, q_d)$ such that $\vect{f}_{d+1} = \x$ and:

$$\vect{f}_k(t_{1:k-1}, q_{k:d}) = \sum_{t_k} \A_k(q_k, t_k) \vect{f}_{k+1}(t_{1:k}, q_{k+1:d}) $$

and observe that $\y = \vect{f}_1$.  We can efficiently compute $\vect{f}_{k}$ from $\vect{f}_{k+1}$ by observing that it is essentially computing a matrix-matrix product between the $ n_k \times n_k $ matrix $\A_k$ and the $ n_k \times n / n_k $ matrix obtained by reorganizing the entries of $\vect{f}_{k+1}$ into a matrix where rows are indexed by $t_k$.  This can be computed in $O(n n_k)$ time.  Thus, the total time required to compute $\y$ is $O(n \sum n_i)$ as stated.
\end{proof}



\end{document}